\newif{\ifarxiv}
\arxivtrue

\ifarxiv
\documentclass{article}
\else
\documentclass{mscs}
\journal{\relax} 
\fi

\usepackage{apalike}
\usepackage{amsmath,amscd,amssymb,amsfonts,latexsym}
\usepackage{stmaryrd}
\usepackage[dvipsnames]{xcolor}
\newcommand{\updated}{\color{Brown}}

\usepackage{url}
\usepackage{graphicx}
\usepackage[all]{xy}

\newtheorem{definition}{Definition}[section]
\newtheorem{theorem}[definition]{Theorem}
\newtheorem{lemma}[definition]{Lemma}
\newtheorem{corollary}[definition]{Corollary}
\newtheorem{proposition}[definition]{Proposition}

\ifarxiv%
\newenvironment{proof}{{\it Proof. }}{\hfill$\Box$\par}
\fi

\newtheorem{newdefinition}{Definition}[section]

\newtheorem{newremark}[newdefinition]{Remark}
\newtheorem{newlemma}[newdefinition]{Lemma}

\newcommand\Topcat{\mathbf{Top}}

\newcommand\dG{{\mathsf{d}}}
\newcommand\patch{{\mathsf{patch}}}

\newcommand\pt{{\mathsf{pt}}}

\newcommand{\real}{\mathbb{R}}
\newcommand{\creal}{\overline{\real}_+}
\newcommand\Rp{\real_+}
\newcommand\nat{\mathbb{N}}

\newcommand\Open{\mathcal O}
\newcommand\upc{\mathop{\uparrow}\nolimits}
\newcommand\dc{\mathop{\downarrow}\nolimits}
\newcommand\uuarrow{\rlap{$\uparrow$}\raise.5ex\hbox{$\uparrow$}}
\newcommand\ddarrow{\rlap{$\downarrow$}\raise.5ex\hbox{$\downarrow$}}
\newcommand\diff{\smallsetminus}

\newcommand\Val{\mathbf V}
\newcommand\Pred{\mathbb{P}}
\newcommand\Angel{{\mathtt{A}}}
\newcommand\Demon{{\mathtt{D}}}
\newcommand\Nature{{\mathtt{P}}}
\newcommand\AN{{\Angel\Nature}}
\newcommand\DN{{\Demon\Nature}}
\newcommand\ADN{{\Angel\Demon\Nature}}

\newcommand\Smyth{\mathcal Q}
\newcommand\Hoare{\mathcal H}
\newcommand\HV{\mathcal H_{\mathcal V}}
\newcommand\SV{\Smyth_{\mathcal V}}
\newcommand\Plotkin{\mathcal P\ell}
\newcommand\Plotkinn{\Plotkin_{\mathcal V}}
\newcommand\PV\Plotkinn 
\newcommand\Sober{{\mathcal S}}

\newcommand\identity[1]{\mathrm{id}_{#1}}

\newcommand\eqdef{\mathrel{\buildrel \text{def}\over=}}
\newcommand\Lform{\mathcal L}
\newcommand\one{\mathbf 1}
\newcommand{\interior}[1]{int ({#1})} 

\ifarxiv%
\title{Isomorphism Theorems between Models of Mixed
  Choice {\updated (Revised)}}
\author{Jean Goubault-Larrecq\\
  Universit\'e Paris-Saclay, CNRS, ENS Paris-Saclay,\\
  Laboratoire M\'ethodes
  Formelles, 91190, Gif-sur-Yvette, France.\\
  \texttt{jgl@lmf.cnrs.fr}
}
\else%
\title[Isomorphisms]{Isomorphism Theorems between Models of Mixed
  Choice {\updated (Revised)}}
\author[J. Goubault-Larrecq]
{{J\ls E\ls A\ls N\ns G\ls O\ls U\ls B\ls A\ls U\ls
  L\ls T\ls -\ls L\ls A\ls R\ls R\ls E\ls C\ls Q$^1$%
  \thanks{Work partially supported by the INRIA ARC ProNoBiS and the
    ANR programme blanc project CPP.}}
\\
$^1$
\begin{tabular}[t]{l}
  Universit\'e Paris-Saclay, CNRS, ENS Paris-Saclay,\\
  Laboratoire M\'ethodes Formelles \\
  91190, Gif-sur-Yvette, France.
\end{tabular}
  \qquad \texttt{jgl@lmf.cnrs.fr}
}

\date{: --}
\fi%

\begin{document}

\maketitle

\begin{abstract}
  We relate the so-called powercone models of mixed non-deterministic
  and probabilistic choice proposed by Tix, Keimel, Plotkin, Mislove,
  Ouaknine, Worrell, Morgan, and McIver, to our own models of
  previsions.  Under suitable topological assumptions, we show that
  they are isomorphic.  We rely on Keimel's cone-theoretic variants of
  the classical Hahn-Banach separation theorems, using functional
  analytic methods, and on the Schr\"oder-Simpson Theorem.  {\updated
    Lemma~3.4 in the original 2017 version, published at MSCS, had a
    wrong proof, and we prove a repaired, albeit slightly less general
    version here.}

    
\end{abstract}

\section{Introduction}
\label{sec:introduction}

Consider the question of giving semantics to a programming language
with \emph{mixed choice}, i.e., with two interacting forms of choice,
non-deterministic and probabilistic.  For example, probabilistic
choice can be the result of random coin flips, while non-deterministic
choice can be the result of interaction with an environment, which
decides which option to take in our place.

We shall be especially interested in domain-theoretic models, where
the space of all choices over a directed complete poset (dcpo) of
values is itself a dcpo.

Models of only one form of choice have been known for a long time.
For non-deterministic choice, the best known models are the Hoare
powerdomain of angelic non-determinism, the Smyth powerdomain of
demonic non-determinism, and the Plotkin powerdomain of erratic
non-determinism, see \cite[Section~6.2]{AJ:domains} or
\cite[Section~IV.8]{GHKLMS:contlatt}: there, choices are represented
as certain sets of values, from which the environment may choose.  For
probabilistic choice, a natural model is the Jones-Plotkin model of
continuous valuations \cite{Jones:proba}, where choices are
represented as objects akin to (sub)probability distributions over the
set of possible values, called continuous valuations.

Several denotational models of mixed choice were proposed in the past.
Let us list some of them:
\begin{enumerate}
\item In the \emph{powercone models}, choices are modeled as certain
  sets $E$ of continuous valuations.  These sets are, in particular,
  convex, meaning that $a \nu + (1-a) \nu'$ is in $E$ for all $\nu,
  \nu' \in E$ and $a \in [0, 1]$.  Choice proceeds by picking a
  valuation $\nu$ from the set $E$, then drawing a value $v$ at random
  according to $\nu$.  The probability that $v$ \emph{may} fall in
  some set $U$ is $\sup_{\nu \in E} \nu (U)$, while the probability
  that it \emph{must} fall in $U$ is $\inf_{\nu \in E} \nu (U)$.  Such
  models were proposed and studied by Mislove
  \cite{Mislove:nondet:prob}, by Tix \emph{et al.}
  \cite{Tix:PhD,TKP:nondet:prob}, and by Morgan and McIver
  \cite{DBLP:journals/acta/McIverM01}.
\item In the \emph{prevision models}, choices are modeled as certain
  second-order functionals $F$, called previsions.  They represent
  directly the probability that a value $v$ may fall in $U$, or must
  fall in $U$, depending on the kind of prevision, as $F (\chi_U)$,
  where $\chi_U$ is the characteristic function of $U$.  These were
  proposed by the author in \cite{Gou-csl07}, and also underlie the
  (generalized) predicate transformer semantics of
  \cite{KP:predtrans:pow} or of \cite{MMIS:prob:wp}.
\item The models of indexed valuations
  \cite{Varacca:PhD,Varacca:indexed}
  are alternate models for mere
  probabilistic choice.  They combine well with angelic
  non-determinism, yielding models akin to the powercone models,
  except that we do not require the sets of indexed
  valuations/continuous random variables to be convex.  Categorically,
  the combination is obtained via a distributivity law between monads.
\item The monad coproduct models of \cite{Luth:PhD}, once instantiated
  to the monads of non-blocking non-deterministic and probabilistic
  choice \cite{GU:idmon:coprod}, also yield a model of mixed choice,
  which the latter authors argue is close to Varacca's.
\end{enumerate}
Our goal is to relate the first two, and to show that, under mild
assumptions, they are isomorphic.  We have already announced this
result in \cite{Gou-fossacs08a}.  The proof was complex and was
limited to continuous dcpos.  Similar isomorphisms were proved, with
the general aim of giving generalized predicate transformer semantics
to powercone models, by Keimel and Plotkin \cite{KP:predtrans:pow},
again for continuous dcpos, and for unbounded continuous valuations,
not (sub)probability valuations\footnote{While we were preparing the
  final version of this document, Klaus Keimel informed me of a newer
  paper by the same authors \cite{KP:mixed}, which, among other
  things, also deals with probability and subprobability valuations,
  relying on a novel notion of Kegelspitze.}.  We generalize these
results to much larger classes of topological spaces, using more
streamlined and more general arguments than in \cite{Gou-fossacs08a}.

The basic idea is simple.  There is a map $r$ from the powercone
model, which maps any set $E$ of continuous valuations to $F = r (E)$,
defined by $F (h) \eqdef \sup_{\nu \in E} \int_x h (x) d\nu$ (in the
angelic case; replace $\sup$ by $\inf$ in the demonic case).  One
needs to show that it is continuous, and that it has a continuous
inverse.

\begin{figure}
  \centering
  \ifarxiv
  \begin{picture}(0,0)%
    \includegraphics{r.pdf}%
  \end{picture}%
  \setlength{\unitlength}{2072sp}%
  \begingroup\makeatletter\ifx\SetFigFont\undefined%
  \gdef\SetFigFont#1#2#3#4#5{%
    \reset@font\fontsize{#1}{#2pt}%
    \fontfamily{#3}\fontseries{#4}\fontshape{#5}%
    \selectfont}%
  \fi\endgroup%
  \begin{picture}(7258,3457)(663,-6373)
  \end{picture}%
  \else
  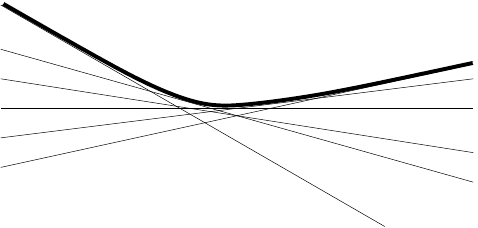 
  \fi
  \caption{Convex hulls}
  \label{fig:r}
\end{figure}

Our approach is typical of convex analysis.  By a form of the Riesz
representation theorem, which we shall make explicit below, there is
an isomorphism between continuous valuations $\nu$ and \emph{linear}
previsions, i.e., previsions $F$ such that
$F (a \cdot h+(1-a) \cdot h') = a F(h) + (1-a) F(h')$ for all
$a \in [0, 1]$: define $F (h)$ as $\int_x h (x) d\nu$.  It is useful
to imagine such linear previsions by drawing the curve $y = F (h)$,
where $h$ serves as $x$-coordinate, and convincing oneself that such
curves should be straight lines.  Modulo this isomorphism, $r$ maps a
set $E$ of linear previsions (the straight lines in
Figure~\ref{fig:r}) to its pointwise sup, shown as a fat curve.  This
fat curve is always convex (meaning that
$F (a \cdot h+(1-a) \cdot h') \leq a F(h) + (1-a) F(h')$ for all
$a \in [0, 1]$), and will be our (angelic) prevision $r (E)$.
Conversely, given any convex, fat curve $F$, the set of straight lines
below it form a convex set of linear previsions $s (F)$.

Showing that $s$ is a right inverse to $r$, i.e., that $r (s (F)) = F$
is essentially the Hahn-Banach theorem (or the corresponding variant
in our setting).  The difficult part will be to show that $s$ takes
its values in the right space, and that it is continuous.

\paragraph{\em Outline.}  We present all needed domain-theoretic,
topological, and cone-theoretic notions in
Section~\ref{sec:preliminaries}.  This is a fairly long section, and
our only excuse is that we have tried to make the paper as
self-contained as possible.  Section~\ref{sec:retr-powerc-onto}
exhibits our prevision models as retracts of spaces of certain sets
(closed, compact, or lenses) of certain continuous valuations
(unbounded, subprobability, probability).  The only difference that
these sets exhibit compared to powercones is that we are not requiring
them to be convex.  The existence of these retractions has a certain
number of interesting consequences, which we also list there.
Finally, using the Schr\"oder-Simpson Theorem, we establish the
desired isomorphisms in Section~\ref{sec:isom}.  We conclude in
Section~\ref{sec:conclusion}.

{\updated This paper is a revised version of \cite{JGL-mscs16}, fixing
  a mistake in Lemma~3.4 there and its consequences throughout the
  paper.  The changes will be shown in color, as with the current
  paragraph.  Some typos and minor errors have been corrected, too,
  most often without any particular highlighting.  In an additional
  Section~\ref{sec:review-papers-citing}, we investigate the
  consequences of the mistake on papers in the literature.  In short,
  no paper seems to be affected by the mistake.}

\section{Preliminaries}
\label{sec:preliminaries}

We refer the reader to \cite
{GHKLMS:contlatt,AJ:domains,Mislove:topo:CS} for background on domain
theory and topology.  We shall write $x \in X \mapsto f (x)$ for the
function that maps every element $x$ of some space $X$ to the value $f
(x)$, sometimes omitting mention of the space $X$.

\paragraph{\em Domain Theory.}

A set with a partial ordering $\leq$ is a {\em poset\/}.  We write
$\upc E$ for $\{y \in X \mid \exists x \in E \cdot x \leq y\}$, $\dc E
\eqdef \{y \in X \mid \exists x \in E \cdot y\leq x\}$.  A {\em dcpo\/} is
a poset in which every directed family ${(x_i)}_{i \in I}$ has a least
upper bound (a.k.a., supremum or {\em sup\/}) $\sup_{i \in I} x_i$.
Symmetrically, we call {\em inf\/} (or infimum) any greatest lower
bound.  A family ${(x_i)}_{i \in I}$ is {\em directed\/} iff it is
non-empty, and any two elements have an upper bound in the family.
Any poset can be equipped with the {\em Scott topology\/}, whose opens
are the upward closed sets $U$ such that whenever ${(x_i)}_{i \in I}$
is a directed family that has a least upper bound in $U$, then some
$x_i$ is in $U$ already.
A dcpo $X$ is {\em pointed\/} iff it has a least element,
which we shall always write $\bot$.

Given a poset $X$, we shall write $X_\sigma$ for $X$ seen as a
topological space, equipped with its Scott topology.

We shall always consider $\Rp$, or
$\creal \eqdef \Rp \cup \{+\infty\}$, as posets, and implicit endow
them with the Scott topology of their ordering $\leq$.  We shall still
write ${\Rp}_\sigma$ or ${\creal}_\sigma$ to make this fact clear.  The
opens of $\Rp$ in its Scott topology are the intervals $(r, +\infty)$,
$r \in \Rp$, together with $\Rp$ and $\emptyset$.  Those of
$[0, 1]$ are $\emptyset$, $[0, 1]$, and $(r, 1]$, $r \in [0, 1)$.
Those of $\creal$ are $\emptyset$, $\creal$, and $(r, +\infty]$,
$r \in \Rp$.

Given two dcpos $X$ and $Y$, a map $f : X \to Y$ is {\em
  Scott-continuous\/} iff it is monotonic and $f (\sup_{i \in
  I} x_i) = \sup_{i \in I} f (x_i)$ for every directed family
${(x_i)}_{i \in I}$ in $X$.

We write $[X \to Y]$ for the space of all Scott-continuous maps from
$X$ to $Y$, for two dcpos $X$ and $Y$.  This is again a dcpo, with the
pointwise ordering.  {\updated We also write $\Lform X$ for
  $[X \to {\creal}_\sigma]$.  Instead of writing $(\Lform X)_\sigma$,
  we will implicitly assume that $\Lform X$ is given its Scott
  topology, unless we say otherwise.}

The {\em way-below\/} relation $\ll$ on a poset $X$ is defined by
$x \ll y$ iff, for every directed family ${(z_i)}_{i \in I}$ that has
a least upper bound $z$ such that $y \leq z$, then $x \leq z_i$ for
some $i \in I$ already.  We also say that $x$ {\em approximates\/}
$y$.  Note that $x \ll y$ implies $x \leq y$, and that
$x' \leq x \ll y \leq y'$ implies $x' \ll y'$.  However, $\ll$ is not
reflexive or irreflexive in general.  Write
$\uuarrow E \eqdef \{y \in X \mid \exists x \in E \cdot x \ll y\}$,
$\ddarrow E \eqdef \{y \in X \mid \exists x \in E \cdot y \ll x\}$.
$X$ is {\em continuous\/} iff, for every $x \in X$, $\ddarrow x$ is a
directed family, and has $x$ as least upper bound.  A {\em basis\/} is
a subset $B$ of $X$ such that any element $x \in X$ is the least upper
bound of a directed family of elements way-below $x$ {\em in $B$\/}.
Then $X$ is continuous if and only if it has a basis, and in this case
$X$ itself is the largest basis.

In a continuous poset with basis $B$, {\em interpolation holds\/}: if
$x_1$, \ldots, $x_n$ are finitely many elements way-below $x$, then
there is a $b \in B$ such that $x_1$, \ldots, $x_n$ are way-below $b$,
and $b \ll x$.  (See for example \cite[Section~4.2]{Mislove:topo:CS}.)
In this case, the Scott opens are exactly the unions of sets of the
form $\uuarrow b$, $b \in B$.

\paragraph{\em Topology.}

A topology $\Open$ on a set $X$ is a collection of subsets of $X$,
called the {\em opens\/}, such that any union and any finite
intersection of opens is open.  The \emph{interior} {\updated
  $\interior A$} of a subset $A$ of $X$ is the largest open included
in $A$.  A \emph{closed} subset is the complement of an open subset.
The \emph{closure} $cl (A)$ of $A$ is the smallest closed subset
containing $A$.  An \emph{open neighborhood} $U$ of a point $x$ is
merely an open subset that contains $x$.

A topology $\Open_1$ is \emph{finer} than $\Open_2$ if and only if it
contains all opens of $\Open_2$.  We also say that $\Open_2$ is
\emph{coarser} than $\Open_1$.

A {\em base\/} $\mathcal B$ (not a basis) of $\Open$ is a collection
of opens such that every open is a union of elements of the base.
Equivalently, a family $\mathcal B$ of opens is a base iff for every
$x \in X$, for every open $U$ containing $x$, there is a $V \in
\mathcal B$ such that $x \in V \subseteq U$.  A {\em subbase\/} of
$\Open$ is a collection of opens such that the finite intersections of
elements of the subbase form a base; equivalently, the coarsest
topology containing the elements of the subbase is $\Open$, and then
we say that $\Open$ is \emph{generated} by the subbase.

The specialization preorder of a space $X$ is defined by $x \leq y$ if
and only if for every open subset $U$ of $X$ that contains $x$, $U$
also contains $y$.  For every subbase $\mathcal B$ of the topology of
$X$, it is equivalent to say that $x \leq y$ if and only if every $U
\in \mathcal B$ that contains $x$ also contains $y$.  The
specialization preorder of a dcpo $X$, with ordering $\leq$, in its
Scott topology, is $\leq$.  A topological space is $T_0$ if and only
if $\leq$ is a partial ordering, not just a preorder.  A subset $A$ of
$X$ is \emph{saturated} if and only if it is upward-closed in the
specialization preorder $\leq$.

A map $f$ from $X$ to $Y$ is \emph{continuous} if and only if $f^{-1}
(V)$ is open in $X$ for every open subset $V$ of $Y$.  The
\emph{characteristic} function $\chi_A \colon X \to {\creal}_\sigma$ of
a subset $A$ of $X$, defined as mapping every $x \in A$ to $1$ and
every $x \not\in A$ to $0$, is continuous if and only if $A$ is open
in $X$.  When $X$ and $Y$ are posets in their Scott topology, $f
\colon X \to Y$ is continuous if and only if it is Scott-continuous.

A subset $K$ of a topological space $X$ is \emph{compact} if and only
if every open cover of $K$ has a finite subcover.  The image $f [K]$
of any compact subset of $X$ by any continuous map $f \colon X \to Y$
is compact in $Y$.  A \emph{homeomorphism} is a bijective continuous
map whose inverse is also continuous.

The product of two topological spaces $X$, $Y$ is the set $X \times Y$
with the \emph{product topology}, which is the coarsest that makes the
projection maps $\pi_1 \colon X \times Y \to X$ and $\pi_2 \colon X
\times Y \to Y$ continuous.  Equivalently, a base of this topology is
given by the \emph{open rectangles} $U \times V$, where $U$ is open in
$X$ and $V$ is open in $Y$.

If $X$ and $Y$ are posets, viewed as topological spaces in their Scott
topology, then there is some ambiguity about the notation $X \times
Y$.  We might indeed see the latter as a topological product, with the
just mentioned product topology, or as a product of two posets, with
the Scott topology of the product ordering.  In general, the Scott
topology on the poset product is strictly finer than the product
topology.  This difficulty vanishes when $X$ and $Y$ are continuous
posets: indeed, in this case the product poset is continuous as well,
and a base of the Scott topology in the product is given by the
subsets of the form $\uuarrow (x \times y)$, which are also open in
the product topology, as they coincide with $\uuarrow x \times
\uuarrow y$.

A topological space $X$ is \emph{locally compact} if and only if for
every $x \in X$, for every open subset $U$ of $X$ containing $X$,
there is a compact subset $K \subseteq U$ whose interior contains $x$.
Then we can require $K$ to be saturated as well, replacing $K$ by its
upward closure.  In a locally compact space, every open subset $U$ is
the union of the directed family of all open subsets $V$ such that $V
\subseteq Q \subseteq U$ for some compact saturated subset $Q$.  A
refinement of this is the notion of a \emph{core-compact} space, which
is by definition a space whose lattice of open subsets is a continuous
dcpo.  We shall agree to write $\Subset$ for the way-below relation on
the lattice of open sets of a space.  Every locally compact space is
core-compact, in which $V \Subset U$ if and only if $V \subseteq Q
\subseteq U$ for some compact saturated subset $Q$.

Every continuous poset is locally compact, since whenever $x \in U$,
$U$ open, there is an $y \in U$ such that $y \ll x$, so that we can
take $K \eqdef \upc y$, and $V \eqdef \uuarrow y$.

A topological space $X$ is \emph{well-filtered} if and only if for
every filtered family ${(Q_i)}_{i \in I}$ of compact saturated subsets
of $X$ whose intersection is contained in some open subset $U$, $Q_i
\subseteq U$ for some $i \in I$ already.  (A family of subsets is
\emph{filtered} if and only if it is directed in the reverse inclusion
ordering $\supseteq$.)  Every sober space is well-filtered
\cite[Theorem~II-1.21]{GHKLMS:contlatt}, and every continuous dcpo is
sober in its Scott topology \cite[Corollary~II-1.12]{GHKLMS:contlatt}.
(We won't define sober spaces; but see the paragraph on Stone duality
below.)

There are several topologies one can put on the space $[X \to Y]$ of
continuous maps from $X$ to $Y$, and more generally on any space $Z$
of continuous maps from $X$ to $Y$.  Looking at $Z$, resp., $[X \to
Y]$, as a subspace of the product $Y^X$ (i.e., the space of \emph{all}
maps from $X$ to $Y$), we obtain the topology of \emph{pointwise
  convergence}.  This is also the coarsest topology that makes all the
maps $f \mapsto f (x)$ continuous, for each $x \in X$.  We write $[X
\to Y]_{\mathsf p}$ for $[X \to Y]$ with the topology of pointwise
convergence.  Note that if $Z \subseteq [X \to Y]$, then the topology
of pointwise convergence on $Z$ is also the subspace topology from $[X
\to Y]_{\mathsf p}$.

When $Y$ is ${\creal}_\sigma$ or ${\Rp}_\sigma$, a subbasis of open
sets of $[X \to Y]_{\mathsf p}$ is given by the subsets
$[x > r] \eqdef \{f \in [X \to Y] \mid f (x) > r\}$, $x \in X$,
$r \in \Rp$.  Note that the latter are Scott open.  In particular, the
Scott topology on $\Lform X$
is always finer than the topology of pointwise convergence.

\paragraph{\em Coherence.}

A space $X$ is \emph{coherent} if and only if the intersection of any
two compact saturated subsets is again compact.

A related notion is the following.  Say the the way-below relation
$\ll$ on a lattice is \emph{multiplicative} (see
\cite[Definition~7.2.18]{AJ:domains} or
\cite[Proposition~I.4.7]{GHKLMS:contlatt}) if and only if for all $x,
y_1, y_2$ with $x \ll y_1$ and $x \ll y_2$, we have $x \ll \inf (y_1,
y_2)$.  We have called \emph{core-coherent} those spaces $X$ such that
the way-below relation $\Subset$ was multiplicative on the lattice of
open subsets of $X$ \cite[Comment before Definition~5.8]{JGL-mscs09}.

It is easy to see that every locally compact, coherent space is
(core-compact and) core-coherent.  Indeed, assume $V \Subset U_1$, $V
\Subset U_2$.  By local compactness, there are compact saturated
subsets $Q_1$, $Q_2$ such that $V \subseteq Q_1 \subseteq U_1$, $V
\subseteq Q_2 \subseteq U_2$.  Then $V \subseteq Q_1 \cap Q_2
\subseteq U_1 \cap U_2$.  Since $X$ is coherent, $Q_1 \cap Q_2$ is
compact.  So $V \Subset U_1 \cap U_2$.  In fact, one can check that
the sober, core-compact, and core-coherent spaces are exactly the
sober, locally compact and coherent spaces
\cite[Theorem~7.2.19]{AJ:domains}.

\paragraph{\em Stable compactness.}

A space that is sober, locally compact, and coherent is called
\emph{stably locally compact}.  It is \emph{stably compact} iff stably
locally compact and compact.  An equivalent definition of a stably
locally compact space is a space that is $T_0$, locally compact,
well-filtered, and coherent.

Stably compact spaces were first studied in \cite{Nachbin:toporder},
see also \cite[Section~VI-6]{GHKLMS:contlatt}, or
\cite[Section~2]{AMJK:scs:prob}.  They enjoy the following property,
called de Groot duality.  Let $X^\dG$ be $X$, topologized by taking as
opens the complements of compact saturated subsets of $X$.  Then
$X^\dG$ is stably compact again, and $X^{\dG\dG} = X$.  Moreover, the
specialization ordering of $X^\dG$ is $\geq$, the opposite of the
specialization ordering of $X$.

The \emph{patch topology} on a stably compact space $A$ is the
coarsest topology finer than those of $A$ and $A^\dG$, i.e., it is
generated by the opens of $A$ and the complements of compact saturated
subsets of $A$.  Write $A^\patch$ for $A$ with its patch topology:
this is a compact $T_2$ space, and the specialization ordering $\leq$
on $A$ has a closed graph in $A^\patch \times A^\patch$, making
$(A^\patch, \leq)$ a compact pospace \cite{Nachbin:toporder}.  We
shall be specially interested in $A \eqdef  {\creal}_\sigma$; then $A^\patch$
is merely $\creal$ with its ordinary $T_2$ topology generated by the
intervals $(a, b)$, $[0, b)$, and $(a, +\infty]$.  We shall also use
the fact that any product of stably compact spaces is stably compact,
and the $\_^\patch$ operation commutes with product
\cite[Proposition~14]{AMJK:scs:prob}.  So ${\creal}_\sigma^C$ is stably
compact, and $({\creal}_\sigma^C)^\patch$ is the compact $T_2$ space
$\creal^C$.

We agree to prefix with ``patch-'' any concept relative to patch
topologies.  E.g., a patch-closed subset of $A$ is a closed subset of
$A^\patch$, and a patch-continuous map from $A$ to $B$ is a continuous
map from $A^\patch$ to $B^\patch$.

\paragraph{\em Stone duality.}

There is a functor $\Open$ from the category of topological spaces to
the opposite of the category of frames (certain complete lattices),
defined by: for every topological space $X$, $\Open (X)$ is the frame
of its open subsets, and for every continuous map $f \colon X \to Y$,
$\Open (f)$ is the frame homomorphism that maps every $V \in \Open
(Y)$ to $f^{-1} (V) \in \Open (X)$.  This functor has a right adjoint
$\pt$, which maps every frame to its set of completely prime filters,
with the so-called hull-kernel topology.  For details, see
\cite[Section~7]{AJ:domains} or \cite[Section~V-5]{GHKLMS:contlatt}.

This adjunction establishes a correspondence between properties of
spaces $X$ and properties of frames $L$.  E.g., $X$ is core-compact if
and only if $\Open (X)$ is continuous, $X$ is core-coherent if and
only if the way-below relation on $\Open (X)$ is multiplicative, and
$X$ is compact if and only if the top element of $\Open (X)$ is
finite, i.e., way-below itself.

Going the other way around, $\pt (L)$ is always a sober space.  (In
fact, it is legitimate to call sober any topological space that is
homeomorphic to one of this form.)  For every topological space $X$,
the fact that $\pt$ is right adjoint to $\Open$ entails that $\pt
(\Open (X))$ obeys the following universal property: there is a
continuous map $\eta_X$ from $X$ to $\pt (\Open (X))$ (namely, the
unit of the adjunction), and every continuous map $f$ from $X$ to any
given sober space $Y$ extends to a unique continuous map $f^!$ from
$\pt (\Open (X))$ to $Y$, in the sense that $f^! \circ \eta_X = f$.  A
space obeying that universal property is called a \emph{sobrification}
$\Sober (X)$ of $X$.  Sobrification $\Sober$ is left adjoint to the
forgetful functor from sober spaces to topological spaces; in other
words, $\Sober (X)$ is a free sober space above $X$.  In particular,
all the sobrifications of $X$ are naturally isomorphic.

If $L$ is a continuous frame, then $\pt (L)$ is a locally compact,
sober space \cite[Theorem~7.2.16]{AJ:domains}.  Equivalently, it is
locally compact, $T_0$, and well-filtered
\cite[Theorem~II-1.21]{GHKLMS:contlatt}.  It follows that if $X$ is
core-compact, then its sobrification is locally compact, $T_0$ and
well-filtered.

If $L$ is arithmetic, i.e., is a continuous frame with a
multiplicative way-below relation, then $\pt (L)$ is stably locally
compact \cite[Theorem~7.2.19]{AJ:domains}.  So the sobrification of a
core-compact, core-coherent space is stably locally compact.
Similarly, the sobrification of a compact, core-compact, and
core-coherent space is stably compact.

For any topological space $X$, $X$ and its sobrification $\Sober (X)
\cong \pt (\Open (X))$ have isomorphic lattices of open subsets.  This
informally states that any construction, any property that can be
expressed in terms of opens will apply to both.

One easy consequence, which we shall use in
Section~\ref{sec:retr-errat-case}, is the following.
\begin{lemma}
  \label{lemma:Stone:h}
  The function $!$ that maps $h \in \Lform X$
  to $h^! \in [\Sober (X) \to {\creal}_\sigma]$ is an
  order-isomorphism, a homeomorphism when the function spaces are
  equipped with their Scott topologies, and is natural in $X$.
\end{lemma}
\begin{proof}
  ${\creal}_\sigma$ is sober, hence $h^!$ is well-defined.  Since
  $\Sober$ is left adjoint to the forgetful functor $U$ from sober
  spaces to topological spaces, the correspondence between morphisms
  $h \colon X \to U ({\creal}_\sigma)$ and morphisms $h^! \colon \Sober
  (X) \to {\creal}_\sigma$ is bijective, and is natural in $X$.  Both
  $!$ and its inverse $h' \mapsto h' \circ \eta_X$ are monotonic, and
  therefore define an order-isomorphism, hence an isomorphism of
  spaces with their Scott topologies.  (To show that $!$ is monotonic,
  consider $f \leq g$, and realize that $\sup (f^!, g^!)$ is a
  continuous map that extends $\sup (f, g)=g$, hence must coincide
  with $g^!$ by uniqueness of extensions.)
\end{proof}

\paragraph{\em Cones.}

A cone $C$ is an additive commutative monoid with a scalar
multiplication by elements of $\Rp$, satisfying laws similar to
those of vector spaces.  Precisely, a {\em cone\/} $C$ is endowed with
an addition $+ : C \times C \to C$, a zero element $0 \in C$, and a
scalar multiplication $\cdot : \Rp \times C \to C$ such that:
\[
\begin{array}{c}
  (x+y)+z = x+(y+z) \quad x+y=y+x \quad x+0=x \\
  (rs) \cdot x = r \cdot (s \cdot x) \quad
  1 \cdot x = x \quad 0 \cdot x = 0 \\
  r \cdot (x+y)= r\cdot x+ r \cdot y \quad
  (r+s) \cdot x = r \cdot x + s \cdot x \quad
\end{array}
\]
An \emph{ordered cone} is a cone with a partial ordering that makes
$+$ and $\cdot$ monotonic.  Similarly, a {\em topological cone\/} is a
cone equipped with a $T_0$ topology that makes $+$ and $\cdot$
continuous, where $\Rp$ is equipped with its Scott topology.  In a
\emph{semitopological cone}, we only require $+$ and $\cdot$ to be
separately continuous, not jointly continuous.

An important example of semitopological cone is given by the ordered
cones in which $+$ and $\cdot$ are Scott-continuous (the
\emph{s-cones}), in particular by Tix, Keimel and Plotkin's
\emph{d-cones} \cite{TKP:nondet:prob}, which are additionally required
to be dcpos.  As noticed by Keimel \cite[Remark before
Proposition~6.3]{Keimel:topcones2}, an s-cone $C$ may fail to be a
topological cone, unless $C$ is a \emph{continuous cone}, i.e., an
ordered cone that is continuous as a poset, and where $+$ and $\cdot$
are Scott-continuous.  In that case, the product of the Scott
topologies is the Scott topology of the product ordering, and separate
continuity implies joint continuity.

The most important cone we shall deal with is the ordered cone
$\Lform X$ 
of all continuous maps from $X$ to ${\creal}_\sigma$.  With its Scott
topology, it is a semitopological cone.  It is a topological cone if
$X$ is core-compact, since it is a continuous dcpo in that case, as a
special case of \cite[Theorem~II-4.7]{GHKLMS:contlatt}.  In previous
work, we were using the sub-cone $\langle X \to {\Rp}_\sigma\rangle$ of
all \emph{bounded} continuous maps from $X$ to ${\Rp}_\sigma$.  This is
again a semitopological cone that happens to be topological when $X$
is core-compact.

A subset $Z$ of a topological cone $C$ is {\em convex\/} iff $r \cdot
x + (1-r) \cdot y$ is in $Z$ whenever $x,y\in Z$ and $0\leq r \leq 1$.
$C$ is itself {\em locally convex\/} iff every point has a basis of
convex open neighborhoods, i.e., whenever $x \in U$, $U$ open in $C$,
then there is a convex open $V$ such that $x \in V \subseteq U$.

This is the notion of local convexity used by K. Keimel
\cite{Keimel:topcones2}.  Beware that there are others, such as
Heckmann's \cite{Heckmann:space:val}, which are inequivalent.

Every continuous cone $C$ is locally convex; this is a special case of
\cite[Lemma~6.12]{Keimel:topcones2}.  The argument, due to J. Lawson,
is as follows: let $x$ be a point in some open subset $U$ of $C$, then
there is a point $x_1 \ll x$ that is also in $U$, hence also an $x_2
\ll x_1$ that is again in $U$, and continuing this way we have a chain
$\ldots \ll x_n \ll \ldots \ll x_2 \ll x_1 \ll x$ of points of $U$.
Now let $V \eqdef \bigcup_{n \in \nat} \uuarrow x_n$: $V$ is clearly open,
$x \in V \subseteq U$, and $V$ is also convex because $V = \bigcup_{n
  \in \nat} \upc x_n$, as one can easily check.

We have already noticed that, given a core-compact space $X$,
$\Lform X$
is a continuous d-cone.  Therefore $\Lform X$
is a locally convex topological cone in that case.

A map $f : C \to \creal$ is {\em positively homogeneous\/} iff $f
(r\cdot x) = r f (x)$ for all $x \in C$ and $r \in \Rp$.  It is
{\em additive\/} (resp.\ {\em super-additive\/}, resp.\ {\em
  sub-additive\/}) iff $f (x+y)=f (x)+f (y)$ (resp.\ $f (x+y) \geq f
(x)+f (y)$, resp.\ $f (x+y) \leq f (x) + f (y)$) for all $x,y \in C$.
It is {\em linear\/} (resp.\ {\em superlinear\/}, resp.\ {\em
  sublinear\/}) iff its is both positively homogeneous and additive
(resp.\ super-additive, resp.\ sub-additive).  Every pointwise
supremum of sublinear maps is sublinear, and every pointwise infimum
of superlinear maps is superlinear.

One of the key results in the theory of cones is Keimel's Sandwich
Theorem \cite[Theorem~8.2]{Keimel:topcones2}, an analogue of the
Hahn-Banach Theorem: in a semitopological cone $C$, given a continuous
superlinear map $q \colon C \to \creal$ and a sublinear map $p \colon
C \to \creal$ such that $q \leq p$, there is a continuous linear map
$\Lambda \colon C \to \creal$ such that $q \leq \Lambda \leq p$.  Note
that $p$ need not be continuous or even monotonic for this to hold.
This is a feature we shall make use of in the proof of
Lemma~\ref{lemma:DN:s:compact}, allowing us to dispense with an
assumption of coherence.

Another construction we shall use is the \emph{lower Minkowski
  functional} $M_A$ of a non-empty closed subset $A$ of a topological
cone $C$ (this was called $F_A$ in \cite[Section~7]{Keimel:topcones2},
but this would conflict with some of our own notations).  This is
defined by:
\begin{eqnarray}
  M_A (x) & \eqdef & \inf \{b > 0 \mid (1/b) \cdot x \in A\}
  \label{eq:MA}
\end{eqnarray}
where we agree that the infimum is equal to $+\infty$ if no $b$ exists
such that $(1/b) \cdot x \in A$.  $M_A$ is continuous
\cite[Proposition~7.3~(a)]{Keimel:topcones2}, superlinear if and only
if $C \diff A$ is convex {\updated and proper (namely, if $A$ is
  non-empty)}, and sublinear if and only if $A$ is convex {\updated
  and non-empty} \cite[Lemma~7.5]{Keimel:topcones2}.

Using this, Keimel's Sandwich Theorem immediately implies the
following Separation Theorem \cite[Theorem~9.1]{Keimel:topcones2}: in
a semitopological cone $C$, for every convex non-empty subset $A$ and
every convex open subset $U$ such that $A \cap U = \emptyset$, there
is a continuous linear map $\Lambda \colon C \to {\creal}_\sigma$ such
that $\Lambda (x)\leq 1$ for every $x \in A$ and $\Lambda (x) > 1$ for
every $x \in U$.  We shall also use the following Strict Separation
Theorem \cite[Theorem~10.5]{Keimel:topcones2}: in a locally convex
semitopological cone $C$, for every compact convex subset $Q$ and
every non-empty closed convex subset $A$ such that $Q \cap A =
\emptyset$, there is a continuous linear map $\Lambda \colon C \to
{\creal}_\sigma$ and a real number $r > 1$ such that $\Lambda (x) \geq
r$ for every $x \in Q$, and $\Lambda (x) \leq 1$ for every $x \in A$.

\paragraph{\em Valuations, Previsions, Forks.}

A \emph{valuation} on a topological space $X$ is a map $\nu$ from the
lattice $\Open (X)$ of open subsets of $X$ to $\creal$ that is
\emph{strict} ($\nu (\emptyset)=0$), \emph{monotonic} (if $U \subseteq
V$ then $\nu (U) \leq \nu (V)$), and \emph{modular} ($\nu (U \cup V) +
\nu (U \cap V) = \nu (U) + \nu (V)$).  A valuation is
\emph{continuous} if and only if it is Scott-continuous, i.e., for
every directed family ${(U_i)}_{i \in I}$ of opens, $\nu (\bigcup_{i
  \in I} U_i) = \sup_{i \in I} \nu (U_i)$.  It is \emph{subnormalized}
iff $\nu (X) \leq 1$, and \emph{normalized} iff $\nu (X)=1$.  A
normalized continuous valuation is a continuous \emph{probability}
valuation.

Given any continuous map $h \colon X \to \creal$, one can define the
\emph{integral} $\int_{x \in X} h (x) d \nu$ of $h$ with respect to
$\nu$ in various equivalent ways.  One is by using a Choquet-type
formula \cite[Section~4.1]{Tix:bewertung}: $\int_{x \in X} h (x) d\nu$
is defined as $\int_0^{+\infty} \nu (h^{-1} (t, +\infty]) dt$, where
the latter is a Riemann integral, which is well-defined since the
integrated function is non-increasing.  Letting
$F (h) \eqdef \int_{x \in X} h (x) d\nu$, one realizes that $F$ is
linear and Scott-continuous on the cone $\Lform X$
\cite[Lemma~4.2]{Tix:bewertung}.

In particular, $F$ also defines a linear Scott-continuous map from the
cone $\langle X \to {\Rp}_\sigma\rangle$ to $\creal$.  Such
functionals were called (continuous) \emph{linear previsions} in
\cite{Gou-csl07}, except that they were not allowed to take the value
$+\infty$.  Conversely, any linear Scott-continuous map $F$ from
$\langle X \to {\Rp}_\sigma\rangle$ to $\creal$ extends to a unique
linear Scott-continuous map from $\Lform X$
to $\creal$, by $F (h) \eqdef \sup_{a \in \Rp} F (\min (h, a))$ for
example.

In general, call \emph{prevision} on $X$ any positively homogeneous,
Scott-continuous map $F$ from $\Lform X$
to $\creal$.  A prevision is \emph{Hoare} (or \emph{lower}) if and
only if it is sublinear, and is \emph{Smyth} (or \emph{upper}) if and
only if it is superlinear.  We say that $F$ is \emph{subnormalized}
(resp., \emph{normalized}) iff $F (\one+h) \leq 1+F (h)$ (resp., $=$) for
every $h \in \Lform X$, {\updated where $\one$ denotes the constant
  function with value $1$.}
These conditions simplify to $F (\one)\leq 1$ (resp., $F (\one)=1$) in
the case of linear previsions, but we shall need the more general form
when $F$ is not linear.

It is easy to see that the posets $\Val (X)$ of all continuous
valuations on $X$ (ordered by $\nu \leq \nu'$ iff $\nu (U) \leq \nu'
(U)$ for every open $U$) and $\Pred_{\Nature} (X)$ of all linear
previsions on $X$ (ordered by $F \leq G$ iff $F (h) \leq G (h)$ for
every $h \in \Lform X$)
are isomorphic.  This is the variant of the Riesz representation
theorem that we mentioned in the introduction.  In one direction, we
obtain $F$ from $\nu$ by $F (h) \eqdef \int_{x \in X} h (x) d\nu$, and
conversely we obtain $\nu$ from $F$ by letting $\nu (U) \eqdef F (\chi_U)$.
(E.g., this is a special case of
\cite[Satz~4.16]{Tix:bewertung}.)  
This isomorphism maps (sub)normalized continuous valuations to
(sub)normalized linear previsions and conversely.

To handle the mixture of probabilistic and \emph{erratic}
non-determinism, we rely on \emph{forks} \cite{Gou-csl07}.
A fork on $X$ is by definition a pair $(F^-, F^+)$ of a Smyth
prevision $F^-$ and a Hoare prevision $F^+$ satisfying \emph{Walley's
  condition}:
\[
F^- (h+h') \leq F^- (h) + F^+ (h') \leq F^+ (h+h')
\]
for all $h, h' \in \Lform X$.
(This condition was independently discovered by
\cite{KP:predtrans:pow}.)  By taking $h' \eqdef 0$, or $h \eqdef 0$,
this in particular implies that $F^- \leq F^+$.  A fork $(F^-, F^+)$
is (sub)normalized if and only if both $F^-$ and $F^+$ are.

Using some notation that we introduced in \cite{JGL:fullabstr:I},
write $\Pred_{\AN} (X)$ for the set of all Hoare previsions on $X$
($\Angel$ for angelic non-determinism, $\Nature$ for probabilistic
choice), $\Pred_{\DN} (X)$ for the set of all Smyth previsions on $X$.
Write $\Pred^1_{\AN} (X)$ for the set of normalized Hoare previsions,
$\Pred^{\leq 1}_{\AN} (X)$ for the set of subnormalized Hoare
previsions, and similarly for Smyth (subscript $\DN$) and linear
(subscript $\Nature$) previsions.  In any case, the \emph{weak
  topology} on any of these spaces $Y$ is generated by subbasic open
sets, which we write uniformly as $[h > r]$, and are defined as those
$F \in Y$ such that $F (h) > r$, where $h$ ranges over $\Lform X$
and $r$ over $\Rp$.  The specialization ordering of the weak topology
is pointwise: $F \leq F'$ iff $F (h) \leq F (h')$ for every
$h \in \Lform X$.
We use a $wk$ subscript, e.g., $\Pred^{\leq 1}_{\AN\;wk} (X)$, to
refer to a space in its weak topology.  Since $[h > r] = [1/r.h > 1]$
when $r \neq 0$, and $[h > 0] = \bigcup_{r > 0} [h > r]$, note that
the subsets of the form $[h > 1]$, $h \in \Lform X$
form a smaller subbase.

Similarly, write $\Pred_{\ADN} (X)$ for the set of all forks on $X$
($\Angel$ for angelic, $\Demon$ for demonic, and $\Nature$ for
probabilistic), and $\Pred^{\leq 1}_{\ADN} (X)$, $\Pred^1_{\ADN} (X)$
for their subsets of subnormalized, resp.\ normalized, forks.  On
each, define the \emph{weak topology} as the subspace topology induced
from the larger space $\Pred_{\DN\;wk} (X) \times \Pred_{\AN\;wk}
(X)$.  It is easy to see that a subbase of the weak topology is
composed of two kinds of open subsets: $[h > b]^-$, defined as
$\{(F^-, F^+) \mid F^- (h) > b\}$, and $[h > b]^+$, defined as
$\{(F^-, F^+) \mid F^+ (h) > b\}$, where $h \in \Lform X$,
$b \in \Rp$.  The specialization ordering of spaces of forks is the
product ordering $\leq \times \leq$, where $\leq$ denotes the
pointwise ordering on previsions.  As usual, we adjoin a subscript
``$wk$'' to denote spaces of forks with the weak topology.

Throughout, we shall use the convention of writing
$\Pred^\bullet_{\AN\;wk} (X)$, $\Pred^\bullet_{\DN\;wk} (X)$, etc.
$\bullet$ is either the empty superscript, ``$\leq 1$'', or ``$1$''.
This will allow us to factor some notation.  Later, we shall also use
this scheme for maps $s^\bullet_{\AN}$, $s^\bullet_{\DN}$, and so
forth.

\section{Retracting Powercones onto Previsions}
\label{sec:retr-powerc-onto}

Given a topological space $X$, let:
\begin{itemize}
\item the \emph{Hoare powerdomain} $\Hoare (X)$ be the set of all
  closed, non-empty subsets of $X$;
\item the \emph{Smyth powerdomain} $\Smyth (X)$ be the set of all
  compact saturated, non-empty subsets of $X$;
\item the \emph{Plotkin powerdomain} $\Plotkin (X)$ be the set of all
  lenses of $X$; a \emph{lens} $L$ is a non-empty subset of $X$ that
  can be written as the intersection of a compact saturated subset $Q$
  and a closed subset $F$ of $X$.  A canonical form is obtained by
  taking $Q \eqdef \upc L$ and $F \eqdef cl (L)$.
\end{itemize}
$\Hoare (X)$ is the traditional model for so-called \emph{angelic}
non-determinism.  We write $\HV (X)$ for $\Hoare (X)$ with its
\emph{lower Vietoris topology}, generated by subbasic open sets
$\Diamond V \eqdef \{C \in \Hoare (X) \mid C \cap V \neq \emptyset\}$,
where $V$ ranges over the open subsets of $X$.  Its specialization
ordering is inclusion $\subseteq$, and we shall also consider $\Hoare
(X)$ as a dcpo in this ordering.
The Scott topology is always finer than the lower Vietoris topology,
and coincides with it when $X$ is a continuous dcpo in its Scott
topology \cite[Section~6.3.3]{schalk:diss}.

$\Smyth (X)$ is the traditional model for so-called \emph{demonic}
non-determinism.  We write $\SV (X)$ for $\Smyth (X)$ with its
\emph{upper Vietoris topology}, generated by basic open sets
$\Box V \eqdef \{Q \in \Smyth (X) \mid Q \subseteq V\}$, where $V$
ranges over the open subsets of $X$.  Its specialization ordering is
reverse inclusion $\supseteq$.  When $X$ is $T_0$, well-filtered, and
locally compact, $\Smyth (X)$ is a continuous dcpo, and the Scott and
upper Vietoris topologies coincide \cite[Section~7.3.4]{schalk:diss}.

$\Plotkin (X)$ is the traditional model for \emph{erratic}
non-determinism.  Write $\PV (X)$ for $\Plotkin (X)$ with its
\emph{Vietoris topology}, generated by subbasic open sets
$\Box V \eqdef \{L \in \Plotkin (X) \mid L \subseteq V\}$ and
$\Diamond V \eqdef \{L \in \Plotkin (X) \mid L \cap V \neq \emptyset
\}$.  Its specialization ordering is the topological Egli-Milner
ordering $\sqsubseteq_{\mathrm{EM}}$, defined by
$L \sqsubseteq_{\mathrm{EM}} L'$ iff $\upc L \supseteq \upc L'$ and
$cl (L) \subseteq cl (L')$.

When $X$ is a semitopological cone $C$, it makes sense to consider the
subsets $\Hoare^{cvx} (C)$, $\Smyth^{cvx} (C)$, $\Plotkin^{cvx} (C)$
of those elements of $\Hoare (C)$, $\Smyth (C)$, $\Plotkin (C)$
respectively that are convex.  We again equip them with their
respective (lower, upper, plain) Vietoris topologies, yielding spaces
which we write with a $\mathcal V$ subscript, and which happen to be
subspaces of $\HV (C)$, $\SV (C)$, $\PV (C)$ respectively.  Their
specialization orderings are as for their non-convex variants, and
again give rise to Scott topologies.
But beware that the Scott topologies on the latter may fail to be
subspace topologies.

We now define formally the maps that we named $r$ and $s$ in the
introduction.  They come in three flavors, angelic, demonic, and
erratic; but we shall ignore the erratic forms for now.  Since
continuous valuations are isomorphic to spaces of linear of
previsions, we reason with the latter; this is what we shall really
need in proofs.
\begin{definition}
  \label{defn:rs}
  Let $X$ be a topological space.  For every non-empty set $E$ of
  linear previsions on $X$, let $r_{\AN} (E) \colon \Lform X \to
  \creal$
  (resp., $r_{\DN} (E)$) map $h$ to $\sup_{G \in E} G (h)$ (resp.,
  $\inf$).

  Conversely, for every Hoare prevision (resp., subnormalized,
  normalized) $F$ on $X$, let $s_{\AN} (F)$ (resp., $s^{\leq 1}_{\AN}
  (F)$, $s^1_{\AN} (F)$) be the set of all linear previsions (resp.,
  subnormalized, normalized) $G$ such that $G \leq F$.  For every
  Smyth prevision (resp., subnormalized, normalized) $F$ on $X$, let
  $s_{\DN} (F)$ (resp., $s^{\leq 1}_{\DN} (F)$, $s^1_{\DN} (F)$) be
  the set of all linear previsions (resp., subnormalized, normalized)
  $G$ such that $F \leq G$.
\end{definition}

Our aim in this section is to show that the various matching pairs of
maps $r$ and $s$ form retractions onto the adequate spaces.  A
\emph{retraction} of $X$ onto $Y$ is a pair of two continuous maps $r
\colon X \to Y$ (also called, somewhat ambiguously, a retraction) and
$s \colon Y \to X$ (called the associated section) such that $r \circ
s = \identity Y$.

Some of our retractions will have the extra property that $s \circ r
\leq \identity X$, meaning that $s$ is left-adjoint to $r$.  We shall
call such retraction \emph{embedding-projection pair}, following a
domain-theoretic tradition \cite[Definition~3.1.15]{AJ:domains}; $r$
is the projection, and $s$ is the embedding.  Accordingly we shall say
that $X$ \emph{projects onto} $Y$ through $r$ in that situation.
There is also a dual situation where $\identity X \leq s \circ r$
instead, meaning that $s$ is right-adjoint to $r$.
In that case, we shall say that $X$ \emph{coprojects onto} $Y$ through
$r$, that the latter is a \emph{coprojection}, and that $s$ is the
associated \emph{coembedding}.

We start with the angelic case.  We shall deal with the demonic case
in Section~\ref{sec:retr-demon-case}, and with the erratic case in
Section~\ref{sec:retr-errat-case}.

\subsection{The Retraction in the Angelic Case}
\label{sec:retr-angel-case}

{\updated%
  Our aim in this section is to prove that if $X$ is a topological space
  satisfying certain conditions, then $r_\AN$ and and
  $s^\bullet_{\AN}$ form a retraction.
}


We proceed through a series of lemmata.  The first one is clear.
\begin{lemma}
  \label{lemma:AN:r:def}
  Let $X$ be a topological space.  For every $C \in \HV
  (\Pred^\bullet_{\Nature\;wk} (X))$, $r_{\AN} (C) = (h \in \Lform X
  \mapsto \sup_{G \in C} G (h))$ is an element of
  $\Pred^\bullet_{\AN} (X)$.
\end{lemma}

For clarity, write $[h > b]_{\AN}$ or $[h > b]_\Nature$ for the
subbasic open subset $[h > b]$, depending whether in a space of Hoare
previsions or of linear previsions.
\begin{lemma}
  \label{lemma:AN:r:cont}
  Let $X$ be a topological space.  The map $r_{\AN}$ is continuous
  from $\HV (\Pred^\bullet_{\Nature\;wk} (X))$ to
  $\Pred^\bullet_{\AN\;wk} (X)$.
\end{lemma}
\begin{proof}
  The inverse image of the subbasic open $[h > b]_{\AN} \subseteq
  \Pred^\bullet_{\AN\;wk} (X)$ by $r_{\AN}$ is $\Diamond [h >
  b]_\Nature$, hence is open in $\HV (\Pred^\bullet_{\Nature\;wk} (X))$.
\end{proof}

{\updated%
  
We will now need to consider spaces $X$ of a certain kind, which we
call $\AN_\bullet$-friendly in Definition~\ref{defn:LX:Hfriendly}
below.  A \emph{locally convex-compact} cone is a semitopological cone
in which every element has a base of convex compact neighborhoods
\cite[Definition~4.9]{Keimel:topcones2}.  An \emph{LCS-complete} space
is a space that is homeomorphic to a $G_\delta$ subspace of a locally
compact sober space; a $G_\delta$ subset is an intersection of
countably many open sets.  The notion was introduced in
\cite{dBGLJL:LCS}.

\begin{newdefinition}[$\AN_\bullet$-friendly]
  \label{defn:LX:Hfriendly}
  An \emph{$\AN$-friendly space} is a
  topological space $X$ such that $\Lform X$ is a locally convex
  semitopological cone when given its Scott topology.  A space is
  \emph{$\AN_{\leq 1}$-friendly} if
  and only if it is $\AN$-friendly.

  An \emph{$\AN_1$-friendly space} is either:
  \begin{enumerate}
  \item a compact $\AN$-friendly space $X$,
  \item or a space $X$ such that $\Lform X$ is a locally
    convex-compact, locally convex sober topological cone,
  \item or an LCS-complete space.
  \end{enumerate}
\end{newdefinition}
We will use the adjective $\AN_\bullet$-friendly to denote any of the
notions, where $\bullet$ is nothing, ``$\leq 1$'', or ``$1$''.  It is
immediate that $\AN_1$-friendliness implies $\AN$-friendliness.

In the 2017 version of this paper \cite{JGL-mscs16}, of which this is
a revised version, the conditions we imposed on $X$ was simply that
$\Lform X$ was locally convex.  We do not have any counterexample or
result showing that the extra condition of $\AN_1$-friendliness is
needed in case $\bullet$ is ``$1$'', but the proof of Lemma~3.4 of
\cite{JGL-mscs16} is wrong, and the proofs below need
$\AN_\bullet$-friendliness.  We have arranged for definitions,
lemmata, propositions and theorems of \cite{JGL-mscs16} to keep their
numbers; for example, this is the reason for the funny name
``Definition~\ref{defn:LX:Hfriendly}''.

We discuss standard classes of $\AN_\bullet$-friendly spaces.

Given a space $X$, and a compact subset $Q$ of $X$, the set
$\blacksquare Q$ of open neighborhoods of $Q$ is Scott-open in
$\Open X$, the lattice of open subsets of $X$.  A space $X$ is
\emph{consonant} if and only those sets $\blacksquare Q$ generate the
Scott topology on $\Open X$.  Every locally compact space is consonant
(see Exercise~5.4.12 in \cite{JGL-topology}), every regular
\v{C}ech-complete space is consonant \cite[Theorem~4.1 and
footnote~8]{DGL:consonant}, and every LCS-complete space is consonant
\cite[Proposition~12.1]{dBGLJL:LCS}.  The class of LCS-complete spaces
contains M. de Brecht's quasi-Polish spaces \cite{deBrecht:qPolish},
which themselves contain all Polish spaces and all $\omega$-continuous
dcpos in their Scott topology; the LCS-complete spaces also contain
all the continuous dcpos.

A space $X$ is \emph{$\odot$-consonant} if and only if every finite
copower of $X$ is consonant \cite[Definition~13.1]{dBGLJL:LCS}.  Since
locally compact spaces and LCS-complete spaces are closed under finite
coproducts (in the category of topological spaces), those spaces are
$\odot$-consonant; see \cite[Lemma~13.2]{dBGLJL:LCS} for the case of
LCS-complete spaces.

By \cite[Lemma~13.6]{dBGLJL:LCS}, for every $\odot$-consonant space
$X$, $\Lform X$ is locally convex (in its Scott topology).  We have
just illustrated how large the class of $\odot$-consonant spaces is,
and they are all $\AN$-friendly.

We turn to $\AN_1$-friendliness, and specifically to conditions~2
and~3.

\begin{newremark}
  \label{rem:Hfriendly:1}
  For every core-compact space $X$, $\Lform X$ is a locally convex,
  locally convex-compact, sober topological cone.  In order to see
  this, we observe that $\Lform X$ is a continuous complete lattice.
  This was mentioned at the end of Section~16 of \cite{dBGLJL:LCS},
  and can be deduced from \cite[Proposition~II-4.6]{GHKLMS:contlatt},
  which says that if $X$ is core-compact and $L$ (here, $\creal$) is
  an injective $T_0$ space (i.e., a continuous complete lattice by
  \cite[Theorem~II-3.8]{GHKLMS:contlatt}), then the poset $[X \to L]$
  of continuous functions from $X$ to $L$ is a continuous complete
  lattice.

  A continuous dcpo, in particular a continuous complete lattice, is
  always sober in its Scott topology \cite[Proposition
  8.2.12]{JGL-topology}.  It is also a c-space in its Scott topology
  \cite[Proposition 5.1.37]{JGL-topology}.  Together with its cone
  operations, $\Lform X$ is then a c-cone in the sense of Keimel,
  which is then locally convex, locally convex-compact and topological
  \cite[Lemma~6.12]{Keimel:topcones2}.  Hence every core-compact space
  $X$ is $\AN_1$-friendly, in that it satisfies condition~2 of
  Definition~\ref{defn:LX:Hfriendly}.  This applies in particular to
  all locally compact spaces $X$, and therefore to all continuous
  dcpos in their Scott topology.

  Hence every core-compact space $X$ satisfies condition~2 of
  Definition~\ref{defn:LX:Hfriendly}.  There is a partial converse:
  the two conditions are equivalent if $X$ is a poset with its Scott
  topology, or a consonant space, or a second-countable space.
  Indeed, in these cases, $\Lform X$ is locally compact if and only if
  $X$ is core-compact \cite[Theorem~G]{JGL:LX:blog}, and certainly
  local convex-compactness implies local compactness.
\end{newremark}

\begin{newremark}
  \label{rem:Hfriendly:1:LCS}
  As for condition~3, for every LCS-complete space $X$, we recall that
  $X$ is $\odot$-consonant, hence that $\Lform X$ is locally convex in
  its Scott topology \cite[Lemma~13.6]{dBGLJL:LCS}.  Therefore every
  LCS-complete space $X$ is $\AN$-friendly.  As a consequence, every
  $\AN_1$-friendly space is $\AN$-friendly (and
  $\AN_{\leq 1}$-friendly).
\end{newremark}

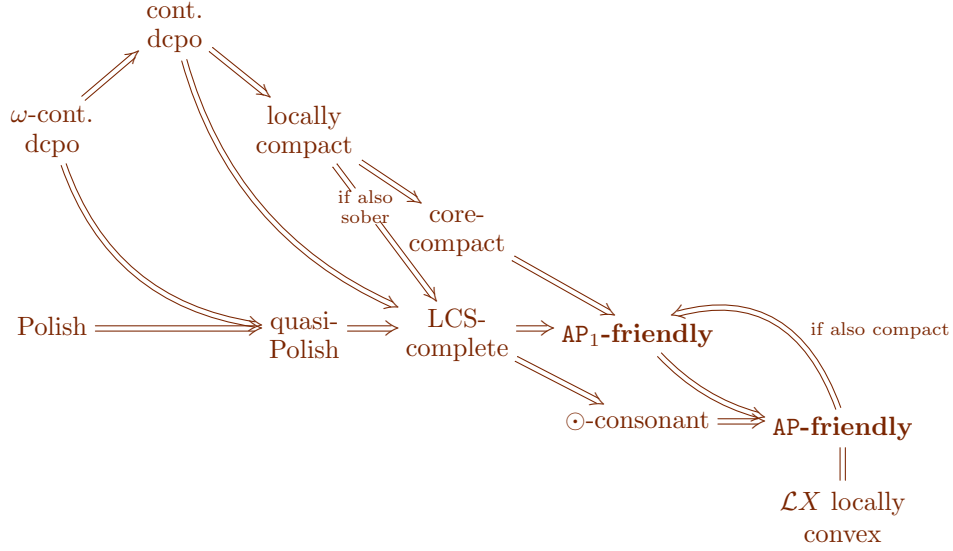
\begin{figure}
  \updated
  \centering
  \[
    \xymatrix@C-10pt@R-10pt{ & \txt{cont.\\dcpo}
      \ar@{=>}[rd] \ar@/_2pc/@{=>}[rrddd] \\
      \txt{$\omega$-cont.\\dcpo} \ar@{=>}[ur] \ar@/_2pc/@{=>}[rrdd] &&
      \txt{locally\\compact} \ar@{=>}[rd]
      \ar@{=>}[rdd]|(0.4){\scriptsize\txt{if also\\
          sober} }
      \\
      &&& \txt{core-\\compact} \ar@{=>}[rd]
      & \\
      \text{Polish} \ar@{=>}[rr] && \txt{quasi-\\Polish} \ar@{=>}[r] &
      \txt{LCS-\\complete} \ar@{=>}[r] \ar@{=>}[rd] & \txt{\bf
        $\AN_1$-friendly} \ar@/_1pc/@{=>}[rd]
      & \\
      &&&& \odot\text{-consonant} \ar@{=>}[r] & \txt{\bf
        $\AN$-friendly} \ar@/_2pc/@{=>}[ul]_(0.3){\text{if also
          compact}} \ar@{=}[d]
      \\
      &&&&& {
        \begin{array}{c}
          \Lform X \text{ locally}\\
          \text{convex}
        \end{array}
      } }
  \]
  \caption{Implications between properties of spaces and $\AN_\bullet$-friendliness}
  \label{fig:APfriendly}
\end{figure}

We sum up this discussion (and perhaps a bit more) by the 
diagram of implications of Figure~\ref{fig:APfriendly}.

The following is Lemma~3.4 of \cite{JGL-mscs16} in the case that
$\bullet$ is nothing or ``$\leq 1$''.
\begin{newlemma}
  \label{lemma:3.4:easy}
  Let $\bullet$ be nothing or ``$\leq 1$'', and $X$ be an
  $\AN$-friendly space.  For every $F \in \Pred^\bullet_{\AN} X$, for
  every $h_0 \in \Lform X$, and every real number $r > 0$ such that
  $F (h_0) > r$, there is a $G \in \Pred^\bullet_{\Nature} X$ with
  $G \leq F$ and $G (h_0) > r$.
\end{newlemma}
\begin{proof}
  By Corollary~9.5 of \cite{Keimel:topcones2}, every lower
  semi-continuous sublinear map from a locally convex semitopological
  cone to $\creal$ is pointwise the supremum of lower semi-continuous
  linear maps.
\end{proof}

The situation is more complex when $\bullet$ is ``$1$''.  We need a
few technical lemmas.  The first one deals with condition~1 of the
definition of $\AN_1$-friendliness, namely when $X$ is compact.

\begin{newlemma}
  \label{lemma:chi+U:compact}
  For every compact topological space $X$, for every $a \in \Rp$, the
  set $\{h \in \Lform X \mid \exists b > a, h \geq b \cdot \one\}$ is
  Scott-open in $\Lform X$.  For every Scott-open subset $U$ of
  $\Lform X$, the set
  $\{h \in \Lform X \mid \exists b > a, k \in U, h \geq b \cdot \one +
  k\}$ is open in $\Lform X$.
\end{newlemma}
\begin{proof}
  We prove the second claim.  The first claim will follow by taking
  $U \eqdef \Lform X$.  Let
  $V \eqdef \{h \in \Lform X \mid \exists b > a, k \in U, h \geq b
  \cdot \one + k\}$.  This is upwards-closed.  Let ${(h_i)}_{i \in I}$
  be a directed family with supremum $h \in V$.  Then
  $h \geq b\cdot \one + k$ for some $b > a$ and some $k \in U$.  In
  particular, $h \geq b\cdot \one$.  Let us pick a $c \in {]a, b[}$.
  
  We claim that there is an $i \in I$ such that
  $h_i \geq c\cdot \one$.  For every $x \in X$, $h (x) \geq b > c$, so
  $h_i (x) > c$ for some $i \in X$.  It follows that the sets
  ${(h_i^{-1} (]c, \infty]))}_{i \in I}$ form an open cover of $X$.
  Also, if $h_i \leq h_j$ then
  $h_i^{-1} (]c, \infty]) \subseteq h_j^{-1} (]c, \infty])$, so the
  family ${(h_i^{-1} (]c, \infty]))}_{i \in I}$ is directed.  Since
  $X$ is compact, it is therefore included in $h_i^{-1} (]c, \infty])$
  for some $i \in I$.  Hence, for that choice of $i$,
  $h_i \geq c\cdot \one$.

  We have $h \geq c\cdot \one + k$.  Let $h'$ map every $x \in X$ to
  $\max (h (x) - c, 0)$, and for every $j \in I$, let $h'_j$ map every
  $x \in X$ to $\max (h_j (x) - c, 0)$.  Then
  $h' = \sup_{j \in I} h'_j$, and $h' \geq k \in U$.  Since $U$ is
  Scott-open, some $h'_j$ is in $U$.  Let $\ell \in I$ be such that
  $h_\ell \geq h_i, h_j$.  Then $h'_\ell \geq h'_j$, so $h'_\ell$ is in $U$.
  Also, $h_\ell \geq h_i \geq c\cdot \one$, so for every $x \in X$,
  $h'_\ell (x) = \max (h_\ell (x) - c, 0) = h_\ell (x) - c$.  It follows that
  $h_\ell = c \cdot \one + h'_\ell$, so $h_\ell \in V$.
\end{proof}

Our second technical lemma will be used when conditions~2 or~3 of the
definition of $\AN_1$-friendliness are met.  Given two elements $k$
and $h$ of $\Lform X$, we write $[k, h]$ for the convex hull of
$\{k, h\}$, namely the collection $\{a \cdot k + (1-a) \cdot h \mid a
\in [0, 1]\}$. 
\begin{newlemma}
  \label{lemma:Hfriendly:1:V}
  Let $X$ be a topological space such that $\Lform X$ is a locally
  convex, locally convex-compact, sober topological cone.  Let
  $g_0 \geq g_1 \geq \cdots \geq g_n \geq \cdots$ be a descending
  sequence of elements of $\Lform X$, and $h_0 \in \Lform X$.  Let
  also $W$ be a Scott-open neighborhood of $h_0$ in $\Lform X$ such
  that $[b \cdot g_n, h]$ is included in $W$ for every $n \in \nat$,
  every $b > 1$ and every $h \in W$.  There is a convex Scott-open
  neighborhood $V$ of $h_0$ in $\Lform X$ that contains every function
  $b \cdot g_n$ with $n \in \nat$ and $b > 1$ and that is included in
  $W$.
\end{newlemma}
\begin{proof}
  Lemma~4.10~(c) of \cite{Keimel:topcones2} states that in a
  topological cone, the convex hull of a finite union of convex
  compact sets is compact; and therefore the upward closure of such a
  convex hull is compact saturated, and still convex. Since $\Lform X$
  is topological, every set $\upc [h, k]$ is therefore convex and
  compact saturated.

  We build convex compact saturated sets $Q_n$ containing
  $\upc [(1+1/2^n) \cdot g_n, h_0]$ and included in $W$, so that $Q_n$
  is included in the interior $\interior {Q_{n+1}}$ of ${Q_{n+1}}$ for
  every $n \in \nat$, by induction on $n$, as follows.

  We let $Q_0 \eqdef \upc [2 \cdot g_0, h_0]$.  Since $h_0 \in W$, by
  assumption on $W$, $[2 \cdot g_0, h_0]$ is included in $W$.  Since
  $W$ is upwards-closed, $Q_0 \subseteq W$.

  Having built $Q_n \subseteq W$, we let $Q'_n$ be the upward closure
  of the convex hull of the union of $\{(1+1/2^{n+1}) \cdot g_{n+1}\}$
  and $Q_n$.  By Lemma~4.10~(c) of \cite{Keimel:topcones2}, which we
  have already used above, $Q'_n$ is convex and compact saturated.
  Additionally, every element $h$ of $Q'_n$ is above some element in
  $[(1+1/2^{n+1}) \cdot g_{n+1}, k]$ for some $k \in Q_n$ (hence
  $k \in W$).  Indeed, every element of the convex hull of the union
  of two convex sets $A$ and $B$ can be written as
  $\alpha \cdot a + (1-\alpha) \cdot b$ for some $\alpha \in [0, 1]$,
  $a \in A$ and $b \in B$.

  By assumption on $W$, $[(1+1/2^{n+1}) \cdot g_{n+1}, k]$ is included
  in $W$, and since $W$ is upwards closed, it follows that every
  element $h \in Q'_n$ is in $W$.  Therefore $Q'_n \subseteq W$.  The
  first half of Proposition~10.6 of \cite{Keimel:topcones2} states
  that in a locally convex, locally convex-compact, sober topological
  cone, every non-empty convex compact saturated set has a base of
  convex compact saturated neighborhoods.  Hence there is a convex
  compact saturated set $Q_{n+1}$ included in $W$ whose interior
  contains $Q_n$.  The induction is complete.

  Let $V \eqdef \bigcup_{n \in \nat} Q_n$.  $V$ contains $h_0$, since
  $h_0 \in Q_0$. $V$ is convex since it is a directed union of convex
  sets. Then,
  $V \subseteq \bigcup_{n \in \nat} \interior {Q_{n+1}} \subseteq
  \bigcup_{n \in \nat} Q_{n+1} \subseteq V$, so
  $V = \bigcup_{n \in \nat} \interior {Q_{n+1}}$, and therefore $V$ is
  open.  Since $Q_n \subseteq W$ for every $n \in \nat$, $V$ is also
  included in $W$.  Finally, $(1+1/2^n) \cdot g_n$ is in $Q_n$ for
  every $n \in \nat$, hence in $V$.  For $n \in \nat$, $b > 1$, let us
  pick $m \in \nat$ such that $m \geq n$ and $1+1/2^m \leq b$.  Then
  $(1+1/2^m) \cdot g_m$ is in $V$, and
  $(1+1/2^m) \cdot g_m \leq b \cdot g_m \leq b \cdot g_n$, and since
  $V$ is upwards-closed, $b \cdot g_n \in V$.
\end{proof}

\begin{newlemma}
  \label{lemma:Hfriendly:1:V:aux}
  Let $X$ be a space satisfying condition~1 or~2 of the definition of
  $\AN_1$-friendliness.  Let $h_0 \in \Lform X$ and $W$ be a
  Scott-open neighborhood of $h_0$ in $\Lform X$ such that
  $[b \cdot \one, h]$ is included in $W$ for every $b > 1$ and every
  $h \in W$.  There is a convex Scott-open neighborhood $V$ of $h_0$
  in $\Lform X$ that contains every function $b \cdot \one$ with
  $b > 1$ and that is included in $W$.
\end{newlemma}
\begin{proof}
  When $X$ satisfies condition~2, namely when $\Lform X$ is a locally
  convex-compact, locally convex sober topological cone, the result is
  simply Lemma~\ref{lemma:Hfriendly:1:V} with $g_n \eqdef \one$ for
  every $n \in \nat$.

  Hence we concentrate on the case where $X$ satisfies condition~1,
  namely when $\Lform X$ is locally convex and $X$ is compact.  Since
  $\Lform X$ is locally convex and $h_0 \in W$, there is a convex
  Scott-open neighborhood $U$ of $h_0$ included in $W$.

  For every $a \in {]0, 1[}$, the set
  $(1-a) \cdot U \eqdef \{(1-a) \cdot h \mid h \in U\}$ is Scott-open,
  because it is equal to the inverse image of $U$ by the map
  $h \mapsto 1/(1-a) \cdot h$.  Hence the set
  $V_a \eqdef \{h \in \Lform X \mid \exists b > a, k \in (1-a) \cdot
  U, h \geq b\cdot \one + k\}$ is Scott-open by
  Lemma~\ref{lemma:chi+U:compact}, which applies since $X$ is compact.
  When $a=1$, we define $V_1$ as
  $\{h \in \Lform X \mid \exists b > 1, h \geq b\cdot \one\}$, and
  that is also Scott-open by the same lemma.  When $a=0$, we define
  $V_0$ as $U$.  For uniformity of treatment, we note that $V_a$ is
  equal to
  $\{h \in \Lform X \mid \exists b > 1, k \in U, h \geq ba \cdot \one
  + (1-a) \cdot k\}$ for every $a \in [0, 1]$, including when $a=0$ or
  $a=1$.  Let $V \eqdef \bigcup_{a \in [0, 1]} V_a$.  This is
  Scott-open, and we claim that $V$ is convex.  Let $h$ and $h'$ be
  two elements of $V$, say $h \geq ba\cdot \one + (1-a) \cdot k$ and
  $h' \geq b'a' \cdot \one + (1-a') \cdot k'$, with $b, b' > 1$,
  $k, k' \in U$, and $a, a' \in [0, 1]$.  Let $\alpha \in {]0, 1[}$.
  Then
  $\alpha \cdot h + (1-\alpha) \cdot h' \geq (\alpha ba + (1-\alpha)
  b'a') \cdot \one + (\alpha (1-a) \cdot k + (1-\alpha) (1-a') \cdot
  k')$.  Let $a'' \eqdef \alpha a + (1-\alpha) a'$.  Then
  $a'' \in [0, 1]$, and $\alpha ba + (1-\alpha) b'a' = b'' a''$ for
  some $b'' > 1$.  Indeed,
  $\alpha ba + (1-\alpha) b'a' \geq \alpha a + (1-\alpha) a' = a''$.
  If $a'' \neq 0$, then $a \neq 0$ or $a' \neq 0$, and since
  $\alpha \neq 0, 1$, the inequality is strict, so we can take
  $b'' \eqdef (\alpha ba + (1-\alpha) b'a') / a''$; if $a''=0$, then
  $a=a'=0$, and any $b'' > 1$ fits.  Also,
  $\alpha (1-a) \cdot k + (1-\alpha) (1-a') \cdot k' = (1-a'') \cdot
  k''$, where
  $k'' \eqdef \frac {\alpha (1-a)} {1-a''} k + \frac {(1-\alpha)
    (1-a')} {1-a''} k'$ if $a'' \neq 1$ (otherwise, we take for $k''$
  an arbitrary element of $U$, for example $h_0$).  Noting that
  $\alpha (1-a) + (1-\alpha) (1-a')
  = 
  1-a''$, $k''$ is a convex combination of $k$ and $k'$ (if
  $a'' \neq 1$), hence is in the convex set $U$.  To sum up,
  $\alpha \cdot h + (1-\alpha) \cdot h' \geq b''a'' \cdot \one +
  (1-a'') \cdot k''$, where $b'' > 1$ and $k'' \in U$ (even if
  $a''=1$).  It follows that $\alpha \cdot h + (1-\alpha) \cdot h'$ is
  in $V_{a''}$, hence in $V$.  Therefore $V$ is convex.

  Every function $b \cdot \one$ with $b>1$ is in $V_1$, hence in $V$,
  and $h_0$ is in $V_0=U$ hence in $V$.  We claim that
  $V \subseteq W$.  For every $h \in V$, there is an $a \in [0, 1]$, a
  number $b > 1$, and a $k \in U$ such that
  $h \geq ba \cdot \one + (1-a) \cdot k$.  The term
  $ba \cdot \one + (1-a) \cdot k$ is in $[b \cdot \one, k]$, which is
  included in $W$ since $k \in U \subseteq W$, and by our assumption
  on $W$.  Since $W$ is open hence upwards-closed, $h \in W$.
\end{proof}

We now deal with the case of LCS-complete spaces.

For every continuous map $f \colon X \to Y$, for every continuous
valuation $\nu$ on $X$, there is an \emph{image} continuous valuation
$f [\nu]$, which maps every open subset $V$ of $Y$ to
$\nu (f^{-1} (V))$.  There is a \emph{change-of-variable formula}: for
every $h \in \Lform Y$,
$\int_{y \in Y} h (y) \,df[\nu] = \int_{x \in X} h (f (x)) \,d\nu$.
This is clear considering our definition of the integral by a
Choquet-type formula.

A continuous valuation $\nu$ on a space $Y$ is \emph{supported} on a
subset $X$ if and only if for all open subsets $U$ and $U'$ of $Y$
such that $U \cap X = U' \cap X$, $\nu (U) = \nu (U')$
\cite[Definition~4.5]{JGL:distval:I}.  Giving $X$ the subspace
topology, and writing $i$ for the inclusion map $X \to Y$, Lemma~4.6
of \cite{JGL:distval:I} states that a continuous valuation $\nu$ on
$Y$ is supported on $X$ if and only if there is a (unique) continuous
valuation $\mu$ on $X$ such that $\nu = i [\mu]$.

The following lemma was already implicit in the proof of
\cite[Theorem~1.1]{dBGLJL:LCS}.  A continuous valuation $\nu$ on $Y$
is \emph{bounded} if and only if $\nu (Y) < \infty$.
\begin{newlemma}
  \label{lemma:supp:Gdelta}
  Let $Y$ be an LCS-complete space, $V_n$ ($n \in \nat$) be countably
  many open subsets of $Y$, and $\nu$ be a bounded continuous
  valuation on $Y$.  If $\nu$ is supported on every $V_n$, then $\nu$
  is supported on $\bigcap_{n \in \nat} V_n$.
\end{newlemma}
\begin{proof}
  Without loss of generality, we may assume that $V_0 \supseteq V_1
  \supseteq \cdots \supseteq V_n \cdots$, otherwise we replace each
  $V_n$ by $\bigcap_{i=0}^n V_i$.
  
  We will use a precursor to \cite[Theorem~1.1]{dBGLJL:LCS}: a theorem
  due to Mauricio Alvarez-Manilla \cite[Theorem~3.27]{alvarez00}
  states that on a locally compact sober space, every locally finite
  continuous valuation extends to a (unique $\tau$-smooth) measure on
  the Borel $\sigma$-algebra.  A continuous valuation is \emph{locally
    finite} if and only if every point has an open neighborhood mapped
  to some finite number by the valuation.  This applies to $\nu$,
  which is locally finite, since bounded.  We write the unique
  extension of $\nu$ we have just obtained as $\mu$.
  
  In $Y$, $X$ is a countable intersection of open sets, hence lies in
  the Borel $\sigma$-algebra, so $\mu (U \cap X)$ makes sense for
  every open subset $U$ of $Y$.  We then have
  $\mu (U \cap X) = \mu (U \cap \bigcap_{n \in \nat} V_n) = \mu
  (\bigcap_{n \in \nat} (U \cap V_n)) = \inf_{n \in \nat} \mu (U \cap
  V_n) = \inf_{n \in \nat} \nu (U \cap V_n)$, since bounded measures
  map intersections of descending sequences of measurable sets to the
  infimum of their measures.  Since $\nu$ is supported on $V_n$, and
  $(U \cap V_n) \cap V_n = U \cap V_n$, we have
  $\nu (U \cap V_n) = \nu (U)$ for every $n \in \nat$.  Hence
  $\mu (U \cap X) = \nu (U)$.  This entails that for all open subsets
  $U$ and $U'$ of $Y$ such that $U \cap X = U' \cap X$, $\nu (U)$ and
  $\nu (U')$ are both equal to $\mu (U \cap X) = \mu (U' \cap X)$, so
  $\nu$ is supported on $X$.
\end{proof}

Finally, here is Lemma~3.4 of \cite{JGL-mscs16} in repaired form.  The
conclusion is the same as originally published.  The assumption that
$X$ is locally convex is replaced by $\AN_\bullet$-friendliness
(Definition~\ref{defn:LX:Hfriendly}).

\begin{lemma}
  \label{lemma:AN:rs=id:aux}
  Let $\bullet$ be nothing, ``$\leq 1$'' or ``$1$'', let $X$ be an
  $\AN_\bullet$-friendly space, and let $F$ be an element of
  $\Pred^\bullet_{\AN} (X)$.  For every $h_0 \in \Lform X$, and every
  real number $r > 0$ such that $F (h_0) > r$, there is a
  $G \in \Pred^\bullet_{\Nature} (X)$ with $G \leq F$ and
  $G (h_0) > r$.
\end{lemma}
\begin{proof}
  The cases where $\bullet$ is nothing or ``$\leq 1$'' were dealt with
  in Lemma~\ref{lemma:3.4:easy}.
  
  We now deal with the case where $\bullet$ is ``$1$''.  This decomposes
  into subcases, according to whether $X$ satisfies condition~1, 2
  or~3 of Definition~\ref{defn:LX:Hfriendly}.

  \emph{If condition~1 or~2 is satisfied.}  We let
  $W \eqdef F^{-1} (]1, \infty])$.  We verify the assumptions of
  Lemma~\ref{lemma:Hfriendly:1:V:aux}.  For every $b > 1$ and every
  $h \in W$, we consider an arbitrary element
  $ba \cdot \one + (1-a) h$ of $[b \cdot \one, h]$, with
  $a \in [0, 1]$.  Since $F$ is normalized and positively homogeneous,
  $F (ba \cdot \one + (1-a) \cdot h) = ba + (1-a) F (h)$.  If
  $a \neq 1$, using the fact that $h \in W = F^{-1} (]1, \infty])$,
  this is strictly larger than $ba + (1-a) \geq a + (1-a)=1$;
  otherwise, this is equal to $b$, which is also strictly larger than
  $1$.  In any case, it follows that $ba \cdot \one + (1-a) h$ is in
  $W$.  Hence we may apply Lemma~\ref{lemma:Hfriendly:1:V:aux} (with
  $(1/r) \cdot h_0$ taking the place of $h_0$, so that
  $(1/r) \cdot h_0 \in W = F^{-1} (]1, \infty])$), and we obtain that
  there is a convex Scott-open neighborhood $V$ of $(1/r) \cdot h_0$
  in $\Lform X$ that contains every function $b \cdot \one$ with
  $b > 1$ and that is included in $W$.

  We note that $V$ is proper (namely, not the whole of $\Lform X$),
  since $W$ is itself proper.  Indeed, the constant $0$ map
  $0 \cdot \one$ is not in $V$, because $F (0 \cdot \one) = 0$.  We
  form the upper Minkowski functional $q$ of $V$, which is by
  definition the lower Minkowski function $M_{\Lform X \diff V}$ of
  its complement: by \cite[Proposition 7.6~(b)]{Keimel:topcones2}, $q$
  is a superlinear lower semi-continuous map from $\Lform X$ to
  $\creal$, such that $q^{-1} (]1, \infty]) = V$.  In fact, taking
  upper Minkowski functionals defines an order-isomorphism between the
  collection of proper open convex sets and superlinear lower
  semi-continuous maps.  Since $V \subseteq W = F^{-1} (]1, \infty])$,
  $q \leq F$.

  Keimel's sandwich theorem \cite[Sandwich
  Theorem~8.2]{Keimel:topcones2} states that given any superlinear
  lower semi-continuous map $q \colon C \to \creal$ on a
  semitopological cone $C$, given any sublinear map (not necessarily
  lower semi-continuous) $p \colon C \to \real$ pointwise above $q$,
  there is a linear lower semi-continuous map $G \colon C \to \creal$
  such that $q \leq G \leq p$.  We apply this to $q$ defined above and
  to $p \eqdef F$, and we obtained the corresponding $G$.

  The map $(1/r) \cdot h_0$ is in $V$, so $q ((1/r) \cdot h_0) > 1$;
  then $q (h_0) > r$, and hence $G (h_0) > r$.  We have
  $G (\one) \leq F (\one) = 1$ because $F$ is normalized, so $G$ is
  subnormalized.  Finally, for every $b > 1$, $b\cdot \one$ is in $V$,
  so $q (b \cdot \one) \geq 1$, namely $q (\one) \geq 1/b$ for every
  $b > 1$.  It follows that $q (\one) \geq 1$, so $G (\one) \geq 1$,
  hence $G$ is normalized.  Therefore $G$ is in $\Pred^1_{\Nature} X$.

  \emph{If condition~3 is satisfied.}  Then $X$ embeds as a $G_\delta$
  subspace of a locally compact sober space $Y$.  In particular, $Y$
  is core-compact, so $\Lform Y$ is a locally convex, locally
  convex-compact, sober topological cone, by
  Remark~\ref{rem:Hfriendly:1}.  We equate $X$ with the intersection
  $\bigcap_{n \in \nat} V_n$ of a descending sequence of open subsets
  $V_n$ of $Y$, $n \in \nat$, and we write $i \colon X \to Y$ for the
  inclusion map.

  Let $F' \colon \Lform Y \to \creal$ map every $h' \in \Lform Y$ to
  $F (h' \circ i)$.  
  It is easy to see that $F'$ is a normalized sublinear prevision on
  $Y$.

  Every $h \in \Lform X$ extends to some $h^* \in \Lform Y$, in the
  sense that $h^* \circ i = h$, namely that $h^* (x) = h (x)$ for
  every $x \in X$.  This is due to the fact that $\creal$ is a
  continuous lattice, and to Scott's theorem that the continuous
  lattices are exactly the injective spaces \cite[Theorem
  II-3.8]{GHKLMS:contlatt}; namely, the spaces $R$ such that every
  continuous map from a subspace (here, $X$) of a space (here $Y$) to
  $R$ has a continuous extension to the whole of $Y$.

  We have $F' (h_0^*) = F (h_0^* \circ i) = F (h_0) > r$.  Let also
  $W \eqdef {F'}^{-1} (]1, \infty])$.  We will apply
  Lemma~\ref{lemma:Hfriendly:1:V}, this time with
  $g_n \eqdef \chi_{V_n}$ for every $n \in \nat$.  We note that
  $g_n \circ i = \one$ for every $n \in \nat$, since
  $X \subseteq V_n$.  We check that every function
  $ba \cdot g_n + (1-a) \cdot h$ where $n \in \nat$, $b > 1$,
  $a \in [0, 1]$ and $h \in W$ is in $W$:
  $F' (ba \cdot g_n + (1-a) \cdot h) = F (ba \cdot (g_n \circ i) +
  (1-a) \cdot (h \cdot i)) = F (ba \cdot \one + (1-a) \cdot (h \cdot
  i)) = ba + (1-a) F' (h) > 1$, since $b > 1$ and $F' (h) > 1$.  By
  Lemma~\ref{lemma:Hfriendly:1:V}, therefore, there is a convex
  Scott-open neighborhood $V$ of $(1/r) \cdot h_0^*$ in $\Lform Y$
  that contains every function $b \cdot \chi_{V_n}$ for every
  $n \in \nat$ and $b > 1$, and such that $V \subseteq W$.

  $V$ is proper, since $W$ is proper, because the constant $0$ map is
  not in $W$.  As before, there is a unique superlinear lower
  semi-continuous map $q \colon \Lform Y \to \creal$ such that
  $q^{-1} (]1, \infty]) = V$.  By Keimel's sandwich theorem on
  $\Lform Y$, there is a lower semi-continuous linear map
  $G' \colon \Lform Y \to \creal$ such that $q \leq G' \leq F'$.

  Hence $G'$ is the integration functional of a unique continuous
  valuation $\nu'$ on $Y$; by integration functional we mean that
  $G' (h) = \int_{y \in Y} h (y) \,d\nu'$ for every $h \in \Lform Y$.
  For every $n \in \nat$ and every $b > 1$, we have
  $G' (b \cdot \chi_{V_n}) \geq q (b \cdot \chi_{V_n}) > 1$, since
  $G' \geq q$ and $b \cdot \chi_{V_n} \in V = q^{-1} (]1, \infty])$.
  Hence $\nu' (V_n) = G' (\chi_{V_n}) > 1/b$ for every $b > 1$.  Since
  $G' \leq F'$,
  $G' (\chi_{V_n}) \leq F' (\chi_{V_n}) = F (\chi_{V_n} \circ i) = F
  (\one)=1$, so $\nu' (V_n) \leq 1$.  It follows that $\nu' (V_n)=1$,
  for every $n \in \nat$.  We also note that $\nu' (Y)=1$: by a
  similar argument involving $G' \leq F'$, $\nu' (Y) \leq 1$, and
  $\nu' (Y) \geq \nu' (V_0)=1$.

  For every $n \in \nat$, since $\nu' (V_n)=1$, $\nu$ is supported on
  $V_n$.  We argue as follows.  We start with
  $1 = \nu' (V_n) \leq \nu' (U \cup V_n) \leq \nu' (Y)=1$, so all
  terms in the sequence of inequalities are equal.  In particular,
  $\nu' (U \cup V_n) = \nu' (V_n) = 1$.  We use this with
  $\nu' (U \cup V_n) + \nu' (U \cap V_n) = \nu' (U) + \nu' (V_n)$ and
  we obtain that $\nu' (U \cap V_n) = \nu' (U)$, for every open subset
  $U$ of $Y$.  In particular, for all open subsets $U$ and $U'$ of $Y$
  such that $U \cap V_n = U' \cap V_n$, $\nu' (U) = \nu' (U')$.  Using
  Lemma~\ref{lemma:supp:Gdelta}, then, $\nu'$ is supported on $X$, and is
  therefore equal to the image continuous valuation $i [\nu]$ of a
  unique continuous valuation $\nu$ on $X$.

  Additionally, $\nu (X) = i [\nu] (Y) = \nu' (Y)=1$.  Let $G$ be the
  integration functional of $\nu$ on $X$, namely
  $G (h) \eqdef \int_{x \in X} h (x) \,d\nu$ for every
  $h \in \Lform X$.  For all $a \in \Rp$ and $h \in \Lform X$,
  $G (a \cdot \one + h) = a \,\nu (X) + G (h) = a + G (h)$ for every
  $h \in \Lform X$, so $G$ is normalized.

  Next, $G' (h) = G (h \circ i)$ for every $h \in \Lform Y$.  Indeed,
  $G' (h) = \int_{y \in Y} h (y) \,d\nu' = \int_{y \in Y} h (y) \,d i
  [\nu] = \int_{x \in X} h (i (x)) \,d\nu$ (by the change-of-variable
  formula) $= G (h \circ i)$.  In particular,
  $G' (h_0^*) = G (h_0^* \circ i) = G (h_0)$.  But $G' \geq q$, and
  since $(1/r) \cdot h_0^*$ is in $V = q^{-1} (]1, \infty])$,
  $q ((1/r) \cdot h_0^*) > 1$, so $G ' (h_0^*) \geq q (h_0^*) > r$.
  It follows that $G (h_0) > r$.  Finally, $G' \leq F'$, so for every
  $h \in \Lform X$, $G' (h^* \circ i) \leq F' (h^* \circ i)$, meaning
  that $G (h) \leq F (h)$; whence $G \leq F$.
\end{proof}



}

\begin{lemma}
  \label{lemma:AN:rs=id}
  Let {\updated $\bullet$ be nothing, ``$\leq 1$" or ``$1$", and $X$
    be an $\AN_\bullet$-friendly space}.
  Then
  $r_{\AN} \circ s^\bullet_{\AN}$ is the identity map on
  $\Pred^\bullet_{\AN\;wk} (X)$.
\end{lemma}
\begin{proof}
  We must show that for every $F \in \Pred_{\AN} (X)$, for every
  $h \in \Lform X$,
  $F (h)$ is equal to
  $\sup_{G \in \Pred^\bullet_{\Nature} (X), G \leq F} G (h)$.  The
  inequality
  $F (h) \leq \sup_{G \in \Pred^\bullet_{\Nature\;wk} (X), G \leq F} G
  (h)$ follows directly from Lemma~\ref{lemma:AN:rs=id:aux}, while the
  converse inequality is obvious.
\end{proof}

{\updated The following lemma assumed $\Lform X$ to be locally convex
  in \cite{JGL-mscs16}, but this is not needed.}
\begin{lemma}
  \label{lemma:AN:s:closed}
  Let {\updated $\bullet$ be nothing, ``$\leq 1$" or ``$1$", and $X$
    be a topological space}.
  For every $F \in \Pred^\bullet_{\AN\;wk} (X)$,
  $s^\bullet_{\AN} (F) = \{G \in \Pred^\bullet_{\Nature} (X) \mid G
  \leq F\}$ is a closed subset of $\Pred_{\Nature\;wk}^\bullet (X)$.
\end{lemma}
\begin{proof}
  Consider any $G \not\in s^\bullet_{\AN} (F)$.   There is an $h \in
  \Lform X$
  such that $G (h) > F (h)$.  Then, $[h > r]$ (where $r \eqdef F (h)$)
  is an open neighborhood of $G$ that does not meet
  $s^\bullet_{\AN} (F)$.  So the complement of $s^\bullet_{\AN} (F)$
  is open.
\end{proof}

\begin{lemma}
  \label{lemma:AN:s:def}
  Let {\updated $\bullet$ be nothing, ``$\leq 1$" or ``$1$", and $X$
    be an $\AN_\bullet$-friendly space}.
  For every $F \in \Pred^\bullet_{\AN\;wk} (X)$,
  $s^\bullet_{\AN} (F) = \{G \in \Pred^\bullet_{\Nature} (X) \mid G
  \leq F\}$ is an element of $\HV (\Pred_{\Nature\;wk}^\bullet (X))$.
\end{lemma}
\begin{proof}
  It is closed by Lemma~\ref{lemma:AN:s:closed}.  When $\bullet$ is
  the empty superscript or ``$\leq 1$'', $s^\bullet_{\AN} (F)$ is
  non-empty since it contains the zero prevision.  When $\bullet$ is
  ``$1$'', non-emptiness follows from Lemma~\ref{lemma:AN:rs=id:aux}
  with $h \eqdef \one$,
  and taking $r=0$ for example. 
\end{proof}

The fact that $s^\bullet_{\AN}$ is continuous is the most complicated
result of this section.  One of the needed ingredients is von
Neumann's original minimax theorem \cite{vN:minimax}.  That theorem
was vastly generalized since then, and in numerous ways, but the
original form will be sufficient to us.
\begin{lemma}[Von Neumann's Minimax]
  \label{lemma:minimax}
  For each $n \in \nat$, let $\Delta_n$ be the set of $n$-tuples of
  non-negative real numbers $(a_1, a_2, \cdots, a_n)$ such that
  $\sum_{i=1}^n a_i=1$.  For every $n \times m$ matrix $M$ with real
  entries,
  \[
  \min_{\vec\alpha \in \Delta_m} \max_{\vec\beta \in \Delta_n}
  {\vec\beta}^t M \vec\alpha
  = \max_{\vec\beta \in \Delta_n} \min_{\vec\alpha \in \Delta_m} {\vec\beta}^t M \vec\alpha.
  \]
\end{lemma}
Note that an implicit fact in that lemma, hidden in the notation
$\min$, $\max$, is that the suprema over $\vec\beta$ and the infima
over $\vec\alpha$ are attained.  This is a consequence of the fact
that $\Delta_m$ and $\Delta_n$ are compact, and that multiplication by
$M$ is continuous.  For that, it is important that $M$ has real-valued
entries, and in particular, not $+\infty$.

\begin{lemma}
  \label{lemma:AN:s:cont}
  Let {\updated $\bullet$ be nothing, ``$\leq 1$" or ``$1$", and $X$
    be an $\AN_\bullet$-friendly space}.
  Then $s^\bullet_{\AN}$ is continuous from
  $\Pred^\bullet_{\AN\;wk} (X)$ to
  $\HV (\Pred^\bullet_{\Nature\;wk} (X))$.
\end{lemma}
\begin{proof}
  Among the subbases of the weak topology on $\Pred^\bullet_{\Nature}
  (X)$, one is given by subsets of the form $[h > b]_\Nature$, $h \in
  \Lform X$,
  $b \in \Rp$.  Since
  $[h > 0]_\Nature = \bigcup_{r > 0} [h > r]_\Nature$, we may restrict
  to $b \neq 0$.  When $b \neq 0$,
  $[h > b]_\Nature = [h/b > 1]_\Nature$, so another subbase is given
  by the subsets of the form $[h > 1]_\Nature$.  Finally, since $h$ is
  the directed supremum of the maps $\min (h, r)$, $r \in \Rp$, we can
  even restrict $h$ to be bounded.

  Since $\Diamond$ commutes with unions, a subbase of the topology of
  $\HV (\Pred^\bullet_{\Nature\;wk} (X))$ is then given by the subsets
  of the form $\Diamond W$, where $W$ is a finite intersection
  $\bigcap_{i=1}^m [h_i > 1]_{\Nature}$, where each $h_i$ is
  continuous and bounded.  Moreover, we may require $m > 0$.

  We must show that the inverse image of $\Diamond W$ by
  $s^\bullet_{\AN}$, where $W$ is as above, is open in
  $\Pred^\bullet_{\AN\;wk} (X)$.  Let $F$ be an arbitrary element of
  ${s^\bullet_{\AN}}^{-1} (\Diamond W)$.  By definition, there is an
  element $G \in \Pred^\bullet_{\Nature\;wk} (X)$ such that $G \leq F$
  and $G \in \bigcap_{i=1}^m [h_i > 1]_{\Nature}$.

  Fix a positive real number $\epsilon$ such that
  $G (h_i) > 1+\epsilon$ for every $i$, $1\leq i \leq m$.  We claim
  that we can find a finite set $A$ of $m$-tuples of non-negative real
  numbers $\vec a \eqdef (a_1, a_2, \cdots, a_m)$ such that
  $\bigcap_{\vec a \in A} [\sum_{i=1}^m a_i h_i > (1+\epsilon)
  \sum_{i=1}^m a_i]_\AN$ is an open neighborhood of $F$ included in
  ${s^\bullet_{\AN}}^{-1} (\Diamond W)$.

  To this end, we shall define $A$ as the set of $m$-tuples of
  non-negative real numbers $\vec a \eqdef (a_1, a_2, \cdots, a_m)$ such
  that $0 < \sum_{i=1}^m a_i \leq 1$ and each $a_i$ is an integer
  multiple of $1/N$, for some fixed, large enough natural number $N$.
  Taking $N$ so that $\frac m N < \frac \epsilon {1+\epsilon}$ will be
  enough for our purposes.

  The fact that $F$ is in $\bigcap_{\vec a \in A} [\sum_{i=1}^m a_i
  h_i > (1+\epsilon) \sum_{i=1}^m a_i]_\AN$ is obvious.  For every
  $\vec a \in A$, $F (\sum_{i=1}^m a_i h_i) \geq G (\sum_{i=1}^m a_i
  h_i) = \sum_{i=1}^m a_i G (h_i)$, and this is larger than or equal
  to $(1+\epsilon) \sum_{i=1}^m a_i$.  To show that it is strictly
  larger, recall that some $a_i$ is non-zero, since $0 < \sum_{i=1}^m
  a_i$.

  To show that $\bigcap_{\vec a \in A} [\sum_{i=1}^m a_i h_i >
  (1+\epsilon) \sum_{i=1}^m a_i]_\AN$ is included in
  ${s^\bullet_{\AN}}^{-1} (\Diamond W)$ is more technical.  The main
  observation is the following:

  \textbf{Fact A.} For every positively homogeneous map $\Phi \colon
  \Lform X \to {\creal}_\sigma$
  such that
  $\Phi (\sum_{i=1}^m a_i h_i) > (1+\epsilon) \sum_{i=1}^m a_i$ for
  every $\vec a \in A$, it holds that
  $\Phi (\sum_{i=1}^m \alpha_i h_i) > 1$ for every
  $\vec\alpha \in \Delta_m$.

  This fact is proved by elementary computation.  Fix $\vec\alpha \in
  \Delta_m$, and let $a_i \eqdef \frac 1 N \lfloor N \alpha_i \rfloor$ be
  the largest multiple of $1/N$ below $\alpha_i$, for each $i$.
  Notice that $\frac 1 {1+\epsilon} < \sum_{i=1}^m a_i \leq 1$; the
  inequality on the left follows from the consideration that $a_i >
  \alpha_i-1/N$, and $\frac m N < 1 - \frac 1 {1+\epsilon} = \frac
  \epsilon {1+\epsilon}$, remembering that $\sum_{i=1}^m \alpha_i=1$.
  In particular, $\vec a = (a_1, a_2, \cdots, a_m)$ is in $A$, so
  $\Phi (\sum_{i=1}^m a_i h_i) > (1+\epsilon) \sum_{i=1}^m a_i$.
  Since $\alpha_i \geq a_i$ for each $i$, $\Phi (\sum_{i=1}^m \alpha_i
  h_i) \geq \Phi (\sum_{i=1}^m a_i h_i) > (1+\epsilon) \sum_{i=1}^m
  a_i$, and we have just seen that the latter is strictly greater than
  $1$.

  Now consider any element $F'$ of $\bigcap_{\vec a \in A}
  [\sum_{i=1}^m a_i h_i > (1+\epsilon) \sum_{i=1}^m a_i]_\AN$.  By
  Lemma~\ref{lemma:AN:rs=id}, $F'$ is a pointwise supremum of elements
  of $\Pred^\bullet_{\Nature\;wk} (X)$.  Write this supremum as a
  directed supremum of finite suprema, and observe that $\bigcap_{\vec
    a \in A} [\sum_{i=1}^m a_i h_i > (1+\epsilon) \sum_{i=1}^m
  a_i]_\AN$ is Scott open.  As a consequence, it contains one of the
  finite suprema, viz., there are finitely many elements $G'_1, G'_2,
  \cdots, G'_n$ of $\Pred^\bullet_{\Nature\;wk} (X)$ below $F'$ such
  that, for every $\vec a \in A$, there is a $j$, $1\leq j\leq n$,
  such that $\sum_{i=1}^m a_i G'_j (h_i) > (1+\epsilon) \sum_{i=1}^m
  a_i$.  We can even take $G'_j$ bounded in the sense that $G'_j (\one) <
  +\infty$:
  this is clear when $\bullet$ is ``$\leq 1$'' or ``$1$'', otherwise
  we use \cite[Theorem~4.2]{Heckmann:space:val}, which says that every
  linear prevision on $X$ is a directed supremum of bounded linear
  previsions, allowing us to replace each $G'_j$ by a bounded linear
  prevision whose value on $h_i$ is close enough to $G'_j (h_i)$.

  For every $\vec a \in A$, there is a $j$ such that
  $\sum_{i=1}^m a_i G'_j (h_i) > (1+\epsilon) \sum_{i=1}^m a_i$, so
  trivially, there is a $\vec\beta \in \Delta_n$ such that
  $\sum_{\substack{1\leq i\leq m\\1\leq j\leq n}} a_i \beta_j G'_j
  (h_i) > (1+\epsilon) \sum_{i=1}^m a_i$.  Applying Fact~A to
  $\Phi \eqdef \sup_{\vec\beta \in \Delta_n} \sum_{j=1}^n \beta_j
  G'_j$, we obtain that for every $\vec\alpha \in \Delta_m$,
  $\sup_{\vec\beta \in \Delta_n} \sum_{j=1}^n \beta_j G'_j
  (\sum_{i=1}^m \alpha_i h_i) > 1$, so
  $\sum_{\substack{1\leq i\leq m\\1\leq j\leq n}} \alpha_i \beta_j
  G'_j (h_i) > 1$ for some $\vec\beta \in \Delta_n$.  Since $G'_j$ is
  bounded, and $h_i$ is bounded too, $G'_j (h_i) < +\infty$, so
  $M \eqdef {(G'_j (h_i))}_{\substack{1\leq j\leq n\\1\leq i\leq m}}$
  is a matrix of real numbers.  Rephrasing what we have just obtained,
  for every $\vec \alpha \in \Delta_m$, there is a
  $\vec\beta \in \Delta_n$ such that ${\vec\beta}^t M \vec\alpha > 1$.
  In particular (and using the fact that infima are attained)
  $\min_{\vec\alpha \in \Delta_m} \max_{\vec\beta \in \Delta_n}
  {\vec\beta}^t M \vec\alpha > 1$.  By von Neumann's minimax theorem
  (Lemma~\ref{lemma:minimax}),
  $\max_{\vec\beta \in \Delta_n} \min_{\vec\alpha \in \Delta_m}
  {\vec\beta}^t M \vec\alpha > 1$.  Therefore, there is a tuple
  $\vec\beta \in \Delta_n$ such that, for every
  $\vec\alpha \in \Delta_m$,
  $\sum_{\substack{1\leq i\leq m\\1\leq j\leq n}} \alpha_i \beta_j
  G'_j (h_i) > 1$.  Take $G' \eqdef \sum_{j=1}^n \beta_j G'_j$.  Since
  $\sum_{j=1}^n \beta_j=1$, $G'$ is in $\Pred^\bullet_{\Nature} (X)$,
  and $G' \leq F'$.  Also, we have just proved that
  $\sum_{i=1}^m \alpha_i G' (h_i) > 1$ for every
  $\vec\alpha \in \Delta_m$, in particular, $G'$ is in
  $\bigcap_{i=1}^m [h_i > 1]_\Nature$.  Therefore $F'$ is, indeed, in
  ${s^\bullet_{\AN}}^{-1} (\Diamond W)$.
\end{proof}

\begin{lemma}
  \label{lemma:AN:sr}
  Let $X$ be a topological space.  For every $C \in \HV
  (\Pred^\bullet_{\Nature\;wk} (X))$, $C \subseteq s^\bullet_{\AN}
  (r_\AN (C))$.
\end{lemma}
\begin{proof}
  That amounts to checking that for every $G \in C$, $G (h) \leq
  \sup_{G' \in C} G' (h)$, which is obvious.
\end{proof}

We sum up these results in the following proposition.
\begin{proposition}
  \label{prop:AN:weak}
  Let $\bullet$ be the empty superscript, ``$\leq 1$'', or ``$1$''.

  Let $X$ be a topological space.  Then $r_{\AN}$ is a continuous map
  from $\HV (\Pred^\bullet_{\Nature\;wk} (X))$ to
  $\Pred^\bullet_{\AN\;wk} (X)$, and $s^\bullet_{\AN} \circ r_\AN$
  is above the identity.

  {\updated If $X$ is $\AN_\bullet$-friendly,}
  then $s^\bullet_{\AN}$ is a continuous map from
  $\Pred^\bullet_{\AN\;wk} (X)$ to
  $\HV (\Pred^\bullet_{\Nature\;wk} (X))$ such that
  $r_{\AN} \circ s^\bullet_{\AN}$ equals the identity.
\end{proposition}

Recall that a coembedding-coprojection pair is a pair of continuous
maps $s$, $r$, such that $r \circ s = \identity \relax$ and $\identity
\relax \leq s \circ r$.
\begin{corollary}[$\HV (\Pred^\bullet_{\Nature\;wk} (X))$ coprojects onto $\Pred^\bullet_{\AN\;wk} (X)$]
  \label{corl:AN:weak}
  Let $\bullet$ be the empty superscript, ``$\leq 1$'', or ``$1$''.
  Let {\updated $X$ be an $\AN_\bullet$-friendly space}, for example,
  a locally compact space (or, more generally, a core-compact space),
  {\updated or an LCS-complete space}.

  Then $s^\bullet_{\AN}$ and $r_{\AN}$ together define an
  coembedding-coprojection pair of $\HV (\Pred^\bullet_{\Nature\;wk}
  (X))$ onto $\Pred^\bullet_{\AN\;wk} (X)$.
\end{corollary}

\subsection{The Retraction in the Demonic Case}
\label{sec:retr-demon-case}

In the demonic cases, we have defined $r_{\DN} (Q) (h)$ as $\inf_{G
  \in Q} G (h)$.  We shall be interested in cases where $Q$ is in a
Smyth powerdomain.  We can then say more:
\begin{lemma}
  \label{lemma:rDN}
  The map $r_{\DN}$, once restricted to the Smyth powerdomain of some
  space of linear previsions, is defined by
  $r_{\DN} (Q) (h) \eqdef \min_{G \in Q} G (h)$.
\end{lemma}
\begin{proof}
  The evaluation map $G \mapsto G (h)$ is continuous from the given
  space of linear previsions to ${\creal}_\sigma$, by the very
  definition of the weak topology.  Such continuous maps are called
  lower semi-continuous real maps in the mathematical literature, and
  it is standard that lower semi-continuous real maps attain their
  infimum on every compact set.
\end{proof}

We wish to show that $r_{\DN}$, $s^\bullet_{\DN}$
forms a retraction.  We again progress through a series of lemmata.
In each, $\bullet$ may be the empty subscript, ``$\leq 1$'', or
``$1$''.

We use a similar proof as sketched in \cite{Gou-fossacs08a} to show
that $r_{\DN} \circ s^\bullet_{\DN}$ is the identity map (for
whichever superscripts $\bullet$).  An additional trick allows us to
dispense with an assumption of stable compactness, in
Lemma~\ref{lemma:DN:s:compact} below.

\begin{lemma}
  \label{lemma:DN:r:def}
  Let $X$ be a topological space.  Then $r_{\DN}$ is a map from $\SV
  (\Pred^\bullet_{\Nature\;wk} (X))$ to $\Pred^\bullet_{\DN\;wk} (X)$.
\end{lemma}
\begin{proof}
  Let $Q \in \SV (\Pred^\bullet_{\Nature\;wk} (X))$.  Writing $F (h)$
  for $r_{\DN} (Q) (h) = \min_{G \in Q} G (h)$, we must show that $F$
  is a Smyth prevision.  Clearly, $F$ is positively homogeneous,
  monotonic and superlinear.  It is also subnormalized in case
  $\bullet$ is ``$\leq 1$'' and normalized when $\bullet$ is ``$1$''.
  It remains to show that $F$ is Scott-continuous.  This follows from
  the fact that the pointwise infimum of a compact family of lower
  semi-continuous functions is lower semi-continuous
  \cite{Keimel:inf}.  For completeness, here is a short argument.  For
  any directed family ${(h_i)}_{i \in I}$ with supremum $h$, we must
  show that $F (h) \leq \sup_{i \in I} F (h_i)$, since the other
  inequality stems from monotonicity.  If that were not the case, let
  $b \eqdef \sup_{i \in I} F (h_i)$, so that $F (h) > b$.  The open
  subsets $[h_i > b]$, $i \in I$, form a directed open cover of $Q$,
  since for every $G \in Q$,
  $\sup_{i \in I} G (h_i) = G (h) \geq F (h) > b$.  By compactness,
  $Q \subseteq [h_i > b]$ for some $i$, whence
  $F (h_i) = \min_{G \in Q} G (h_i) > b$, contradiction.
\end{proof}

\begin{lemma}
  \label{lemma:DN:r:cont}
  Let $X$ be a topological space.  Then $r_{\DN}$ is a continuous map
  from $\SV (\Pred^\bullet_{\Nature\;wk} (X))$ to
  $\Pred^\bullet_{\DN\;wk} (X)$.
\end{lemma}
\begin{proof}
  For clarity, write $[h > b]_{\DN}$ or $[h > b]_\Nature$ for the
  subbasic open subset $[h > b]$, depending whether in a space of
  Smyth previsions or of linear previsions.

  The inverse image of the subbasic open $[h > b]_{\DN} \subseteq
  \Pred^\bullet_{\DN\;wk} (X)$ by $r_{\DN}$ is $\{Q \in \SV
  (\Pred^\bullet_{\Nature\;wk} (X)) \mid \min_{G \in Q} G (h) > b\} =
  \{Q \in \SV (\Pred^\bullet_{\Nature\;wk} (X)) \mid \forall G \in Q
  \cdot G (h) > b\}$ (that this is a $\min$, and not just an $\inf$,
  is important here) $= \Box [h > b]_\Nature$.  This is open in $\SV
  (\Pred^\bullet_{\Nature} (X))$.  Therefore $r_{\DN}$ is continuous.
\end{proof}

To show that $s^\bullet_{\DN}$ maps Smyth previsions to compact
saturated subsets of linear previsions, we shall use the following
piece of logic.

Let $A$ be a fixed stably compact space.  (This will be
${\creal}_\sigma$ in our case.)  Let $T$ be a set.  A
\emph{patch-continuous inequality} on $T, A$ is any formula of the
form:
\[ 
f (\_ (t_1), \ldots, \_ (t_m)) \mathrel{\dot\leq} g (\_ (t'_1), \ldots, \_
(t'_n)),
\]
where $f$ and $g$ are patch-continuous maps from $A^m$ to $A$ and from
$A^n$ to $A$ respectively, and $t_1$, \ldots, $t_m$, $t'_1$, \ldots,
$t'_n$ are $m+n$ fixed elements of $T$.  $E$ {\em holds\/} at $\alpha
\colon T \to A$ iff $f (\alpha (t_1), \ldots, \alpha (t_m)) \leq g
(\alpha (t'_1), \ldots, \alpha (t'_n))$, where $\leq$ is the
specialization quasi-ordering of $A$.  A \emph{patch-continuous
  system} $\Sigma$ on $T, A$ is a set of patch-continuous inequalities
on $T, A$, and $\Sigma$ holds at $\alpha \colon T \to A$ if every
element of $\Sigma$ holds at $\alpha$.  By convention, we shall allow
equations $a \doteq b$ in such systems, and agree that they stand for
the pair of inequalities $a \dot\leq b$ and $b \dot\leq a$.  Then the
set $[\Sigma]$ of maps $\alpha \colon T \to A$ such that $\Sigma$
holds at $\alpha$ is patch-closed in $A^T$, and in particular stably
compact. This is \cite[Proposition~5.5]{JGL-mscs09}, but also a fairly
simple exercise.

\begin{lemma}
  \label{lemma:DN:s:compact}
  Let $X$ be a topological space.  For every $F \in
  \Pred^\bullet_{\DN\;wk} (X)$, $s^\bullet_{\DN} (F)$ is a compact
  saturated subset of $\Pred^\bullet_{\Nature\;wk} (X)$.
\end{lemma}
\begin{proof}
  Recall that $s^\bullet_{\DN}$ maps every $F \in \Pred^\bullet_{\DN}
  (X)$ to the set of all $G \in \Pred^\bullet_{\Nature} (X)$ such that
  $G \geq F$.

  Let $Y^\bullet$ be the space of all linear, monotonic maps from
  $C \eqdef \Lform X$
  to $\creal$ that are subnormalized if $\bullet$ is ``$\leq 1$'', and
  normalized if $\bullet$ is ``$1$''.  Compared to
  $\Pred^\bullet_{\Nature} (X)$, we are no longer requiring
  Scott-continuity.  Equip $Y^\bullet$ with the weak topology, which
  is again generated by subbasic open sets $[h > b]_Y$, defined as
  $\{\alpha \in Y \mid \alpha (h) > b\}$.  $Y^\bullet$ is then a
  subspace of the space ${\creal}_\sigma^C$ with its product topology,
  and $\Pred^\bullet_{\Nature} (X)$ is a subspace of $Y^\bullet$.

  Now $Y^\bullet$ is $[\Sigma^\bullet]$, where $\Sigma^\bullet$ is the
  set of (polynomial, hence patch-continuous) inequalities:
  \begin{itemize}
  \item $\_ (a h) \doteq a \; \_ (h)$, for all $a \in \Rp$, $h \in
    C$ (positive homogeneity);
  \item $\_ (h+h') \doteq \_ (h) + \_ (h')$ for all $h, h' \in C$ (additivity);
  \item $\_ (h) \dot\leq \_ (h')$ for all $h, h' \in C$ with $h \leq h'$
    (monotonicity);
  \item if $\bullet$ is ``$\leq 1$'', $\_ (a+h) \mathrel{\dot\leq} a +
    \_ (h)$ for all $a \in \Rp$, $h \in C$;
  \item if $\bullet$ is ``$1$'', $\_ (a+h) \doteq a + \_ (h)$ for all $a
    \in \Rp$, $h \in C$.
  \end{itemize}
  So $Y^\bullet$ is a patch-closed subset of ${\creal}_\sigma^C$, and a
  stably compact space.

  Consider the set $s (F)$ of all $\alpha \in Y^\bullet$ such that
  $\alpha \geq F$, meaning $\alpha (h) \geq F (h)$ for every $h \in
  C$.  This is again patch-closed in ${\creal}_\sigma^C$ (consider
  $\Sigma^\bullet$ plus the inequalities $F (h) \dot\leq \_ (h)$),
  hence also in $Y^\bullet$.  Note that $s (F)$ is almost
  $s^\bullet_{\DN} (F)$: the latter is the subset of those elements of
  $s (F)$ that are Scott-continuous maps.

  At this point, the natural next move would be to show that
  $s^\bullet_{\DN} (F)$ arises as a retract of $s (F)$, and conclude
  that $s (F)$ is (stably) compact, using the fact that every retract
  of a stably compact space is stably compact.  This idea was
  pioneered by \cite{Jung:scs:prob}, and taken again in
  \cite{JGL-mscs09}.  But this requires $X$ to be stably compact, and
  does not use the fact that $F$ is superlinear and Scott-continuous.
  We use a different argument.

  Note that $s (F)$ is upward closed.  But the patch-closed upward
  closed subsets of a stably compact space are exactly its compact
  saturated subsets \cite[Theorem~VI.6.18~(3)]{GHKLMS:contlatt}, so $s
  (F)$ is compact saturated in $Y^\bullet$.

  To show that $s^\bullet_{\DN} (F)$ is compact saturated in
  $\Pred^\bullet_{\DN} (X)$, we shall appeal to Alexander's Subbase
  Lemma \cite[Theorem~5.6]{Kelley:topology}, which states that in a
  space $X$ with subbase $\mathcal A$, a subset $K$ is compact if and
  only if one can extract a finite subcover from every cover of $K$
  consisting of elements of $\mathcal A$.  In our case, assume
  $s^\bullet_{\DN} (F)$ is included in a union of open subsets
  $\bigcup_{i \in I} [h_i > b_i]$.  We wish to show that there is a
  finite subset $J$ of $I$ such that $s^\bullet_{\DN} (F) \subseteq
  \bigcup_{i \in J} [h_i > b_i]$.  We do it by contraposition: we
  assume that $s^\bullet_{\DN} (F) \not\subseteq \bigcup_{i \in J}
  [h_i > b_i]$ for any $J$, namely, we assume that for every finite
  subset $J$ of $I$, there is a $G_J \in s^\bullet_{\DN} (F)$ such
  that $G_J (h_i) \leq b_i$ for every $i \in J$; and we build an
  element $G$ of $s^\bullet_{\DN} (F)$ such that $G (h_i) \leq b_i$
  for every $i \in I$.  Note that, for every finite $J$, $F \leq G_J$,
  so $G_J$ is in $s (F)$.

  We claim that there is an $\alpha \in s (F)$ such that $\alpha (h_i)
  \leq b_i$ for every $i \in I$.  Otherwise, every $\alpha \in s (F)$
  would be in some $[h_i > b_i]_Y$, $i \in I$.  Since $s (F)$ is
  compact in $Y^\bullet$, there would be a finite subset $J$ of $I$
  such that $s (F) \subseteq \bigcup_{i \in J} [h_i > b_i]_Y$.  In
  particular, $G_J \in \bigcup_{i \in J} [h_i > b_i]_Y$, contradicting
  the fact that $G_J (h_i) \leq b_i$ for every $i \in J$.

  Since $\alpha$ is linear, it is in particular sublinear.  Since
  $\alpha \in s (F)$, we have $F \leq \alpha$.  So we can apply
  Keimel's Sandwich Theorem on the semitopological cone
  $C \eqdef \Lform X$, with the Scott topology:
  there is a \emph{continuous} linear map $G$ such that
  $F \leq G \leq \alpha$.  (Yes, we are using Keimel's Sandwich
  Theorem only to buy continuity.  It is crucial that the larger map,
  $\alpha$ here, need not be continuous to apply this theorem.)  When
  $\bullet$ is ``$\leq 1$'', $\alpha (\one) \leq 1$
  implies that $G$ is subnormalized, and when $\bullet$ is ``$1$'',
  this together with $F (\one)=1$
  implies that $G$ is normalized.  So $G$ is in
  $\Pred^\bullet_{\Nature\;wk} (X)$.  Also, since $F \leq G$, $G$ is
  in $s^\bullet_{\DN} (F)$.  Finally, $G (h_i) \leq b_i$ for every
  $i \in I$, because $G \leq \alpha$.
\end{proof}

The most difficult part now is the following lemma, which will be used
in particular to show that $s^\bullet_{\DN}$ is well defined.
\begin{lemma}
  \label{lemma:DN:s:*}
  Let $X$ be a topological space, $F \in \Pred^\bullet_{\DN} (X)$, and
  $h_0 \in \Lform X$.
  Then there is a $G \in \Pred^\bullet_{\Nature\;wk} (X)$ such that
  $F \leq G$ and $F (h_0) = G (h_0)$.
\end{lemma}
\begin{proof}
  If $F$ is the constant $0$ prevision, this is clear.  Indeed, first,
  $\bullet$ cannot be ``$1$'' in this case (since we must have $F
  (\one)=1$
  for example).  Then we can take $G$ to be the constant $0$ prevision
  again.  So assume $F$ is not constantly $0$.

  Write $C$ for $\Lform X$.
  Let $c$ be the smallest non-negative real number such that
  $F (a \cdot \one+h) \leq ac + F (h)$
  for every $a \in \Rp$ and $h \in C$, if one exists, and $+\infty$
  otherwise.  Notice that if $\bullet$ is ``$\leq 1$'', then
  $c \leq 1$.  If $\bullet$ is ``$1$'', then $c=1$, as one sees by
  taking $a \eqdef 1$, $h \eqdef 0$.  We also note that $c > 0$, in
  any case.  Indeed, if $c=0$, then we would have
  $F (a \cdot \one +h) \leq F (h)$, hence
  $F (a \cdot \one +h) = F (h)$ for all $a$, $h$, in particular $F (a)
  = 0$ for
  every $a \in \Rp$ (taking $h \eqdef 0$).  Since $F$ is Scott-continuous,
  for every $h \in \Lform X$,
  $F (h) = \sup_{a \in \Rp} F (\min (h, a)) \leq \sup_{a \in \Rp} F
  (a) = 0$, contradicting the fact that $F$ is not constantly $0$.

  For every real $\lambda > 0$, let $A_\lambda$ be the closed convex
  hull of the pair of points $\{\lambda/c \cdot \one, h_0\}$ in $C$;
  $\lambda/c \cdot \one$ is the constant function with value
  $\lambda/c$.
  By \cite[Lemma~4.10~$(a)$]{Keimel:topcones2}, $A_\lambda$ is the
  closure of the convex hull
  $H_\lambda \eqdef \{a \lambda/c \cdot \one + (1-a) h_0 \mid a \in
  [0,
  1]\}$ of $\{\lambda/c \cdot \one, h_0\}$.  Define
  $p_\lambda \colon C \to \creal$ by
  $p_\lambda (h) \eqdef \lambda M_{A_\lambda} (h)$.  (Notice that this
  makes sense only for $\lambda > 0$: this is undefined when
  $\lambda=0$ if $M_{A_\lambda} (h)=+\infty$.  As before,
  $M_{A_\lambda}$ is the lower Minkowski functional, see (\ref{eq:MA})
  in the Cones subsection of Section~\ref{sec:preliminaries}.)  Since
  $A_\lambda$ is non-empty, closed and convex, $p_\lambda$ is
  continuous and sublinear for every $\lambda > 0$.  Let
  $p (h) \eqdef \inf_{\lambda > F (h_0)} p_\lambda (h)$.  It may fail to be
  continuous, but we don't need $p$ to be continuous to apply Keimel's
  Sandwich Theorem.

  We shall see below that $p$ is sublinear.  (Beware that an infimum
  of sublinear maps is not in general sublinear.)  Before we prove
  this, we notice that whenever $0 < \lambda \leq \mu$: $(*)$ for
  every $h \in C$, for every $r > 0$ such that $(1/r) \cdot h \in
  H_\mu$, there is an $r' > 0$ such that $(1/r') \cdot h \in
  H_\lambda$ and $r' \lambda \leq r \mu$.  Indeed, by assumption
  $(1/r) \cdot h = a \mu / c + (1-a) h_0$ for some $a \in [0, 1]$,
  i.e., $h = ra \mu/c + r (1-a) h_0 $.  Let $r' \eqdef r a \frac \mu
  \lambda + r (1-a)$.  We have $r' \lambda = r a \mu + r (1-a) \lambda
  \leq r \mu$.  Moreover, letting $a' \eqdef \frac {ra \mu} {r' \lambda}$,
  we check that $1-a' = \frac {r (1-a)} {r'}$, so that $h = r' a'
  \lambda/c \cdot \one + r' (1-a') h_0$,
  and in particular
  $(1/r') \cdot h = a' \lambda/c \cdot \one + (1-a') h_0$ is in
  $H_\lambda$.  This finishes the proof of $(*)$.  In turn, $(*)$
  implies that for every $r > 0$, for every $h \in C$ such that
  $(1/r) \cdot h \in H_\mu$,
  \begin{eqnarray*}
    p_\lambda (h)  = \lambda M_{A_\lambda} (h)
    & = & \lambda \inf \{r' > 0 \mid (1/r') \cdot h \in A_\lambda\} \\
    & \leq & \lambda \inf \{r' > 0 \mid (1/r') \cdot h \in H_\lambda\}
    \quad \text{(since $H_\lambda \subseteq A_\lambda$)} \\
    & \leq & r \mu \quad \text{(since we can pick $r'$ so that
      $r'\lambda \leq r\mu$).}
  \end{eqnarray*}
  So $H_\mu$ is included in the set of those maps $(1/r) \cdot h$ such
  that $p_\lambda (h) \leq r \mu$, i.e., in $r \cdot p_\lambda^{-1}
  ([0, r\mu]) = p_\lambda^{-1} [0, \mu]$.  Since the latter is closed
  ($p_\lambda$ is continuous), $A_\mu$ is also included in this set.
  So, working backwards, we obtain that for every $h \in C$, for every
  $r > 0$ such that $(1/r) \cdot h \in A_\mu$, $p_\lambda (h) \leq r
  \mu$.  Taking infima over $r$, $p_\lambda (h) \leq p_\mu (h)$.  We
  have proved that $p_\lambda$ was a monotonic function of $\lambda$.

  It follows that the family ${(p_\lambda)}_{\lambda > F (h_0)}$ is a
  chain.  In particular, ${(p_\lambda)}_{\lambda > F (h_0)}$ is
  filtered, i.e., directed in the opposite ordering $\geq$.  Since
  addition commutes with filtered infima, $p (h) + p (h') =
  \inf_{\lambda > F (h_0)} [p_\lambda (h) + p_\lambda (h')] \geq
  \inf_{\lambda > F (h_0)} p_\lambda (h+h') = p (h+h')$, using the
  fact that $p_\lambda$ is sublinear.  So $p$ is sub-additive, and it
  follows easily that it is sublinear.

  For every $\lambda > 0$, since $\lambda/c \cdot \one$ is in $A_\lambda$,
  $M_{A_\lambda} (\lambda/c \cdot \one) \leq 1$,
  so $p_\lambda (\lambda/c \cdot \one) \leq \lambda$.
  Since $p_\lambda$ is positively homogeneous and $\lambda > 0$,
  $p_\lambda (\one) \leq c$.
  It follows that $p (\one) \leq c$,
  a fact we shall need later.

  Since $h_0$ is in $A_\lambda$, $M_{A_\lambda} (h_0) \leq 1$, so
  $p_\lambda (h_0) \leq \lambda$.  Taking infs over $\lambda > F
  (h_0)$, it follows that $p (h_0) \leq F (h_0)$, another fact we
  shall need later.

  For every $\lambda > F (h_0)$, for all $r >0$ and $h \in C$ such
  that $(1/r) \cdot h \in H_\lambda$, write $(1/r) \cdot h$ as $a
  \lambda/c \cdot \one + (1-a) h_0$, $a \in [0, 1]$.
  If $c \neq +\infty$, then
  $F (h) = r F (a\lambda/c \cdot \one + (1-a) h_0)
  \leq r (a \lambda + (1-a) F (h_0))$ (by definition of $c$)
  $\leq r (a \lambda + (1-a) \lambda) = \lambda r$.  If $c = +\infty$,
  then $F (h) = r F ((1-a) h_0) \leq r F (h_0) \leq \lambda r$.
  Taking infs over $r$,
  $F (h) \leq \lambda M_{A_\lambda} (h) = p_\lambda (h)$.  So, taking
  infs over $\lambda > F (h_0)$, $F (h) \leq p (h)$.  So $F \leq p$.

  So we can apply Keimel's Sandwich Theorem: there is a continuous
  linear map $G$ such that $F \leq G \leq p$.  Also, $G (\one) \leq p
  (\one)$, and we have seen that $p (\one) \leq c$, and that $c \leq 1$
  whenever $\bullet$ is ``$\leq 1$'' or ``$1$'': so, in these cases,
  $G (\one) \leq 1$, implying that $G$ is subnormalized.  If $\bullet$ is
  ``$1$'', additionally, $F (\one)=1$, so $G (\one)=1$ and $G$ is
  normalized.  In any case, $G$ is in $\Pred^\bullet_{\Nature\;wk}
  (X)$.

  We have also seen that $p (h_0) \leq F (h_0)$.  Since also $F (h_0)
  \leq G (h_0) \leq p (h_0)$, $G (h_0) = F (h_0)$, and we are done.
\end{proof}

\begin{lemma}
  \label{lemma:DN:s:def}
  Let $X$ be a topological space.  Then $s^\bullet_{\DN}$ is a map
  from $\Pred^\bullet_{\DN} (X)$ to $\SV (\Pred^\bullet_{\Nature\;wk}
  (X))$.
\end{lemma}
\begin{proof}
  For every $F \in \Pred^\bullet_{\DN} (X)$, $s^\bullet_{\DN} (F)$ is
  compact saturated by Lemma~\ref{lemma:DN:s:compact}.
  Lemma~\ref{lemma:DN:s:*} implies that it is non-empty.
\end{proof}

\begin{lemma}
  \label{lemma:DN:rs=id}
  Let $X$ be a topological space.  Then $r_{\DN} \circ
  s^\bullet_{\DN}$ is the identity map on $\Pred^\bullet_{\DN\;wk}
  (X)$.
\end{lemma}
\begin{proof}
  We must show that for every $F \in \Pred^\bullet_{\DN} (X)$, for
  every $h \in \Lform X$,
  $F (h) = \min_{G \in \Pred^\bullet_{\Nature\;wk} (X), F \leq G} G
  (h)$.  The difficult part is to show that
  $F (h) \geq \min_{G \in \Pred^\bullet_{\Nature\;wk} (X), F \leq G} G
  (h)$ for every $h$, and this is a direct consequence of
  Lemma~\ref{lemma:DN:s:*}, taking $h_0 \eqdef h$.
\end{proof}

\begin{lemma}
  \label{lemma:DN:s:cont}
  Let $X$ be a topological space.  Then $s^\bullet_{\DN}$ is a
  continuous map from $\Pred^\bullet_{\DN\;wk} (X)$ to
  $\SV (\Pred^\bullet_{\Nature\;wk} (X))$.
\end{lemma}
\begin{proof}
  As we already mentioned in the proof of Lemma~\ref{lemma:AN:s:cont},
  every open subset $V$ of $\Pred^\bullet_{\Nature\;wk} (X)$ can be
  written $\bigcup_{i \in I} \bigcap_{j \in J_i} V_{ij}$, where each
  $V_{ij}$ is of the form $[h > 1]_\Nature$, $h \in \Lform X$,
  and each $J_i$ is finite.  We can also write this as the directed
  union, over all finite subsets $I'$ of $I$, of
  $\bigcup_{i \in I'} \bigcap_{j \in J_i} V_{ij}$.  Now we can
  distribute unions over intersections, and write $V$ as a directed
  union of finite intersections of finite unions of subsets of the
  form $[h > b]_\Nature$.  It is easy to check that $\Box$ distributes
  over directed unions (using the fact that the elements of a Smyth
  powerdomain are compact, and that from every open cover of a compact
  subset $K$ by a directed family, one can extract a single element of
  the family containing $K$), and over finite intersections.  It
  follows that a subbase of the topology on
  $\SV (\Pred^\bullet_{\Nature\;wk} (X))$ is given by the subsets of
  the form $\Box W$ with $W \eqdef {\bigcup_{i=1}^m [h_i > 1]_\Nature}$,
  $U_i$ open in $X$.

  To show that $s^\bullet_{\DN}$ is continuous, we fix such an open
  $W$, and we claim that ${s^\bullet_{\DN}}^{-1} (\Box W)$ is open in
  the weak topology.  It will be easier to show that its complement
  $D$ is closed.  Observe that $F$ is in $D$ if and only if there is a
  linear prevision $G$ such that $F \leq G$, and for every $i$, $1\leq
  i \leq m$, $G (h_i) \leq 1$.

  \paragraph{\em Case 1: $\bullet$ is neither ``$\leq 1$'' nor ``$1$''.}

  We claim that $D$ is the intersection of the complements of the
  subbasic weak opens $[\sum_{i=1}^m a_i h_i > 1]_{\DN}$ over all
  $m$-tuples of non-negative real numbers 
  $a_1$, \ldots, $a_m$, such that $\sum_{i=1}^m a_i=1$, which will
  prove the claim.  Clearly $D$ is included in this intersection: if
  $F \leq G$, $G$ is linear, and $G (h_i) \leq 1$ for every $i$, then
  $F (\sum_{i=1}^m a_i h_i) \leq G (\sum_{i=1}^m a_i h_i) \leq
  \sum_{i=1}^m a_i = 1$.  Conversely, let $F \in \Pred_{\DN} (X)$ be
  such that $F (\sum_{i=1}^m a_i h_i) \leq 1$ for all $a_1, \ldots,
  a_m \in \Rp$ such that $\sum_{i=1}^m a_i=1$, and let us show
  that $F$ is in $D$.  Define $p (h)$ as $\inf \{\sum_{i=1}^m a_i \mid
  a_1, \ldots, a_m \in \Rp, \sum_{i=1}^m a_i h_i \geq h\}$.  This
  is easily seen to be a sublinear map such that $F \leq p$.  By
  Keimel's Sandwich Theorem, there is a continuous linear map $G$
  between $F$ and $p$.  By taking $a_i=1$ and $a_j=0$ for all $j \neq
  i$, we see that $p (h_i) \leq 1$, whence $G (h_i) \leq 1$ for every
  $i$, $1\leq i\leq m$.  So $F$ is in $D$.  This is all we need to
  show that $D$ is closed in the weak topology, hence that
  $s^\bullet_{\DN}$ is weakly continuous.

  \paragraph{\em Case 2: $\bullet$ is ``$\leq 1$'' or ``$1$''.}

  Underlining the changes, we claim that $D$ is the intersection of
  the complements of the subbasic weak opens $[\underline{a_0} +
  \sum_{i=1}^m a_i h_i > 1]_{\DN}$ over all \underline{$(m+1)$}-tuples
  of non-negative real numbers \underline{$a_0$,} $a_1$, \ldots, $a_m$
  such that $\sum_{i=\underline{0}}^m a_i=1$, which will prove the
  claim.  Clearly $D$ is included in this intersection: if $F \leq G$,
  $G$ is linear \underline{and subnormalized}, and $G (h_i) \leq 1$
  for every $i$, then $F (\underline{a_0 +} \sum_{i=1}^m a_i h_i) \leq
  G (\underline{a_0 +} \sum_{i=1}^m a_i h_i) \underline{\leq a_0 +}
  \sum_{i=1}^m a_i = 1$.  Conversely, let $F \in \Pred_{\DN} (X)$ be
  such that $F (\underline{a_0+} \allowbreak \sum_{i=1}^m a_i h_i)
  \leq \underline{a_0+} \sum_{i=1}^m a_i$ for all $a_1, \ldots, a_m
  \in \Rp$ such that $\sum_{i=\underline{0}}^m a_i=1$, and let us
  show that $F$ is in $D$.  Define $p (h)$ as $\inf \{\underline{a_0+}
  \sum_{i=1}^m a_i \mid \underline{a_0,} a_1, \ldots, a_m \in \Rp,
  \underline{a_0+} \allowbreak \sum_{i=1}^m a_i h_i \geq h\}$.  This
  is a sublinear map such that $F \leq p$.  By Keimel's Sandwich
  Theorem, there is a continuous linear map $G$ between $F$ and $p$.
  \underline{Additionally,} since $G (\one) \leq p (\one)$ and $p(\one) \leq 1$
  (take $a_0 \eqdef 1$, $a_i \eqdef 0$ for $i\neq 0$), $G$ is
  subnormalized; if in addition $\bullet$ is ``$1$'', then
  $F (\one)=1$, so $F (\one) \leq G (\one)$,
  so $G$ is normalized.  In any case, \underline{$G$ is in
    $\Pred^\bullet_{\Nature\;wk} (X)$}.  By taking $a_i=1$ and $a_j=0$
  for all $j \neq i$, we see that $p (h_i) \leq 1$, whence
  $G (h_i) \leq 1$ for every $i$, $1\leq i\leq m$.  So $F$ is in $D$.
  This is all we need to show that $D$ is closed in the weak topology,
  hence that $s^\bullet_{\DN}$ is continuous.
\end{proof}

\begin{lemma}
  \label{lemma:DN:sr}
  Let $X$ be a topological space.  For every $Q \in \SV
  (\Pred^\bullet_{\Nature\;wk} (X))$, $Q \subseteq s^\bullet_{\DN}
  (r_\DN (Q))$.
\end{lemma}
\begin{proof}
  For every $G \in Q$, $G (h) \geq \inf_{G' \in Q} G' (h)$.
\end{proof}

We sum up the above results as follows.  More than a retraction, we
now have a projection, by Lemma~\ref{lemma:DN:sr}: remember that the
ordering on Smyth powerdomains is reverse inclusion $\supseteq$, so
$s^\bullet_{\DN} \circ r_\DN$ is below the identity.
\begin{proposition}[$\SV (\Pred^\bullet_{\Nature\;wk} (X))$ projects
  onto $\Pred^\bullet_{\DN\;wk} (X)$]
  \label{prop:DN:weak}
  Let $\bullet$ be the empty superscript, ``$\leq 1$'', or ``$1$''.
  Let $X$ be a topological space.  Then $r_{\DN}$ defines a projection
  of $\SV (\Pred^\bullet_{\Nature\;wk} (X))$ onto
  $\Pred^\bullet_{\DN\;wk} (X)$, with associated embedding
  $s^\bullet_{\DN}$.
\end{proposition}

\subsection{The Retraction in the Erratic Case}
\label{sec:retr-errat-case}

In the erratic cases, recall that a fork $(F^-, F^+)$ is a pair of a
Smyth prevision $F^-$ and of a Hoare prevision $F^+$ satisfying
Walley's condition.

\begin{definition}
  \label{defn:rs:ADN}
  Let $X$ be a topological space.  For every non-empty set $E$ of
  linear previsions on $X$, let
  $r_{\ADN} (E) \in
  (\Lform X \to \creal)^2$ be
  $(r_{\DN} (E), r_{\AN} (E))$.

  Conversely, for every (subnormalized, normalized) fork $(F^-,
  F^+)$ on $X$, let $s_{\ADN} (F^-, F^+)$ (resp., $s^{\leq 1}_{\ADN}
  (F^-, F^+)$, $s^1_{\ADN} (F^-, F^+)$) be the set of all
  (subnormalized, normalized) linear previsions $G$ such that $F^-
  \leq G \leq F^+$.
\end{definition}
So $s^\bullet_{\ADN} (F^-, F^+) = s^\bullet_{\DN} (F^-) \cap
s^\bullet_{\AN} (F^+)$, whatever the superscript $\bullet$.
We shall now prove that $r_{\ADN}$, $s^\bullet_{\ADN}$ form a
retraction, once again.

This will require not only that $C \eqdef \Lform X$ (with its Scott
topology)
be locally convex, as in the Hoare cases, but also that that addition
on $C$ be almost open.  Following
\cite[Definition~4.6]{Keimel:topcones2}, we say that addition is
\emph{almost open} on a semitopological cone if and only if
$\upc (U+V)$ is open for every pair of open subsets $U$, $V$ of $C$.

For every monotonic map $q$ from $C$ to $\creal$, there is a
(pointwise) largest continuous map $\check q$ less than or equal to
$q$: $\check q (h)$ is the least upper bound of all real numbers $r$
such that $h$ is in the interior of $q^{-1} (r, +\infty]$.  When $C$
is almost open, and $q$ is superlinear, $\check q$ is again
superlinear \cite[Lemma~5.7]{Keimel:topcones2}.  This will be our main
new ingredient.

There is a canonical situation in which all the above assumptions are
satisfied:
\begin{lemma}
  \label{lemma:funspace:almostopen}
  Let $X$ be a core-compact, core-coherent space (for example, a
  stably locally compact space).  Then $\Lform X$ 
  is a locally convex topological cone in which addition is almost open.
\end{lemma}
\begin{proof}
  Since $X$ is core-compact, we already know that $\Lform X$ is a
  locally convex topological cone.  $\Lform X$
  is a continuous d-cone, too, and in that case, the fact that
  addition is almost open is equivalent to the property that the
  way-below relation $\ll$ on $C$ is \emph{additive}, namely,
  $f \ll f'$ and $g \ll g'$ together imply $f+g \ll f'+g'$
  \cite[Lemma~6.14]{Keimel:topcones2}.

  The space $Y \eqdef \Sober (X)$ is stably locally compact, and
  Proposition~2.28 of \cite{TKP:nondet:prob} then states that $\ll$ is
  additive on $[Y \to {\creal}_\sigma]$.  By Lemma~\ref{lemma:Stone:h},
  $\ll$ is also additive on the homeomorphic space $\Lform X$,
  which allows us to conclude.  (Note that
  core-coherence is required here.  Using the same sobrification
  trick, Proposition~2.29 of \cite{TKP:nondet:prob} says that if $X$
  is core-compact, and $\ll$ is additive on $\Lform X$,
  then $X$ must in fact be core-coherent.)
\end{proof}

\begin{lemma}
  \label{lemma:<<:achiX}
  Let $X$ be a compact space.  For every $a \in [0, 1)$, $a \cdot \one \ll
  \one$
  in $\Lform X$,
  hence $\uuarrow (a \cdot \one)$ is an open neighborhood of $\one$ in
  $\Lform X$.
\end{lemma}
\begin{proof}
  Let ${(f_i)}_{i \in I}$ be a directed family of continuous maps such
  that $\sup_{i \in I} f_i \geq \one$.
  So
  ${(\sup_{i \in I} f_i)}^{-1} (a, +\infty] = \bigcup_{i \in I}
  f_i^{-1} (a, +\infty]$ contains the compact set $X$.  Since the
  union is directed, $X \subseteq f_i^{-1} (a, +\infty]$ for some
  $i \in I$, proving the claim.
\end{proof}

\begin{lemma}
  \label{lemma:ADN:r:cont}
  Let $X$ be a topological space.  Then $r_{\ADN}$ is a continuous map
  from $\PV (\Pred^\bullet_{\Nature\;wk} (X))$ to
  $\Pred^\bullet_{\ADN\;wk} (X)$.
\end{lemma}
\begin{proof}
  We check that $r_{\ADN} (L)$ satisfies Walley's condition for every
  lens $L$.  Let $(F^-, F^+) \eqdef r_{\ADN} (L)$.  Then
  $F^- (h+h') = \inf_{G \in L} G (h+h') = \inf_{G \in L} (G (h) + G
  (h')) \leq \inf_{G \in L} (G (h) + \sup_{G \in L} G (h')) = \inf_{G
    \in L} G (h) + \sup_{G \in L} G (h')$ (since $L$ is non-empty)
  $= F^- (h) + F^+ (h')$.  We prove
  $F^- (h) + F^+ (h') \leq F^+ (h+h')$ similarly.  That $r_{\ADN} (L)$
  is then a fork, and that $r_{\ADN}$ is continuous, then follows from
  Lemma~\ref{lemma:AN:r:cont} and Lemma~\ref{lemma:DN:r:cont}.
\end{proof}

\begin{lemma}
  \label{lemma:ADN:s:def}
  Let $X$ be a topological space.  Then $s^\bullet_{\ADN}$ is a map
  from $\Pred^\bullet_{\ADN} (X)$ to $\PV (\Pred^\bullet_{\Nature\;wk}
  (X))$.
\end{lemma}
\begin{proof}
  We check that $s^\bullet_{\ADN} (F^-, F^+)$ is non-empty for every
  (subnormalized, normalized) fork $(F^-, F^+)$, i.e., that there is a
  (subnormalized, normalized) linear prevision $G$ such that $F^- \leq
  G \leq F^+$: this is a trivial consequence of Keimel's Sandwich
  Theorem on the semitopological cone $C \eqdef \Lform X$.
  Then $s^\bullet_{\ADN} (F^-, F^+) = s^\bullet_{\DN} (F^-) \cap
  s^\bullet_{\AN} (F^+)$ is a lens, by Lemma~\ref{lemma:AN:s:closed}
  and Lemma~\ref{lemma:DN:s:compact}.
\end{proof}

Before we continue, we establish two consequences of Walley's
condition.  Those are akin to the Main Lemma (Lemma~5.1) of
\cite{KP:predtrans:pow}.
\begin{lemma}
  \label{lemma:ADN:intermezzo:1}
  Let $X$ be a topological space, and $(F^-, F^+) \in
  \Pred^\bullet_{\ADN} (X)$.  For every $G' \in
  \Pred^\bullet_{\Nature} (X)$ such that $F^- \leq G'$, there is a $G
  \in \Pred^\bullet_{\Nature} (X)$ such that $F^- \leq G \leq F^+$ and
  $G \leq G'$.
\end{lemma}
\begin{proof}
  Write $C$ for $\Lform X$.
  Let
  $p (h) \eqdef \inf_{\substack{h_1, h_2 \in C\\h \leq h_1+h_2}} (F^+
  (h_1) + G' (h_2))$.  This is a sublinear map, notably, subadditivity
  is proved as follows:
  \begin{eqnarray*}
    p (h) + p (h') & = & \inf_{\substack{h_1, h_2, h'_1, h'_2 \in C\\h
        \leq h_1+h_2\\ h' \leq h'_1+h'_2}} (F^+ (h_1) + F^+ (h'_1) +
    G' (h_2) + G' (h'_2)) \\
    & \geq & \inf_{\substack{h_1, h_2, h'_1, h'_2 \in C\\h
        \leq h_1+h_2\\ h' \leq h'_1+h'_2}} (F^+ (h_1 + h'_1) +
    G' (h_2 + h'_2)) \\
    & \geq & \inf_{\substack{h''_1, h''_2 \in C\\ h+h' \leq
        h''+h''_2}} (F^+ (h''_1) + G' (h''_2)) = p (h+h').
  \end{eqnarray*}
  Note also that $p \leq F^+$ (take $h_1 \eqdef h$, $h_2 \eqdef 0$)
  and $p \leq G'$ (take $h_1 \eqdef 0$, $h_2 \eqdef h$).  We check
  that $F^- \leq p$, i.e., that for all $h_1, h_2 \in C$ such that
  $h \leq h_1+h_2$, $F^- (h) \leq F^+ (h_1) + G' (h_2)$.  This is by
  Walley's condition, since
  $F^- (h) \leq F^- (h_1+h_2) \leq F^+ (h_1) + F^- (h_2)$, and
  $F^- \leq G'$.  By Keimel's Sandwich Theorem, there is a continuous
  linear map $G$ such that $F^- \leq G \leq p$.  When $\bullet$ is
  ``$\leq 1$'', $G (\one) \leq p (\one) \leq F^+ (\one) \leq 1$, so
  $G$ is subnormalized, and
  when $\bullet$ is ``$1$'', $p (\one) \geq F^- (\one) = 1$, so $G$ is
  normalized.  In any case, $G$ is in $\Pred^\bullet_{\Nature} (X)$.
  Moreover, $F^- \leq G \leq p \leq F^+$, and $G \leq p \leq G'$.
\end{proof}

\begin{lemma}
  \label{lemma:ADN:intermezzo:2}
  Let $X$ be a topological space such that $\Lform X$
  is locally convex and has an almost open addition map.  Assume that
  $X$ is also compact in case $\bullet$ is ``$1$''.  Let
  $(F^-, F^+) \in \Pred^\bullet_{\ADN} (X)$.  For every
  $G' \in \Pred^\bullet_{\Nature} (X)$ such that $G' \leq F^+$, there
  is a $G \in \Pred^\bullet_{\Nature} (X)$ such that
  $F^- \leq G \leq F^+$ and $G' \leq G$.
\end{lemma}
\begin{proof}
  Write again $C$ for $\Lform X$.
  Let $q_1 (h)$ be defined as
  $\sup_{\substack{h_1, h_2 \in C\\ h_1+h_2 \leq h}} (F^- (h_1) + G'
  (h_2))$.  By an argument similar to the one we used in the proof of
  Lemma~\ref{lemma:ADN:intermezzo:1} on $p$, $q_1$ is superlinear, and
  $F^- \leq q_1$, $G' \leq q_1$.  Moreover, we check that
  $q_1 \leq F^+$, i.e., that for all $h_1, h_2 \in C$ such that
  $h_1+h_2 \leq h$, $F^- (h_1) + G' (h_2) \leq F^+ (h)$: this is by
  Walley's condition again, since
  $F^+ (h) \geq F^+ (h_1+h_2) \geq F^- (h_1) + F^+ (h_2) \geq F^-
  (h_1) + G' (h_2)$.

  However, to apply Keimel's Sandwich Theorem to $q_1$ and $F^+$, we
  would need $q_1$ to be continuous.  Instead, consider $\check q_1$.
  We have already noted that, since addition on $\Lform X$
  is almost open, $\check q_1$ is not only
  continuous, but also superlinear \cite[Lemma~5.7]{Keimel:topcones2}.
  Since $\check q_1$ is the largest Scott-continuous map below $q_1$,
  and $F^-$, $G'$ are Scott-continuous and below $q_1$, we also have
  $F^- \leq \check q_1$, $G' \leq \check q_1$.  Since $\check q_1 \leq
  q_1 \leq F^+$, we can now apply Keimel's Sandwich Theorem, and
  obtain a linear prevision $G$ such that $\check q_1 \leq G \leq
  F^+$.

  If $\bullet$ is ``$\leq 1$'', then
  $G (\one) \leq F^+ (\one) \leq 1$, so
  $G$ is subnormalized.  If $\bullet$ is ``$1$'', remember that we
  have assumed that $X$ was also compact.  By
  Lemma~\ref{lemma:<<:achiX}, $\uuarrow (a \cdot \one)$
  is an open neighborhood of $\one$
  for every $a \in [0, 1)$.  We check that is is included in
  $q_1^{-1} (r, +\infty]$ for every $r < a$: for every
  $h \in \uuarrow (a \cdot \one)$,
  $q_1 (h) \geq q_1 (a \cdot \one) \geq F^- (a \cdot \one) = a > r$.  In
  particular, for every $r \in [0, 1)$, we have just shown that
  $\one$
  is in the interior of $q_1^{-1} (r, +\infty]$ (pick any
  $a \in (r, 1)$).  Since $\check q_1 (\one))$ is
  the least upper bound of all real numbers $r$ with that property,
  $\check q_1 (\one) \geq 1$.
  Since $G (\one) \geq \check q_1 (\one)$,
  $G$ is normalized in the ``$1$'' case.  To sum up, whatever
  $\bullet$ is, $G$ is in $\Pred^\bullet_{\Nature} (X)$.

  Finally, $F^- \leq \check q_1 \leq G \leq F^+$, and $G' \leq \check
  q_1 \leq G$.
\end{proof}

\begin{lemma}
  \label{lemma:ADN:s:cont}
  Let $X$ be a topological space such that $\Lform X$ 
  is locally convex and has an almost open addition map.  Assume that
  $X$ is also compact in case $\bullet$ is ``$1$''.  Then
  $s^\bullet_{\ADN}$ is a continuous map from
  $\Pred^\bullet_{\ADN\;wk} (X)$ to $\PV (\Pred^\bullet_{\Nature\;wk}
  (X))$.
\end{lemma}
\begin{proof}
  For every open subset $V$ of $\Pred^\bullet_{\Nature\;wk} (X)$,
  ${s^\bullet_{\ADN}}^{-1} (\Box V)$ is the set of (subnormalized,
  normalized) forks $(F^-, F^+)$ such that every $G \in
  \Pred^\bullet_{\Nature} (X)$ such that $F^- \leq G \leq F^+$ is in
  $V$.  We claim that this is exactly $\pi_1^{-1}
  ({s^\bullet_{\DN}}^{-1} (\Box V))$, where $\pi_1$ is first
  projection.  The inclusion $\pi_1^{-1} ({s^\bullet_{\DN}}^{-1} (\Box
  V)) \subseteq {s^\bullet_{\ADN}}^{-1} (\Box V)$ is trivial.
  Conversely, for every element $(F^-, F^+)$ of
  ${s^\bullet_{\ADN}}^{-1} (\Box V)$, we must show that every $G' \in
  \Pred^\bullet_{\Nature} (X)$ such that $F^- \leq G'$ is in $V$.  By
  Lemma~\ref{lemma:ADN:intermezzo:1}, one can find $G \in
  \Pred^\bullet_{\Nature} (X)$ such that $F^- \leq G \leq F^+$ and $G
  \leq G'$.  In particular, since $(F^-, F^+)$ is in
  ${s^\bullet_{\ADN}}^{-1} (\Box V)$, $G$ is in $V$.  Since $V$ is
  upward closed, $G'$ is also in $V$.

  For every open subset $V$ of $\Pred^\bullet_{\Nature\;wk} (X)$,
  ${s^\bullet_{\ADN}}^{-1} (\Diamond V)$ is the set of (subnormalized,
  normalized) forks $(F^-, F^+)$ such that there is a $G \in
  \Pred^\bullet_{\Nature} (X)$ such that $F^- \leq G \leq F^+$ that is
  also in $V$.  We claim that this is exactly $\pi_2^{-1}
  ({s^\bullet_{\DN}}^{-1} (\Diamond V))$, where $\pi_2$ is second
  projection.  The inclusion ${s^\bullet_{\ADN}}^{-1} (\Diamond V)
  \subseteq \pi_2^{-1} ({s^\bullet_{\DN}}^{-1} (\Diamond V))$ is
  trivial.  Conversely, for every element $(F^-, F^+)$ of $\pi_2^{-1}
  ({s^\bullet_{\DN}}^{-1} (\Diamond V))$, i.e., such that there is a
  $G' \in V$ with $G' \leq F^+$, we must show that there is a $G \in
  V$ such that $F^- \leq G \leq F^+$.  We use
  Lemma~\ref{lemma:ADN:intermezzo:2} to this end: this yields such a
  $G$, since $V$ is upward closed.

  To sum up, we have shown that both ${s^\bullet_{\ADN}}^{-1} (\Box V)
  = \pi_1^{-1} ({s^\bullet_{\DN}}^{-1} (\Box V))$ and
  ${s^\bullet_{\ADN}}^{-1} (\Diamond V) = \pi_2^{-1}
  ({s^\bullet_{\DN}}^{-1} (\Diamond V))$ are open, so
  $s^\bullet_{\ADN}$ is continuous.
\end{proof}

\begin{lemma}
  \label{lemma:ADN:rs=id}
  Let $X$ be topological space such that $[X \to
  {\creal}_\sigma]_\sigma$ is locally convex and has an almost open
  addition map.  Assume that $X$ is also compact in case $\bullet$ is
  ``$1$''.  Then $r_{\ADN} \circ s^\bullet_{\ADN}$ is the identity map
  on $\Pred^\bullet_{\ADN\;wk} (X)$.
\end{lemma}
\begin{proof}
  For every (subnormalized, normalized) fork $(F^-, F^+)$, let
  $L \eqdef s^\bullet_{\ADN} (F^-, F^+)$.  One may write $L$ as
  $Q \cap F$, where $Q \eqdef s^\bullet_{\DN} (F^-)$, and
  $F \eqdef s^\bullet_{\AN} (F^+)$, and by
  Proposition~\ref{prop:DN:weak} and Proposition~\ref{prop:AN:weak},
  $r_{\DN} (Q) = F^-$, $r_{\AN} (F) = F^+$.  To show that
  $r_{\ADN} (s^\bullet_{\ADN} (F^-, F^+)) = (F^-, F^+)$, we must show
  that $r_{\DN} (L) = F^-$ and $r_{\AN} (L) = F^+$.

  Since $L \subseteq Q$, using the definition of $r_{\DN}$
  (Definition~\ref{defn:rs}), $r_{\DN} (L) \geq F^-$.  Since $r_{\DN}
  (Q) = F^-$ (and using Lemma~\ref{lemma:rDN}), for every $h \in C$,
  there is a $G' \in Q$ such that $G' (h) = F^- (h)$.  By definition
  of $Q$, $F^- \leq G'$.  By Lemma~\ref{lemma:ADN:intermezzo:1}, there
  is a $G \in \Pred^\bullet_{\Nature} (X)$ such that $F^- \leq G \leq
  F^+$ (i.e., such that $G \in L$) and $G \leq G'$.  So $r_{\DN} (L)
  (h) \leq G (h) \leq G' (h) = F^- (h)$.  It follows that $r_{\DN} (L)
  \leq F^-$, hence $r_{\DN} (L) = F^-$.

  Since $L \subseteq F$, using the definition of $r_{\AN}$, $r_{\AN}
  (L) \leq F^+$.  Since $r_{\AN} (F) = F^+$, for every $h \in C$, and
  for every real number $r < F^+ (h)$, there is a $G' \in F$ such that
  $G' (h) \geq r$.  By definition of $F$, $G' \leq F^+$.  By
  Lemma~\ref{lemma:ADN:intermezzo:2}, there is a $G \in
  \Pred^\bullet_{\Nature} (X)$ such that $F^- \leq G \leq F^+$ (i.e.,
  such that $G \in L$) and $G' \leq G$.  So $r_{\AN} (L) (h) \geq G
  (h) \geq G' (h) \geq r$.  Taking sups over $r$, $r_{\AN} (L) (h)
  \geq F^+ (h)$.  So $r_{\AN} (L) \geq F^+$, whence $r_{\AN} (L) =
  F^+$, and we are done.
\end{proof}

We sum up these results as follows.
\begin{proposition}[$\PV (\Pred^\bullet_{\Nature\;wk} (X))$ retracts
  onto $\Pred^\bullet_{\ADN\;wk} (X)$]
  \label{prop:ADN:weak}
  Let $\bullet$ be the empty superscript, ``$\leq 1$'', or ``$1$''.
  Let $X$ be a topological space such that $[X \to
  {\creal}_\sigma]_\sigma$ is locally convex and has an almost open
  addition map, for example a stably locally compact space, or more
  generally, a core-compact, core-coherent space.  Assume also that
  $X$ is compact in case $\bullet$ is ``$1$''.

  Then $r_{\ADN}$ defines a retraction of $\PV
  (\Pred^\bullet_{\Nature\;wk} (X))$ onto $\Pred^\bullet_{\ADN\;wk}
  (X)$, with associated section $s^\bullet_{\ADN}$.
\end{proposition}

\subsection{Domain-Theoretic Consequences}
\label{sec:doma-theor-cons}

The focus in domain theory is on Scott topologies rather than on weak
topologies.  But Scott-continuity quickly follows provided we use the
following result \cite[Lemma~3.8]{JGL:qrb}: for every quasi monotone
convergence space $Z$, and every $T_0$ topological space $Z'$, every
continuous map from $Z$ to $Z'$ is continuous from $Z_\sigma$ to
$Z'_\sigma$.  Here $Z_\sigma$ is $Z$ with the Scott topology of its
specialization preorder, and a quasi monotone convergence space is a
space where this topology is finer than the original one on $Z$.  The
specialization ordering of any space of previsions with the weak
topology is the pointwise ordering $\leq$, and each of these spaces is
quasi monotone convergence, because $[h > r]$ is Scott open.  So,
given any space of previsions $Z$, $Z_\sigma$ is just the same space
with the Scott topology of its ordering.  Similarly for spaces of
forks.  Also, $\SV (Y)_\sigma = \Smyth (Y)$ and $\HV (Y)_\sigma =
\Hoare (Y)$.  To stress that we are using Scott-continuity, call a
\emph{Scott retraction} any retraction between posets equipped with
their Scott topologies.

Applying this reasoning to Proposition~\ref{prop:AN:weak}, we obtain:
\begin{proposition}
  \label{prop:AN:Scott}
  Let $\bullet$ be the empty superscript, ``$\leq 1$'', or ``$1$''.
  Let $X$ be a core-compact space.  Then $r_{\AN}$ defines a Scott
  retraction of $\Hoare (\Pred^\bullet_{\Nature\;wk} (X))$ onto
  $\Pred^\bullet_{\AN} (X)$, with associated section
  $s^\bullet_{\AN}$.  Moreover
  $r_{\AN} \circ s^\bullet_{\AN}
  \geq \identity\relax$.
\end{proposition}

{\updated Indeed, we recall from Remark~\ref{rem:Hfriendly:1} that
  every core-compact space is $\AN_\bullet$-friendly.}
 
Applying this reasoning to Proposition~\ref{prop:DN:weak}, we obtain:
\begin{proposition}
  \label{prop:DN:Scott}
  Let $\bullet$ be the empty superscript, ``$\leq 1$'', or ``$1$''.
  Let $X$ be a topological space.  Then $r_{\DN}$ defines a Scott
  retraction of $\Smyth (\Pred^\bullet_{\Nature\;wk} (X))$ onto
  $\Pred^\bullet_{\DN} (X)$, with associated section
  $s^\bullet_{\DN}$.  Moreover
  $r_{\DN} \circ s^\bullet_{\DN}
  \leq \identity\relax$.
\end{proposition}

Applying it to Proposition~\ref{prop:ADN:weak}, finally, we obtain:
\begin{proposition}
  \label{prop:ADN:Scott}
  Let $\bullet$ be the empty superscript, ``$\leq 1$'', or ``$1$''.
  Let $X$ be a core-compact, core-coherent space (and compact if
  $\bullet$ is ``$1$'').  Then $r_{\ADN}$ defines a Scott retraction
  of $\Plotkin (\Pred^\bullet_{\Nature\;wk} (X))$ onto
  $\Pred^\bullet_{\ADN} (X)$, with associated section $s^\bullet_{\ADN}$.
\end{proposition}

There is still a bit of the weak topology lying around in the latter
three propositions, in the various spaces $\Pred^\bullet_{\Nature\;wk}
(X)$ of linear previsions involved.

The Scott topology is finer than the weak topology on any space of
previsions.  Up to the canonical isomorphism between
$\Pred^\bullet_{\Nature} (X)$ and $\Val^\bullet (X)$ (for whichever
superscript $\bullet$), the \emph{Kirch-Tix Theorem} states that,
whenever $X$ is a continuous dcpo, the Scott and weak topologies
coincide on $\Pred_{\Nature} (X)$ \cite[Satz~4.10]{Tix:bewertung}, and
on $\Pred^{\leq 1}_{\Nature} (X)$ \cite[Satz~8.6]{Kirch:bewertung}.

Given a poset $X$, let $X_\bot$ be $X$ plus a fresh bottom element
$\bot$, below all points of $X$.  By a trick due to Edalat
\cite[Section~3]{Edalat:int}, the spaces $\Val^1 (X_\bot)$ of all
normalized continuous valuations on $X_\bot$ and $\Val^{\leq 1} (X)$
of all subnormalized continuous valuations on $X$ are
order-isomorphic, and also homeomorphic in their weak topologies.
(For $\nu \in \Val^1 (X_\bot)$, define a subnormalized valuation on
$X$ by considering $\nu$ restricted to the opens contained in $X$.
Conversely, for $\nu \in \Val^{\leq 1} (X)$, define
$\nu' \in \Val^1 (X_\bot)$ by $\nu' (U) \eqdef \nu (U)$ if $U$ is open
in $X$, and $\nu' (X_\bot) \eqdef 1$.)  It follows that, on pointed
continuous dcpos, that is, on continuous dcpos that we can write as
$X_\bot$ for some, necessarily continuous, dcpo $X$, the Scott and
weak topologies also coincide on $\Pred^1_{\Nature} (X)$.

Every continuous dcpo is locally compact, and every pointed poset is
compact in its Scott topology.  We therefore obtain the following
purely domain-theoretic statements (all spaces come with their Scott
topologies) from Proposition~\ref{prop:AN:Scott},
Proposition~\ref{prop:DN:Scott}, and Proposition~\ref{prop:ADN:Scott}
respectively.
\begin{proposition}
  \label{prop:AN:Scott2}
  Let $\bullet$ be the empty superscript, ``$\leq 1$'', or ``$1$''.
  Let $X$ be a continuous dcpo (and pointed if $\bullet$ is ``$1$'').
  Then $r_{\AN}$ defines a Scott retraction (even a
  coembedding-coprojection pair) of $\Hoare (\Pred^\bullet_{\Nature}
  (X))$ onto $\Pred^\bullet_{\AN} (X)$, with associated section
  $s^\bullet_{\AN}$.
\end{proposition}

\begin{proposition}
  \label{prop:DN:Scott2}
  Let $\bullet$ be the empty superscript, ``$\leq 1$'', or ``$1$''.
  Let $X$ be a continuous dcpo.  Then $r_{\DN}$ defines a Scott
  retraction (even an embedding-projection pair) of $\Smyth
  (\Pred^\bullet_{\Nature} (X))$ onto $\Pred^\bullet_{\DN} (X)$, with
  associated section $s^\bullet_{\DN}$.
\end{proposition}

\begin{proposition}
  \label{prop:ADN:Scott2}
  Let $\bullet$ be the empty superscript, ``$\leq 1$'', or ``$1$''.
  Let $X$ be a coherent, continuous dcpo (and pointed if $\bullet$ is
  ``$1$'').  Then $r_{\ADN}$ defines a Scott retraction of $\Plotkin
  (\Pred^\bullet_{\Nature} (X))$ onto $\Pred^\bullet_{\ADN} (X)$, with
  associated section $s^\bullet_{\ADN}$.
\end{proposition}

Let us draw a few domain-theoretic consequences of the above results.
All of these will state that, under suitable conditions, spaces of
Smyth, resp.\ Hoare previsions, and of forks are continuous dcpos,
with natural bases; and that the Scott and the weak topologies will
coincide on such spaces.

When $X$ is a continuous dcpo, $\Val (X)$ is a continuous dcpo, with a
basis of \emph{simple valuations}, i.e., valuations of the form
$\sum_{i=1}^n a_i \delta_{x_i}$, where $a_i \in \Rp$, $x_i \in X$
\cite[Theorem~IV.9.16]{GHKLMS:contlatt}.  This is an extension of
Jones' original theorem, that $\Val^{\leq 1} (X)$ is a continuous
dcpo, with a basis of subnormalized simple valuations
\cite[Chapter~5]{Jones:proba}.  Using Edalat's trick, we obtain a
similar result for $\Val^1 (X)$ and normalized simple valuations,
provided $X$ is also pointed.  Notice that the isomorphism with linear
previsions yields bases of $\Pred_{\Nature} (X)$ (resp., $\Pred^{\leq
  1}_{\Nature} (X)$, $\Pred^1_{\Nature} (X)$) consisting of
\emph{simple linear previsions} (resp., subnormalized, normalized) of
the form $h \mapsto \sum_{i=1}^n a_i h (x_i)$.

In turn, if $Y$ is a continuous dcpo, with basis $B$, then $\Hoare
(Y)$ is a continuous dcpo, too, and a basis is given by the subsets of
the form $\dc E$, where $E$ is a finite non-empty subset of $B$.  (See
\cite[Corollary~IV.8.7]{GHKLMS:contlatt}, which does not mention the
basis explicitly, or \cite[Theorem~6.2.10, Item~1]{AJ:domains}, which
does, but in less explicit a form; in the latter case, one should also
note that our version of the Hoare powerdomain coincides with theirs,
by their own Theorem~6.2.13.)  Now, for any retraction $r \colon D \to
E$, where $D$ is a continuous dcpo, $E$ is also a continuous dcpo, and
a basis of $E$ is given by the image under $r$ of a basis of $D$
\cite[Lemma~3.1.3]{AJ:domains}.  Using
Proposition~\ref{prop:AN:Scott2}, we obtain:
\begin{proposition}
  \label{prop:AN:base}
  Let $X$ be a continuous dcpo.  Then $\Pred_{\AN} (X)$ (resp.,
  $\Pred^{\leq 1}_{\AN} (X)$) is a continuous dcpo, with basis given
  by the finite non-empty sups of simple (resp., and subnormalized)
  linear previsions:
  \[
  h \mapsto \max_{i=1}^m \sum_{j=1}^{n_i} a_{ij} h (x_{ij})
  \]
  where $m\geq 1$ (resp., and $\sum_{j=1}^{n_i} a_{ij} \leq 1$ for
  every $i$).

  If $X$ is a pointed continuous dcpo, then $\Pred^1_{\AN} (X)$ is a
  pointed continuous dcpo, with basis given by the finite sups of
  simple normalized previsions (i.e., $\sum_{j=1}^{n_i} a_{ij} = 1$
  for every $i$).  The least element is $h \mapsto h (\bot)$, where
  $\bot$ is the least element of $X$.
\end{proposition}
Recall that,
when $Y$ is a continuous dcpo, then the Scott and the lower Vietoris
topologies coincide on $\Hoare (Y)$ \cite[Section~6.3.3]{schalk:diss}.
In particular, under the assumptions of
Proposition~\ref{prop:AN:base}, $\Hoare (\Pred_{\Nature} (X)) = \HV
(\Pred_{\Nature} (X)) = \HV (\Pred_{\Nature\;wk} (X))$ (by the
Kirch-Tix Theorem), and similarly in the subnormalized and normalized
cases.

Now, every section is a topological embedding.  Under the same
assumptions as above, Proposition~\ref{prop:AN:weak} and
Proposition~\ref{prop:AN:Scott2} imply that $\Pred_{\AN\;wk} (X)$ and
$\Pred_{\AN} (X)$ both embed into the same space $\Hoare
(\Pred_{\Nature} (X)) = \HV (\Pred_{\Nature\;wk} (X))$.  Hence they
have the same topology:
\begin{proposition}
  \label{prop:AN:weak=Scott}
  Let $X$ be a continuous dcpo.  The Scott topology coincides with the
  weak topology on $\Pred_{\AN} (X)$, on $\Pred^{\leq 1}_{\AN} (X)$;
  also on $\Pred^1_{\AN} (X)$ if $X$ is additionally assumed to be
  pointed.
\end{proposition}

Similarly, if $Y$ is a continuous dcpo with basis $B$, then $\Smyth
(Y)$ is a continuous dcpo, with basis given by the subsets of the form
$\upc E$, $E$ a finite and non-empty subset of $B$
\cite[Theorem~6.2.10, Item~2]{AJ:domains} (and our Smyth powerdomain
is the same as theirs, by their Theorem~6.2.14).
\begin{proposition}
  \label{prop:DN:base}
  Let $X$ be a continuous dcpo.  Then $\Pred_{\DN} (X)$ (resp.,
  $\Pred^{\leq 1}_{\DN} (X)$) is a continuous dcpo, with basis given
  by the finite non-empty infs of simple (resp., and subnormalized)
  linear previsions:
  \[
  h \mapsto \min_{i=1}^m \sum_{j=1}^{n_i} a_{ij} h (x_{ij})
  \]
  where $m\geq 1$ (resp., and $\sum_{j=1}^{n_i} a_{ij} \leq 1$ for
  every $i$).

  If $X$ is a pointed continuous dcpo, then $\Pred^1_{\DN} (X)$ is a
  pointed continuous dcpo, with basis given by the finite mins of
  simple normalized previsions (i.e., $\sum_{j=1}^{n_i} a_{ij} = 1$
  for every $i$).  The least element is $h \mapsto h (\bot)$, where
  $\bot$ is the least element of $X$.
\end{proposition}
The Scott and the upper Vietoris topologies coincide on $\Smyth (Y)$
on every $T_0$, well-filtered, locally compact space
\cite[Section~7.3.4]{schalk:diss},
in particular on every continuous dcpo $Y$.
By the same argument as for Proposition~\ref{prop:AN:weak=Scott}:
\begin{proposition}
  \label{prop:DN:weak=Scott}
  Let $X$ be a continuous dcpo.  The Scott topology coincides with the
  weak topology on $\Pred_{\DN} (X)$, on $\Pred^{\leq 1}_{\DN} (X)$;
  also on $\Pred^1_{\DN} (X)$ is $X$ is additionally assumed to be
  pointed.
\end{proposition}
 
When $Y$ is a continuous, coherent dcpo with basis $B$, then $\Plotkin
(Y)$ is again a continuous dcpo, with basis given by the subsets of
the form $\upc E \cap \dc E$, $E$ a finite and non-empty subset of
$B$.  This is Theorem~6.2.3 of \cite{AJ:domains}, together with
Theorem~6.2.22, which states that our Plotkin powerdomain is the same
as theirs.
\begin{proposition}
  \label{prop:ADN:base}
  Let $X$ be a continuous, coherent dcpo.  Then $\Pred_{\ADN} (X)$
  (resp., $\Pred^{\leq 1}_{\ADN} (X)$) is a continuous dcpo, with basis
  given by the simple (resp., and subnormalized) forks of the form:
  \[
  (h \mapsto \min_{i=1}^m \sum_{j=1}^{n_i} a_{ij} h (x_{ij}),
  h \mapsto \max_{i=1}^m \sum_{j=1}^{n_i} a_{ij} h (x_{ij}))
  \]
  where $m\geq 1$ (resp., and $\sum_{j=1}^{n_i} a_{ij} \leq 1$ for
  every $i$).

  If $X$ is a pointed continuous, coherent dcpo, then $\Pred^1_{\ADN}
  (X)$ is a pointed continuous dcpo, with basis the simple normalized
  forks (i.e., $\sum_{j=1}^{n_i} a_{ij} = 1$ for every $i$).  The
  least element is $(h \mapsto h (\bot), h \mapsto h (\bot))$, where
  $\bot$ is the least element of $X$.
\end{proposition}
We conclude this section with a result similar to
Proposition~\ref{prop:AN:weak=Scott} and
Proposition~\ref{prop:DN:weak=Scott}.  This time, we use the fact that
the way-below relation $\ll_{\Plotkin}$ on $\Plotkin (Y)$, when $Y$ is
a continuous, coherent dcpo, if given by $L \ll_{\Plotkin} L'$ iff
$\upc L \ll_\Smyth \upc L'$ and $cl (L) \ll_\Hoare cl (L')$
\cite[Section~6.2.1]{AJ:domains}; so our subbasic Scott open subsets
are $\{L' \in \Plotkin (Y) \mid L' \subseteq \uuarrow E \text{ and }
L' \cap \ddarrow E' \neq \emptyset\} = \Box {\uuarrow E} \cap
\bigcap_{y \in E'} \Diamond {\ddarrow y}$, where $E$ and $E'$ are
finite and non-empty.  Since they are all open in the Vietoris
topology, the Scott and Vietoris topologies coincide on $\Plotkin
(Y)$.  As before, we use the fact that our spaces of forks with the
weak and the Scott topologies are subspaces of the same space to
conclude:
\begin{proposition}
  \label{prop:ADN:weak=Scott}
  Let $X$ be a continuous, coherent dcpo.  The Scott topology
  coincides with the weak topology on $\Pred_{\ADN} (X)$, on
  $\Pred^{\leq 1}_{\ADN} (X)$; also on $\Pred^1_{\ADN} (X)$ is $X$ is
  additionally assumed to be pointed.
\end{proposition}

\section{The Isomorphisms}
\label{sec:isom}

Let $\bullet$ be the empty superscript, or ``$\leq 1$'', or ``$1$'',
depending on the case.  When {\updated $X$ is $\AN_\bullet$-friendly},
we know that $r_{\AN} \circ s^\bullet_{\AN}$ is the identity map,
where
$r_{\AN} \colon \HV (\Pred^\bullet_{\Nature\;wk} (X)) \to
\Pred^\bullet_{\AN\;wk} (X)$ and
$s^\bullet_{\AN} \colon \Pred^\bullet_{\AN\;wk} (X) \to \HV
(\Pred^\bullet_{\Nature\;wk} (X))$ (Corollary~\ref{corl:AN:weak}).

Now write $\HV^{cvx} (D)$ for the subspace of $\HV (D)$ consisting of
\emph{convex} closed, non-empty, subsets of $D \eqdef
\Pred^\bullet_{\AN\;wk} (X)$, and similarly $\Hoare^{cvx} (D)$ is the
underlying poset, with the inclusion ordering.  Note that although
$\Pred^\bullet_{\AN\;wk} (X)$ is not a cone (when $\bullet$ is ``$\leq
1$'' or ``$1$''), convexity makes sense.  We also write $\SV^{cvx}
(D)$, $\Smyth^{cvx} (D)$, $\Plotkin^{cvx} (D)$, $\PV^{cvx} (D)$ with
the obvious meaning.

Clearly, for every $F \in \Pred^\bullet_{\AN\;wk} (X)$,
$s^\bullet_{\AN} (F) = \{G \in D \mid G \leq F\}$ is convex, and
similarly for $s^\bullet_{\DN} (F) = \{G \in D \mid F \leq G\}$ and
$s^\bullet_{\ADN} (F^-, F^+) = \{G \in D \mid F^- \leq G \leq F^+\}$.
So $s^\bullet_{\AN}$ corestricts to a continuous map from
$\Pred^\bullet_{\AN\;wk} (X)$ to $\HV^{cvx}
(\Pred^\bullet_{\Nature\;wk} (X))$, and similarly for
$s^\bullet_{\DN}$ and $s^\bullet_{\ADN}$.  We wish to show that this
is a homeomorphism, with inverse the corresponding restriction of
$r_{\AN}$ to $\HV^{cvx} (\Pred^\bullet_{\Nature\;wk} (X))$, resp., of
$r_{\DN}$ to $\SV^{cvx} ((\Pred^\bullet_{\Nature\;wk} (X))$, resp., of
$r_{\ADN}$ to $\PV^{cvx} ((\Pred^\bullet_{\Nature\;wk} (X))$.

\subsection{The case of unbounded previsions}
\label{sec:case-unbo-prev}

The case where $\bullet$ is the empty superscript, i.e., where general
previsions are considered, was already dealt with in
\cite[Section~6.1]{KP:predtrans:pow}.  We improve on their result,
which required $X$ to be a continuous dcpo.  The general plan of the
proof is the same.

The key to the generalization is the following
\emph{Schr\"oder-Simpson Theorem}.
\begin{theorem}[Schr\"oder-Simpson]
  \label{thm:schsimp}
  Let $X$ be a topological space.  For every continuous linear map $\psi$
  from $\Val (X)$ to ${\creal}_\sigma$ (resp., from $\Pred_{\Nature\;wk}
  (X)$ to ${\creal}_\sigma$), there is a unique continuous
  map $h \colon X \to {\creal}_\sigma$ such that, for every $\nu \in
  \Val (X)$, $\psi (\nu) = \int_{x \in X} h (x) d \nu$ (resp., for
  every $G \in \Pred_{\Nature\;wk} (X)$, $\psi (G) = G (h)$).
\end{theorem}
The converse direction, that given a unique continuous map $h \colon X
\to {\creal}_\sigma$, the map $G \mapsto G (h)$ is continuous and linear
from $\Pred_{\Nature\;wk} (X)$ to ${\creal}_\sigma$, is obvious.

Theorem~\ref{thm:schsimp} is due to Schr\"oder and Simpson.
It was announced at the end of the presentation
\cite{SS:prob:obs}, and a full proof was given in another talk
\cite{SchSimp:obs:2}.  Another proof was discovered by Keimel, who
stresses the role of Hahn-Banach-like extension theorems in
quasi-uniform cones
\cite{Keimel:SchroderSimpson}.  A short, elementary proof of this
theorem can be found in \cite{JGL:SchSimp}.

\begin{lemma}
  \label{lemma:AN:orderembed}
  Let $X$ be a topological space.  For all $A, B \in \Hoare^{cvx}
  (\Pred_{\Nature\;wk} (X))$, if $r_{\AN} (A) \leq r_{\AN} (B)$ then
  $A \subseteq B$.
\end{lemma}
\begin{proof}
  Assume $A \not\subseteq B$, so there is a $G \in A$ that is not in
  $B$.  $\Pred_{\Nature\;wk} (X)$ is a locally convex topological
  cone, since every subbasic open set $[h > b]$ is convex.  Therefore,
  there is a convex open subset $U$ containing $G$ that does not
  intersect $B$.  By the Separation Theorem
  \cite[Theorem~9.1]{Keimel:topcones2}, there is a continuous linear
  map $\Lambda \colon C \to \creal$ such that $\Lambda (G') \leq 1$
  for every $G' \in B$, and $\Lambda (G') > 1$ for every $G' \in U$;
  in particular, $\Lambda (G) > 1$.  By the Schr\"oder-Simpson
  Theorem, $\Lambda$ is the map $G' \mapsto G' (h)$
  for some 
  $h \in [X \to {\creal}_\sigma]$.  So $\Lambda (G') = G' (h) \leq 1$
  for every $G' \in B$, which implies that $r_{\AN} (B) \leq 1$.  And
  $\Lambda (G) = G (h) > 1$, which implies that $r_{\AN} (A) > 1$,
  contradiction.
\end{proof}

\begin{proposition}
  \label{prop:AN:iso:gnrl}
  Let $X$ be a topological space such that $[X \to
  {\creal}_\sigma]_\sigma$ is locally convex, for example, a
  core-compact space.  Then $r_{\AN}$ defines a homeomorphism between
  $\HV^{cvx} (\Pred_{\Nature\;wk} (X))$ and $\Pred_{\AN\;wk} (X)$ and
  an order-isomorphism between $\Hoare^{cvx} (\Pred_{\Nature\;wk}
  (X))$ and $\Pred_{\AN} (X)$.
\end{proposition}
\begin{proof}
  By Proposition~\ref{prop:AN:weak}, $r_{\AN}$ is continuous hence
  monotonic, and $r_{\AN} \circ s_{\AN}$ equals the identity.  In
  particular, $r_{\AN}$ is surjective, and
  Lemma~\ref{lemma:AN:orderembed} implies in particular that it is
  injective.  So $r_{\AN}$ is a bijection, with inverse $s_{\AN}$.
  Both are continuous, by Proposition~\ref{prop:AN:weak}, hence define
  a homeomorphism.  The final part follows from the fact that every
  homeomorphism induces an order-isomorphism with respect to the
  underlying specialization preorders.
\end{proof}

\begin{lemma}
  \label{lemma:DN:orderembed}
  Let $X$ be a topological space.  For all $Q, Q' \in \Smyth^{cvx}
  (\Pred_{\Nature\;wk} (X))$, if $r_{\DN} (Q) \leq r_{\DN} (Q')$ then
  $Q \supseteq Q'$.
\end{lemma}
\begin{proof}
  Assume $Q \not\supseteq Q'$, so there is a $G \in Q'$ that is not in
  $Q$.  $\Pred_{\Nature\;wk} (X)$ is a locally convex topological
  cone, and $A \eqdef \dc G$ is a closed convex set disjoint from $Q$.
  By the Strict Separation Theorem
  \cite[Theorem~10.5]{Keimel:topcones2}, there is a continuous linear
  map $\Lambda \colon \Lform X \to {\creal}_\sigma$
  and a real number $r > 1$
  such that $\Lambda (G') \geq r$ for every $G' \in Q$, and
  $\Lambda (G') \leq 1$ for every $G' \in A$.  In particular,
  $\Lambda (G) \leq 1$.  By the Schr\"oder-Simpson Theorem,
  $\Lambda = G (h)$
  for some 
  $h \in [X \to {\creal}_\sigma]$, so $G (h) \leq 1$, which implies
  $r_{\DN} (Q') (h) \leq 1$.  Also,
  $r_{\DN} (Q) (h) = \min_{G' \in Q} G' (h) = \min_{G' \in Q} \Lambda
  (G') \geq r > 1$, contradiction.
\end{proof}

\begin{proposition}
  \label{prop:DN:iso:gnrl}
  Let $X$ be a topological space.  Then $r_{\DN}$ defines a
  homeomorphism between $\SV^{cvx} (\Pred_{\Nature\;wk} (X))$ and
  $\Pred_{\DN\;wk} (X)$, and an order-iso\-morphism between
  $\Smyth^{cvx} (\Pred_{\Nature\;wk} (X))$ and $\Pred_{\DN} (X)$.
\end{proposition}
\begin{proof}
  Same argument as in Proposition~\ref{prop:AN:iso:gnrl}, using
  Proposition~\ref{prop:DN:weak} and Lemma~\ref{lemma:DN:orderembed}.
\end{proof}

\begin{lemma}
  \label{lemma:ADN:r}
  Let $X$ be a topological space.  For every $L \in \Plotkin
  (\Pred_{\Nature\;wk} (X))$, $r_{\AN} (L) = r_{\AN} (cl (L))$, and
  $r_{\DN} (L) = r_{\DN} (\upc L)$.
\end{lemma}
\begin{proof}
  Since $L \subseteq cl (L)$, $r_{\AN} (L) (h) = \sup_{G \in L} G (h)
  \leq r_{\AN} (cl (L)) (h)$ for every $h \in [X \to {\creal}_\sigma]$.
  Assume the inequality were strict: for some $h \in [X \to
  {\creal}_\sigma]$, there would be a $G \in cl (L)$ such that $r_{\AN}
  (L) (h) < G (h)$.  Let $b \eqdef r_{\AN} (L) (h)$.  Then $G$ is in the
  open subset $[h > b]$, and since $G \in cl (L)$, $[h > b]$ must meet
  $L$.  So there is a $G' \in L$ such that $G' (h) > b$, which implies
  $r_{\AN} (L) (h) > b$, contradiction.  The second claim is obvious.
\end{proof}

\begin{lemma}
  \label{lemma:ADN:orderembed}
  Let $X$ be a topological space.  For all $L, L' \in \Plotkin^{cvx}
  (\Pred_{\Nature\;wk} (X))$, if $r_{\ADN} (L) \leq r_{\DN} (L')$ then
  $L \sqsubseteq_{\mathrm{EM}} L'$.
\end{lemma}
\begin{proof}
  By assumption, $r_{\AN} (L) \leq r_{\AN} (L')$, and $r_{\DN} (L)
  \leq r_{\DN} (L')$.  Using Lemma~\ref{lemma:ADN:r},
  Lemma~\ref{lemma:AN:orderembed} and Lemma~\ref{lemma:DN:orderembed},
  we obtain $cl (L) \subseteq cl (L')$ and $\upc L \supseteq \upc L'$.
\end{proof}

\begin{proposition}
  \label{prop:ADN:iso:gnrl}
  Let $X$ be a topological space such that $[X \to
  {\creal}_\sigma]_\sigma$ is locally convex and has an almost open
  addition map, for example, a stably locally compact space, or more
  generally, a core-compact, core-coherent space.

  Then $r_{\ADN}$ defines a homeomorphism between $\PV^{cvx}
  (\Pred_{\Nature\;wk} (X))$ and $\Pred_{\ADN\;wk} (X)$, and an
  order-isomorphism between $\Plotkin^{cvx} (\Pred_{\Nature\;wk} (X))$
  and $\Pred_{\ADN} (X)$.
\end{proposition}
\begin{proof}
  Same argument as in Proposition~\ref{prop:AN:iso:gnrl}, using
  Proposition~\ref{prop:ADN:weak} and
  Lemma~\ref{lemma:ADN:orderembed}.
\end{proof}

\subsection{Subnormalized and normalized previsions}
\label{sec:subn-norm-prev}

The subnormalized and normalized cases reduce to the unbounded case.
\begin{lemma}
  \label{lemma:AN:iso:aux}
  Let $\bullet$ be ``$\leq 1$'' or ``$1$''.  Let $X$ be a topological
  space.  For every non-empty closed subset $C$ of
  $\Pred^\bullet_{\Nature\;wk} (X)$, $r_{\AN} (C) = r_{\AN}
  (cl_\Nature (C))$, where $cl_\Nature (C)$ denotes the closure of $C$
  in $\Pred_{\Nature\;wk} (X)$.
\end{lemma}
\begin{proof}
  Since $C \subseteq cl_\Nature (C)$, for every $h \in [X \to
  {\creal}_\sigma]$, $r_{\AN} (C) (h) = \sup_{G \in C} G (h) \leq
  r_{\AN} (cl_\Nature (C)) (h)$.  Assume by contradiction that
  $r_{\AN} (C) (h) < r_{\AN} (cl_\Nature (C)) (h)$ for some $h$.  Let
  $a$ be a real number such that $r_{\AN} (C) (h) < a < r_{\AN}
  (cl_\Nature (C)) (h)$.  So $cl_\Nature (C)$ is in $r_{\AN}^{-1} [h >
  a]$, which is open in $\HV (\Pred_{\Nature\;wk} (X))$, since
  $r_{\AN}$ is continuous (Lemma~\ref{lemma:AN:r:cont}).  Every open
  subset of $\HV (\Pred_{\Nature\;wk} (X))$ is also
  Scott open, and $cl_\Nature (C)$ is the least upper bound in $\Hoare
  (\Pred_{\Nature\;wk} (X))$ of the directed family of closed
  subsets of the form $\dc_\Nature E$, where $E$ ranges over the
  finite subsets of $C$.  (We write $\dc_\Nature$ for downward closure
  in $\Pred_{\Nature\;wk} (X)$.)  So there is a finite subset
  $E$ of $C$ such that $\dc_\Nature E \in r_{\AN}^{-1} [h > a]$, i.e.,
  such that $\sup_{G \in E} G (h) > a$.  It follows that there is a $G
  \in E$ (hence $G \in C$) such that $G (h) > a$.  But $r_{\AN} (C)
  (h) = \sup_{G \in C} G (h) < a$, contradiction.  So $r_{\AN} (C) (h)
  = r_{\AN} (cl_\Nature (C)) (h)$ for every $h$.
\end{proof}

\begin{lemma}
  \label{lemma:AN:orderembed:1}
  Let $\bullet$ be ``$\leq 1$'' or ``$1$''.  Let $X$ be a topological
  space.  For all $C, C' \in \Hoare^{cvx} (\Pred^\bullet_{\Nature\;wk}
  (X))$, if $r_{\AN} (C) \leq r_{\AN} (C')$ then $C \supseteq C'$.
\end{lemma}
\begin{proof}
  Assume $r_{\AN} (C) \leq r_{\AN} (C')$.  By
  Lemma~\ref{lemma:AN:iso:aux},
  $r_{\AN} (cl_\Nature (C)) \leq r_{\AN} (cl_\Nature (C'))$.  Now
  $cl_\Nature (C)$ and $cl_\Nature (C')$ are closed and non-empty.
  They are convex, because the closure of any convex subset of a
  semitopological cone is convex again
  \cite[Lemma~4.10~(a)]{Keimel:topcones2}.  So
  Lemma~\ref{lemma:AN:orderembed} applies:
  $cl_\Nature (C) \subseteq cl_\Nature (C')$.  In particular,
  $C \subseteq cl_\Nature (C')$.  Since $C'$ is closed in the subspace
  $Y \eqdef \Pred^\bullet_{\Nature\;wk} (X)$ of
  $Z \eqdef \Pred_{\Nature\;wk} (X)$, one can write it as $Y \cap A$
  for some closed subset $A$ of $Z$.  Clearly
  $cl_\Nature (C') \subseteq A$, so
  $C' = Y \cap A \supseteq Y \cap cl_\Nature (C')$.  Since
  $C \subseteq cl_\Nature (C')$ and $C \subseteq Y$, it follows that
  $C \subseteq C'$.
\end{proof}

\begin{theorem}[Isomorphism, Angelic Cases]
  \label{thm:AN:iso}
  Let $\bullet$ be the empty superscript, ``$\leq 1$'' or ``$1$''.
  Let {\updated $X$ be an $\AN_\bullet$-friendly space}, for example,
  a locally compact space (or, more generally, a core-compact space),
  {\updated or an LCS-complete space}.

  Then $r_{\AN}$ defines a homeomorphism with inverse $s^\bullet_\AN$
  between $\HV^{cvx} (\Pred^\bullet_{\Nature\;wk} (X))$ and
  $\Pred^\bullet_{\AN\;wk} (X)$ and an order-isomorphism between
  $\Hoare^{cvx} (\Pred^\bullet_{\Nature\;wk} (X))$ and
  $\Pred^\bullet_{\AN} (X)$.
\end{theorem}
\begin{proof}
  When $\bullet$ is the empty superscript, this is
  Proposition~\ref{prop:AN:iso:gnrl}.  Otherwise, we use the same
  argument as in its proof, replacing Lemma~\ref{lemma:AN:orderembed}
  by Lemma~\ref{lemma:AN:orderembed:1}.  I.e., $r_{\AN}$ is an
  injective retraction, hence an isomorphism.
\end{proof}

\begin{corollary}
  \label{corl:AN:iso}
  Let $\bullet$ be the empty superscript, ``$\leq 1$'' or ``$1$''.
  Let $X$ be a continuous dcpo.  Then $r_{\AN}$ define an isomorphism
  with inverse $s^\bullet_{\AN}$ between $\Hoare^{cvx}
  (\Pred^\bullet_{\Nature} (X))$ and $\Pred^\bullet_{\AN} (X)$.
\end{corollary}

The demonic cases are slightly simpler.
\begin{lemma}
  \label{lemma:DN:iso:aux}
  Let $\bullet$ be ``$\leq 1$'' or ``$1$''.  Let $X$ be a topological
  space.  For every non-empty compact saturated subset $Q$ of
  $\Pred^\bullet_{\Nature\;wk} (X)$, $r_{\DN} (Q) = r_{\DN}
  (\upc_\Nature Q)$, where $\upc_\Nature Q$ denotes the upward closure
  of $Q$ in $\Pred_{\Nature\;wk} (X)$.
\end{lemma}
Note that $\upc_\Nature Q$ is compact saturated in
$\Pred_{\Nature\;wk} (X)$.  Indeed $Q$ is compact in
$\Pred^\bullet_{\Nature\;wk} (X)$, which is a subspace of
$\Pred_{\Nature\;wk} (X)$.  In particular Lemma~\ref{lemma:rDN}
applies, and $r_{\DN} (\upc_\Nature Q) (h) = \min_{G \in \upc_\Nature
  Q} G (h)$ for every $h$.

\begin{proof}
  Since $Q \subseteq \upc_\Nature Q$, $r_{\DN} (\upc_\Nature Q) (h)
  \leq r_{\DN} (Q) (h)$.  Conversely, for every $G \in \upc_\Nature
  Q$, there is a $G_1 \in Q$ such that $G _1 \leq G$, so $r_{\DN}
  (\upc_\Nature Q) (h) = \min_{G \in \Pred_{\Nature\;wk} (X), G_1 \in
    Q, G_1 \leq G} G (h) \geq \min_{G_1 \in Q} G_1 (h) = r_{\DN} (Q)
  (h)$.
\end{proof}

\begin{lemma}
  \label{lemma:DN:orderembed:1}
  Let $\bullet$ be ``$\leq 1$'' or ``$1$''.  Let $X$ be a topological
  space.  For all $Q, Q' \in \Smyth^{cvx} (\Pred^\bullet_{\Nature\;wk}
  (X))$, if $r_{\DN} (Q) \leq r_{\DN} (Q')$ then $Q \supseteq Q'$.
\end{lemma}
\begin{proof}
  Assume $r_{\DN} (Q) \leq r_{\DN} (Q')$.  By
  Lemma~\ref{lemma:DN:iso:aux}, $r_{\DN} (\upc_\Nature Q) \leq r_{\DN}
  (\upc_\Nature Q')$.  Now $\upc_\Nature Q$ is compact saturated,
  non-empty, and it is easy to see that it is convex, since $Q$ is
  convex.  We can therefore apply Lemma~\ref{lemma:DN:orderembed} and
  conclude that $\upc_\Nature Q \supseteq \upc_\Nature Q'$.  In
  particular, $\upc_\Nature Q \supseteq Q'$.  So, for every element
  $G'$ of $Q'$, there is a $G \in Q$ such that $G \leq G'$.  But $G'$
  is in $\Pred^\bullet_{\Nature\;wk} (X)$ since in $Q'$, and $Q$ is
  upward closed in $\Pred^\bullet_{\Nature\;wk} (X)$, so $G'$ itself
  is in $Q$.  It follows that $Q \supseteq Q'$.
\end{proof}

\begin{theorem}[Isomorphism, Demonic Cases]
  \label{thm:DN:iso}
  Let $\bullet$ be the empty superscript, ``$\leq 1$'' or ``$1$''.
  Let $X$ be a topological space.  Then $r_{\DN}$ defines a
  homeomorphism with inverse $s^\bullet_\DN$ between $\SV^{cvx}
  (\Pred^\bullet_{\Nature\;wk} (X))$ and $\Pred^\bullet_{\DN\;wk}
  (X)$, and an order-iso\-morphism between $\Smyth^{cvx}
  (\Pred^\bullet_{\Nature\;wk} (X))$ and $\Pred^\bullet_{\DN} (X)$.
\end{theorem}
\begin{proof}
  When $\bullet$ is the empty superscript, this is
  Proposition~\ref{prop:DN:iso:gnrl}.  Otherwise, we use the same
  argument as in its proof, replacing Lemma~\ref{lemma:DN:orderembed}
  by Lemma~\ref{lemma:DN:orderembed:1}.  I.e., $r_{\DN}$ is an
  injective retraction, hence an homeomorphism.
\end{proof}

\begin{corollary}
  \label{corl:DN:iso}
  Let $\bullet$ be the empty superscript, ``$\leq 1$'' or ``$1$''.
  Let $X$ be a continuous dcpo.  Then $r_{\DN}$ defines an isomorphism
  with inverse $s^\bullet_\DN$ between $\Smyth^{cvx}
  (\Pred^\bullet_{\Nature} (X))$ and $\Pred^\bullet_{\DN} (X)$.
\end{corollary}

There is nothing left to do to conclude in the erratic cases.
\begin{theorem}[Isomorphism, Erratic Cases]
  \label{thm:ADN:iso}
  Let $\bullet$ be ``$\leq 1$'' or ``$1$''.
  Let $X$ be a topological space such that $[X \to
  {\creal}_\sigma]_\sigma$ is locally convex and has an almost open
  addition map, for example a stably locally compact space, or more
  generally, a core-compact, core-coherent space.  Assume also that
  $X$ is compact in case $\bullet$ is ``$1$''.

  Then $r_{\ADN}$ defines an homeomorphism with inverse
  $s^\bullet_{\ADN}$ between $\PV^{cvx} (\Pred^\bullet_{\Nature\;wk}
  (X))$ and $\Pred^\bullet_{\ADN\;wk} (X)$, and an order-isomorphism
  between $\Plotkin^{cvx} (\Pred^\bullet_{\Nature\;wk} (X))$ and
  $\Pred^\bullet_{\ADN} (X)$.
\end{theorem}
\begin{proof}
  The only thing left to prove is that $s^\bullet_{\ADN} (r_{\ADN}
  (L)) = L$ for every $L \in \Plotkin^{cvx}
  (\Pred^\bullet_{\Nature\;wk} (X))$.  But $s^\bullet_{\ADN} (r_{\ADN}
  (L)) = s^\bullet_{\ADN} (r_{\DN} (L), r_{\AN} (L)) =
  s^\bullet_{\ADN} (r_{\DN} (\upc L), r_{\AN} (cl (L)))$ (by
  Lemma~\ref{lemma:ADN:r}) $= s^\bullet_{\DN} (r_{\DN} (\upc L)) \cap
  s^\bullet_{\AN} (r_{\AN} (cl (L))) = \upc L \cap cl (L)$ (by
  Theorem~\ref{thm:DN:iso} and Theorem~\ref{thm:AN:iso}) $= L$.
\end{proof}

\begin{corollary}
  \label{corl:ADN:iso}
  Let $\bullet$ be the empty superscript, ``$\leq 1$'' or ``$1$''.
  Let $X$ be a coherent, continuous dcpo (and pointed if $\bullet$ is
  ``$1$'').  Then $r_{\ADN}$ defines an isomorphism with inverse
  $s^\bullet_{\ADN}$ between $\Plotkin^{cvx} (\Pred^\bullet_{\Nature}
  (X))$ and $\Pred^\bullet_{\ADN} (X)$.
\end{corollary}

\subsection{Hulls}
\label{sec:hulls}

When $A$ is not convex, $s^\bullet_{\AN} (r_{\AN} (A))$ cannot be
equal to $A$, since the former is always convex.  Similarly in the
demonic cases.  It is a natural question to ask what the compositions
$s^\bullet_{\AN} \circ r_{\AN}$, and similar compositions, compute.

The \emph{convex hull} $conv (A)$ of a set $A$ in a topological cone
(or in a convex subspace) is the smallest convex set that contains
$A$.  This is also the set of linear combinations of the form
$ax+(1-a)y$, where $x, y \in A$, $a \in [0, 1]$.

The \emph{closed convex hull} of a set $A$ in a semitopological cone
is the smallest closed and convex set containing $A$.  This is the
closure of the convex hull of $A$
\cite[Lemma~4.10~(a)]{Keimel:topcones2}.
\begin{proposition}
  \label{prop:AN:sr:hull}
  Let $\bullet$ be the empty superscript, ``$\leq 1$'', or ``$1$''.
  Let {\updated $X$ be an $\AN_\bullet$-friendly space}, for example,
  a locally compact space (or, more generally, a core-compact space),
  {\updated or an LCS-complete space}.

  For every non-empty closed subset $A$ of
  $\Pred^\bullet_{\Nature\;wk} (X)$, $s^\bullet_{\AN} (r_{\AN} (A))$
  is the closed convex hull $cl (conv (A))$ of $A$ in
  $\Pred^\bullet_{\Nature\;wk} (X)$.
\end{proposition}
\begin{proof}
  By Proposition~\ref{prop:AN:weak}, $s^\bullet_{\AN} (r_{\AN} (A))$
  is non-empty and closed.  It is convex, and clearly contains $A$.
  Conversely, for every closed convex subset $A'$ that contains $A$,
  $s^\bullet_{\AN} (r_{\AN} (A)) \subseteq s^\bullet_{\AN} (r_{\AN}
  (A')) = A'$, where the last equality follows from
  Theorem~\ref{thm:AN:iso}.  Therefore $s^\bullet_{\AN} (r_{\AN} (A))$
  is the smallest closed convex set containing $A$, namely $cl (conv (A))$.
\end{proof}

Symmetrically, one can define the \emph{compact saturated convex hull}
of a set $A$ as the smallest compact saturated, convex subset
containing $A$, if it exists.  The following shows that this exists if
$A$ is any non-empty compact saturated subset of
$\Pred^\bullet_{\Nature\;wk} (X)$, in particular.

\begin{proposition}
  \label{prop:DN:sr:hull}
  Let $X$ be a topological space.   Let $\bullet$ be the empty
  superscript, ``$\leq 1$'', or ``$1$''.

  For every non-empty compact saturated subset $Q$ of
  $\Pred^\bullet_{\Nature\;wk} (X)$, $s^\bullet_{\DN} (r_{\DN} (Q))$
  is the compact saturated convex hull of $Q$ in
  $\Pred^\bullet_{\Nature\;wk} (X)$.
\end{proposition}
\begin{proof}
  By Proposition~\ref{prop:DN:weak}, $s^\bullet_{\DN} (r_{\DN} (Q))$
  is compact saturated.  It is convex, and clearly contains $Q$.  If
  $Q'$ is any convex, compact saturated superset of $Q$, then
  $s^\bullet_{\DN} (r_{\DN} (Q'))$ contains $s^\bullet_{\DN} (r_{\DN}
  (Q))$.  However, $s^\bullet_{\DN} (r_{\DN} (Q')) = Q'$ since $r_\DN$
  is an isomorphism with inverse $s^\bullet_\DN$ on non-empty convex
  compact saturated sets, by Theorem~\ref{thm:DN:iso}; so $Q'$
  contains $s^\bullet_{\DN} (r_{\DN} (Q))$.
\end{proof}

\begin{proposition}
  \label{prop:ADN:sr:hull}
  Let $\bullet$ be the empty superscript, ``$\leq 1$'', or ``$1$''.
  Let {\updated $X$ be an $\AN_\bullet$-friendly space}, for example,
  a locally compact space (or, more generally, a core-compact space),
  {\updated or an LCS-complete space}.

  For every lens $L$ of $\Pred^\bullet_{\Nature\;wk} (X)$,
  $s^\bullet_{\ADN} (r_{\ADN} (L))$ is a lens, and is the smallest
  convex lens that contains $L$.
\end{proposition}
\begin{proof}
  We know that $s^\bullet_{\ADN} (r_{\ADN} (L)) = s^\bullet_{\AN}
  (F^+) \cap s^\bullet_{\DN} (F^-)$ where $(F^+, F^-) = r_{\ADN} (L)$.
  By Lemma~\ref{lemma:ADN:r}, $F^+ = r_{\AN} (cl (L))$ and $F^- =
  r_{\DN} (\upc L)$.  By Proposition~\ref{prop:AN:sr:hull},
  $s^\bullet_\AN (F^+)$ is the closed convex hull of $cl (L)$, hence
  also the smallest closed convex subset that contains $L$.  By
  Proposition~\ref{prop:DN:sr:hull}, $s^\bullet_\DN (F^-)$ is the
  compact saturated convex hull of $\upc L$, hence also the smallest
  compact saturated convex subset that contains $L$.  As a
  consequence, their intersection is a convex lens.

  For the final part of the Proposition, if $L'$ is any convex lens
  that contains $L$, then $cl (L')$ contains $L$, is closed, and is
  convex \cite[Lemma~4.10~(a)]{Keimel:topcones2}, so $cl (L')$
  contains $s^\bullet_\AN (F^+)$; also, $\upc L'$ contains $L$, is
  compact saturated and convex, hence contains $s^\bullet_\DN (F^-)$.
  It follows that $L' = cl (L') \cap \upc L'$ contains $s^\bullet_\AN
  (F^+) \cap s^\bullet_\DN (F^-) = s^\bullet_\ADN (r_\ADN (L))$.
\end{proof}

\section{Conclusion}
\label{sec:conclusion}

We have proved that, under some natural conditions, the powercone
models and the prevision models of mixed non-deterministic and
probabilistic choice coincide.  This involved functional analytic
methods that heavily rely on Keimel's cone-theoretic variants of the
classical Hahn-Banach separation theorems, plus the Schr\"oder-Simpson
Theorem.

The demonic cases are the nicest, and require absolutely no assumption
on the underlying topological space $X$.  As should be expected, the
erratic cases demand stronger assumptions.  An intermediate case is
given by the angelic cases, which require $[X \to
{\creal}_\sigma]_\sigma$ to be locally convex.  We know that this is the
case when $X$ is locally compact, and more generally, core-compact,
and we do not know of an example of a space $X$ for which $[X \to
{\creal}_\sigma]_\sigma$ would fail to be locally convex.  Can we
characterize those spaces $X$ for which $\Lform X$ 
is locally convex?

\textbf{Acknowledgments.}  I thank the anonymous referees for many
useful comments.  More than anything, I must thank them for their
urging me to try to do away with an assumption of coherence that I had
made in some results concerning the Hoare cases, which I had found
surprising, but could not do away with initially.  That assumption has
now disappeared, after some more hard work.  I also thank them, as
well as the editor, Prakash Panangaden, for their patience.  I did not
have time to work on this paper for more than two years, between March
2013 and June 2015, and they are to be commended for having waited for
so long: my sincere apologies to them.

{\updated%
  
\section{Review of papers citing \cite{JGL-mscs16}}
\label{sec:review-papers-citing}

Since the present paper fixes a mistake in \cite{JGL-mscs16}, we
review whether and how this mistake propagates in the literature.  To
jump to the conclusion, let me say right away that no published paper
is affected by the mistake as of November 1st, 2024.

As of November 1st, 2024, Google Scholar reports on 17 papers citing
\cite{JGL-mscs16}.  Most of them only cite \cite{JGL-mscs16} without
using its results, except for the following.
\begin{itemize}
\item \cite{dBGLJL:LCS} uses Theorem~4.11 of \cite{JGL-mscs16} in its
  Corollary~13.7, which states that $r_\AN$ defines a homeomorphism
  with inverse $s^\bullet_\AN$ between
  $\HV^{cvx} (\Pred^\bullet_{\Nature\;wk} (X))$ and
  $\Pred^\bullet_{\AN\;wk} (X)$ when $X$ is LCS-complete.  Our
  corrected version (Theorem~\ref{thm:AN:iso} of the present version)
  confirms that this is indeed true.

  We have used a few results from \cite{dBGLJL:LCS} in the present
  revised version.  The danger is of course that of circular
  reasoning.  Those results are Proposition~12.1 of \cite{dBGLJL:LCS},
  which states that every LCS-complete space is consonant, and
  Lemma~13.2 of \cite{dBGLJL:LCS}, which states that every
  LCS-complete space is $\odot$-consonant.  Those are proofs that are
  entirely independent of any result of \cite{JGL-mscs16}.
\item \cite{JGL:fullabstr:I} uses Propositions~3.41, 3.42 and 3.44 of
  \cite{JGL-mscs16}, which are unaffected.
\item \cite{JGL:distval:III} collects several properties about $r_\AN$
  and $s^\bullet_\AN$ at the beginning of its Section~3, in order to
  use them without mention later on.  This includes the faulty
  Corollary~3.12 of \cite{JGL-mscs16}.  But no result of
  \cite{JGL:distval:III} is affected, since $r_\AN$ and
  $s^\bullet_\AN$ are only used on continuous complete quasi-metric
  spaces in their $d$-Scott topology, and all such spaces are
  LCS-complete \cite[Theorem~4.1]{dBGLJL:LCS}.
\item \cite{JGL:prev:qmet}, which is a superset of
  \cite{JGL:distval:III}, is similarly unaffected.
\item \cite{JGL:functproj} (versions 1 and 2) relies on the faulty
  results of \cite{JGL-mscs16} in its Section~13.  One has to replace
  the categories $\mathbf K$ by categories of $\AN_\bullet$-friendly
  spaces.  This is corrected in version~3 of \cite{JGL:functproj}.
  \cite{JGL:functproj} does not appear to be cited by any paper yet,
  according to Google Scholar.
\item \cite{effects:LOP18} uses \cite{JGL-mscs16} in its Lemma~22.
  This is unaffected by the mistake, since it only uses results on the
  demonic cases, namely, on $r_\DN$ and $s^\bullet_\DN$.
  Alternatively, the primary justification used in Lemma~22 is really
  \cite{KP:mixed}, which is unaffected by the mistake, for obvious
  reasons.
\item \cite{GLJ:Valg} uses only one result from \cite{JGL-mscs16},
  namely that $\Lform X = \Lform {(X^s)}$, where $X^s$ is the
  sobrification of $X$, in its Example~3.14; this is unaffected.
\item \cite{JGL:distval:IV} cites \cite{JGL-mscs16} at several places.
  Lemma~3.24 of \cite{JGL-mscs16} is used before Proposition~4.1; this
  result is unaffected by the mistake.  The proof of Proposition~4.1
  uses a variety of results from \cite{JGL-mscs16}, namely Lemma~3.28,
  Lemma~3.29, Proposition~4.5 and Proposition~4.3, none of which is
  affected by the mistake; it also references the faulty
  Corollary~3.12 of \cite{JGL-mscs16}, but only in order to inform the
  reader of the roots of the results of \cite{JGL:distval:III} used in
  the proof and, as we have seen, none of the results of
  \cite{JGL:distval:III} are affected.  At any rate, Proposition~4.1
  of \cite{JGL:distval:IV} requires not only $\Lform X$ to be locally
  convex (and to have an almost addition map) but also $X$ to be
  compact, and that is condition~1 of the definition of
  $\AN_1$-friendliness, under which all results of \cite{JGL-mscs16}
  hold.  Remark~4.2 of \cite{JGL:distval:IV} makes use of
  Proposition~4.8 of \cite{JGL-mscs16}, which is not affected.

  We have used one result from \cite{JGL:distval:I} in the present
  revised version.  This is the fact that a continuous valuation $\nu$
  on a space $Y$ is supported on a subset $X$ if and only if
  $\nu = i [\mu]$ for some (unique) continuous valuation $\mu$ on the
  subspace $Y$, where $i$ is the inclusion map
  \cite[Lemma~4.6]{JGL:distval:I}.  This does not rely on any result
  of \cite{JGL-mscs16}, avoiding the pitfall of circular reasoning,
  just like in the case of \cite{dBGLJL:LCS}.
  
\item \cite{JGL:V:loccomp} only uses results on $r_\DN$ and
  $s^\bullet_\DN$ (e.g., Lemma~3.20 of \cite{JGL-mscs16} before
  Lemma~3.2, and Proposition~3.22 and Theorem~4.15 in the conclusion),
  not $r_\AN$ and $s^\bullet_\AN$, and is unaffected by the mistake.
\end{itemize}
Finally, I know about the following paper citing \cite{JGL-mscs16},
but not reported by Google Scholar:
\begin{itemize}
\item \cite{GL:weak:distr} versions~1 through~3 are affected by the
  mistake, through the definition of the category $\Topcat^\flat$
  (Definition~5.6; Definition~5.7 in version~4), which must consist of
  $\AN_\bullet$-friendly spaces, instead of spaces $X$ such that
  $\Lform X$ is locally convex.  This is corrected in version~4.
  \cite{GL:weak:distr} does not appear to be cited by any paper yet,
  according to Google Scholar.
\end{itemize}

I would like to mention that \cite{GL:weak:distr} contains simpler
proofs of the continuity of $s^\bullet_\DN$
\cite[Lemma~4.6]{GL:weak:distr} and of $s^\bullet_\AN$
\cite[Lemma~6.3]{GL:weak:distr}.  I have not included them here,
although a revised version would be an opportunity to do so, because I
wanted to concentrate on fixing the mistake in Lemma~3.4 of
\cite{JGL-mscs16}.  Certain other proofs, such as that of
Lemma~\ref{lemma:ADN:intermezzo:2}, could also be made simpler, and I
have left them as is for the same reason.  Finally,
\cite{JGL:functproj} contains needed results on the natural of $r_\DN$
and $s^\bullet_\DN$ \cite[Lemma~11.2]{JGL:functproj}, of $r_\AN$ and
$s^\bullet_\DN$ \cite[Lemma~13.2]{JGL:functproj}, of $r_\ADN$ and
$s^\bullet_\ADN$ \cite[Lemma~15.1]{JGL:functproj}, which have been
left out of this revised version for similar reasons.


}

\ifarxiv

\else
\bibliography{iso}
\fi

\end{document}
